\documentclass[aps,pra,10pt,floatfix,superscriptaddress,longbibliography,showpacs,twocolumn]{revtex4-2}

\newif\ifapp
\newcommand{\lsp}{\hspace{0.1em}}

\usepackage{float}
\usepackage{makecell}
\usepackage{hyperref}
\usepackage{physics}
\usepackage{xfrac}
\usepackage{graphicx}
\usepackage{subfigure}
\usepackage{mathdots}
\usepackage{amsfonts,amssymb,amsmath}
\usepackage{mathrsfs}
\usepackage[]{graphics,graphicx,epsfig}
\usepackage{amsthm}
\usepackage{csquotes}
\MakeOuterQuote{"}
\usepackage{epstopdf}
\usepackage{tikz}
\usepackage{booktabs}
\usepackage{array}

\usepackage{algorithm}
\usepackage{algorithmicx}
\usepackage{algpseudocode}

\usepackage{paralist}
\usepackage{diagbox}
\usepackage[inline]{enumitem}

\usepackage{tabularx}
\usepackage{setspace}
\usepackage{mathtools}
\usepackage{dsfont}

\hypersetup{
	breaklinks=true,
	colorlinks=true,       % false: boxed links; true: colored links
	linkcolor=blue,          % color of internal links
	citecolor=red,        % color of links to bibliography
	filecolor=magenta,      % color of file links
	urlcolor=blue,           % color of external links
	runcolor=cyan
}

\newcommand{\equad}{\,\hphantom{=}\,}

\newcommand{\bfx}{\mathbf{x}}
\newcommand{\bfy}{\mathbf{y}}
\newcommand{\bbD}{\mathbb{D}}
\newcommand{\bbE}{\mathbb{E}}
\newcommand{\bbN}{\mathbb{N}}
\newcommand{\bbR}{\mathbb{R}}
\newcommand{\bbS}{\mathbb{S}}
\newcommand{\bbW}{\mathbb{W}}

\newcommand{\caC}{\mathcal{C}}
\newcommand{\caE}{\mathcal{E}}
\newcommand{\caH}{\mathcal{H}}
\newcommand{\caI}{\mathcal{I}}

\newcommand{\caL}{\mathcal{L}}
\newcommand{\caM}{\mathcal{M}}
\newcommand{\caN}{\mathcal{N}}
\newcommand{\caO}{\mathcal{O}}

\newcommand{\caR}{\mathcal{R}}
\newcommand{\caS}{\mathcal{S}}
\newcommand{\caT}{\mathcal{T}}
\newcommand{\caU}{\mathcal{U}}
\newcommand{\caV}{\mathcal{V}}
\newcommand{\caW}{\mathcal{W}}

\newcommand{\rmA}{\mathrm{A}}
\newcommand{\rmB}{\mathrm{B}}
\newcommand{\rmC}{\mathrm{C}}
\newcommand{\rmH}{\mathrm{H}}
\newcommand{\rmU}{\mathrm{U}}
\newcommand{\rmW}{\mathrm{W}}
\newcommand{\rmCM}{\mathrm{CM}}
\newcommand{\rmRM}{\mathrm{RM}}

\newcommand{\SN}{\mathrm{SN}}
\newcommand{\blambda}{\boldsymbol{\lambda}}
\newcommand{\tblambda}{\tilde{\boldsymbol{\lambda}}}

\newcommand{\Tomega}{\widetilde{\Omega}}

\newcommand{\mb}[1]{\mathbb{#1}}
\newcommand{\te}[1]{\text{#1}}

\newcommand{\sn}{\mathrm{SN}}

\newcommand{\hr}{\hat{\rho}}
\newcommand{\ha}{\hat{a}}
\newcommand{\hatt}{\hat{t}}
\newcommand{\Haar}{\mathrm{Haar}}

\newcommand{\UF}{\mathrm{UF}}
\newcommand{\diag}{\operatorname{diag}}

\def\<{\langle}  %% overriding the original command \<
\def\>{\rangle}  %% overriding the original command \>

\newtheorem{lemma}{Lemma}
\newtheorem{proposition}{Proposition}

\newtheorem{corollary}{Corollary}
\newtheorem{remark}{Remark}
\newcommand{\lref}[1]{Lemma~\ref{#1}}
\newcommand{\lsref}[1]{Lemmas~\ref{#1}}
\newcommand{\Lref}[1]{Lemma~\ref{#1}}

\newcommand{\thref}[1]{Theorem~\ref{#1}}
\newcommand{\thsref}[1]{Theorems~\ref{#1}}
\newcommand{\Thref}[1]{Theorem~\ref{#1}}

\newcommand{\pref}[1]{Proposition~\ref{#1}}
\newcommand{\psref}[1]{Propositions~\ref{#1}}

\newcommand{\Psref}[1]{Propositions~\ref{#1}}
\newcommand{\coref}[1]{Corollary~\ref{#1}}

\newcommand{\eref}[1]{Eq.~\textup{(\ref{#1})}}
\newcommand{\eqsref}[2]{Eqs.~(\ref{#1}) and (\ref{#2})}
\newcommand{\Eref}[1]{Equation~\textup{(\ref{#1})}}

\newcommand{\tref}[1]{Table~\ref{#1}}
\newcommand{\Tref}[1]{Table~\ref{#1}}

\newcommand{\Tsref}[1]{Tables~\ref{#1}}
\newcommand{\sref}[1]{Sec.~\ref{#1}}

\newcommand{\fref}[1]{Fig.~\ref{#1}}
\newcommand{\Fref}[1]{Figure~\ref{#1}}

\newcommand{\Fsref}[1]{Figures~\ref{#1}}
\newcommand{\aref}[1]{Appendix~\ref{#1}}
\newcommand{\asref}[1]{Appendices~\ref{#1}}

\theoremstyle{plain} 
\theoremstyle{plain} \newtheorem{theo}{Theorem}
\theoremstyle{plain} 
\theoremstyle{plain}

\begin{document}
	\title{Certifying entanglement dimensionality by \texorpdfstring{$k$}{\empty}-reduction moments}
	\author{Changhao Yi} 
	\email{cyi@shu.edu.cn}
	\thanks{These authors contributed equally to this work.}
	\affiliation{Department of Physics, Shanghai University, Shanghai 200444, China}
	\affiliation{State Key Laboratory of Surface Physics, Department of Physics, and Center for Field Theory and Particle Physics, Fudan University, Shanghai 200433, China}
	\affiliation{Institute for Nanoelectronic Devices and Quantum Computing, Fudan University, Shanghai 200433, China}
	\affiliation{Shanghai Research Center for Quantum Sciences, Shanghai 201315, China}
	
	\author{Xiaodi Li} 
	\email{lixiaodi@fudan.edu.cn}
	\thanks{These authors contributed equally to this work.}
	\affiliation{iFLYTEK Research, Hefei 230088, China}
	\affiliation{State Key Laboratory of Surface Physics, Department of Physics, and Center for Field Theory and Particle Physics, Fudan University, Shanghai 200433, China}
	\affiliation{Institute for Nanoelectronic Devices and Quantum Computing, Fudan University, Shanghai 200433, China}
	\affiliation{Shanghai Research Center for Quantum Sciences, Shanghai 201315, China}
	\author{Huangjun Zhu}
	\affiliation{State Key Laboratory of Surface Physics, Department of Physics, and Center for Field Theory and Particle Physics, Fudan University, Shanghai 200433, China}
	\affiliation{Institute for Nanoelectronic Devices and Quantum Computing, Fudan University, Shanghai 200433, China}
	\affiliation{Shanghai Research Center for Quantum Sciences, Shanghai 201315, China}
	\date{\today}

	\begin{abstract}
		In this paper, we combine the $k$-reduction map, the moment method, and randomized measurements into a practical protocol for certifying  the entanglement dimensionality. Our approach is based on the observation that a state with entanglement dimensionality at most $k$ must stay positive under the action of the $k$-reduction map. The core of our protocol utilizes the moment method to determine whether the $k$-reduced operator, i.e., the operator obtained after applying the $k$-reduction map on a quantum state, contains a negative eigenvalue or not. Notably, we propose a systematic method for constructing reduction moment criteria that apply  to a broader range of states, including $k$-unfaithful states, compared with fidelity-based methods. The performance of our approach gets better and better with the moment order employed, which is corroborated by extensive numerical simulations.
		To apply our approach, it suffices to implement a unitary 3-design instead of a 4-design,  which is more feasible in practice than the correlation matrix method.  In the course of study, we show that the $k$-reduction negativity, the absolute sum of the negative eigenvalues of the  $k$-reduced operator,  is monotonic under local operations and classical communication for pure states.
	\end{abstract} 
	
	\maketitle
	
	\tableofcontents
	
	\section{Introduction}
	
	Entanglement is a characteristic feature of quantum mechanics and a vital resource in quantum information processing \cite{horodecki2009quantum}. Consequently, certifying entanglement is crucial for many practical applications and has been a cornerstone of quantum information science since its early days \cite{guhne2009entanglement,friis2019entanglement}.  
	Over the past three decades, numerous entanglement criteria have emerged. Among these, the \emph{positive partial transpose} (PPT) criterion stands out due to its simplicity and powerful detection capability \cite{Pere96,Horo01,aubrun2012phase}.  However, the PPT criterion demands complete knowledge of the density matrix, which is rarely available. To remedy this problem, more practical approaches based on partial transpose moments and \emph{randomized measurements} \cite{elben2023randomized,ciesidg2024analysing} have recently been developed \cite{elben2020mixed,neven2021symmetry,YuIG21,elben2023randomized}. 
	Similar approaches can also estimate R{\'e}nyi entanglement entropy \cite{brydges2019probing}, entanglement entropy \cite{huan2020}, and entanglement negativity \cite{zhou2020single}.
	
	For many quantum information processing tasks, including quantum computation and quantum simulation, the mere presence of entanglement is usually insufficient. Instead, specific forms of entanglement possessing desired properties, such as high-dimensional entanglement, are often required to fully benchmark the capabilities of a quantum platform. In this context, high-dimensional entanglement is typically quantified by the \emph{Schmidt number}~\cite{terhal2000schmidt}. With the rapid advancement of quantum technologies, preparing sophisticated entangled states has evolved from theoretical imagination to experimental reality. Consequently, the efficient certification of high-dimensional entanglement has become a pressing challenge \cite{huang2016high}, one that remains difficult even when the density matrix is fully known.
	
	\textit{Schmidt-number witnesses} \cite{terhal2000schmidt,bavaresco2018measurements}, analogous to entanglement witnesses, are popular tools for certifying the entanglement dimensionality. Many Schmidt-number witnesses are based on fidelities with certain target states, leading to \emph{fidelity-based criteria}. A state is called \textit{$k$-unfaithful} \cite{weilenmann2020entanglement} if it cannot be certified by such Schmidt-number-$k$ witnesses; otherwise, it is \textit{$k$-faithful}. So far most laboratory protocols rely on fidelity-based criteria \cite{huang2016high,bavaresco2018measurements,morelli2023resource,li2025high}, which necessitate significant prior information about the prepared states and require the target state itself be $k$-faithful. These drawbacks impose severe limitations on practical applications and call for more powerful alternatives.
	
	An ideal method for certifying high-dimensional entanglement should be invariant under local unitary transformations and require a measurement cost no greater than that of full state tomography.
	Recently, the \emph{correlation matrix} (CM) method was  proposed  to address this challenge \cite{imai2021bound,liu2023characterizing,wyderka2023probing,liu2024bounding}. However, estimating the moments of the CM typically requires randomized measurements based on 4-designs. Unfortunately, 4-designs cannot be realized by discrete groups except for one special case, and most 4-designs known in the literature are not amenable to experimental realization \cite{gross2007evenly,ZhuKGG16,bannai2020unitary}. Although a CM-based protocol has been experimentally realized \cite{lib2024experimental}, its application in higher dimensions remains challenging.
	
	In this work, we propose a simple yet powerful approach for certifying the entanglement dimensionality that is more amenable to experimental realization. Our approach leverages the observation that the \textit{$k$-reduction map} \cite{terhal2000schmidt,sanpera2001schmidt}, a typical positive map, preserves the positivity of any density operator with Schmidt number up to $k$. A violation of this positivity condition implies that the Schmidt number of the target state is at least $k+1$. By combining the $k$-reduction map and the \textit{moment method} \cite{schmudgen2017moment,schmudgen2020lectures}, we develop a practical approach for certifying  the entanglement dimensionality by virtue of the $k$-reduction moments, the moments of the  $k$-reduced operator. These moments can be  efficiently estimated using  randomized measurements based on 3-designs, such as the \textit{Clifford group} \cite{webb2015clifford,zhu2017multiqubit}, 
	which is much more practical than previous approaches relying on 4-designs. 
	Furthermore, our approach applies to both $k$-faithful and $k$-unfaithful states, offering a significantly broader scope of applications compared with fidelity-based methods.
	
	When the first $(2D-1)$  $k$-reduction moments  are available for quantum states on a $D$-dimensional Hilbert space, our criterion is equivalent to the $k$-reduction criterion and thus can certify the Schmidt numbers of all pure states. Moreover, the Schmidt numbers of \textit{isotropic states} \cite{terhal2000schmidt} can be certified by the first three $k$-reduction moments. The superior performance of our approach is supported by rigorous theoretical analysis and extensive numerical experiments. When collective measurements are available, the $k$-reduction moments featured in our approach can be estimated much more efficiently with constant sample complexity.
	
	The rest of this paper is organized as follows. In  \sref{sec:preliminaries}, we review the definition of the Schmidt number and give an overview of the CM method. In \sref{sec:k_spectrum}, we
	clarify the properties of the $k$-reduction map, $k$-reduction negativity, and $k$-reduction criterion in comparison with the CM criterion. In \sref{sec:moment}, we review the basic idea of the moment method and propose our certification criteria using $k$-reduction moments.  In \sref{sec:performance}, we demonstrate the performance of our certification criteria through numerical simulations. In \sref{sec:moment_estimation}, 
	we present two methods  for estimating reduction moments and discuss the sample complexities. Finally, we conclude in \sref{sec:conclusion}. A summary of the main results is presented in Table \ref{tab:main}. Technical proofs are relegated to the Appendices.
	
	\begin{table*}[!htbp]
		\centering
		\caption{\label{tab:main}Comparison between different Schmidt number certification protocols for quantum states with total dimension $D$. The parameter $\chi$ in full state tomography means the algorithm is only guaranteed to work when the state has rank no larger than $\chi$. This parameter is irrelevant to the Schmidt number of the state.}
		\renewcommand\arraystretch{2.5}
		\begin{tabular}{ c  c  c  c  c }
			\hline
			\hline
			Method & \makecell{Require knowledge \\ about the target state} & \makecell{Order of \\ unitary designs} & \makecell{Sample \\ complexity}  \\
			\hline
			Fidelity-based \cite{bavaresco2018measurements} & Yes & / & $\caO(1)$  \\ 
			Correlation matrix \cite{liu2023characterizing,wyderka2023probing} & No & $\ge 4$ on $\caH_\rmA,\caH_\rmB$ & Unknown  \\ 
			\makecell{Full state tomography \\ of rank-$\chi$ states \cite{kueng2017low}} & No & 4 on $\caH_{\rmA\rmB}$ & $\caO(D^2\chi)$ \\
			\makecell{$k$-reduction moments \\ (This work)} & No & $\ge 3$ on $\caH_{\rmA\rmB}$ & \makecell{$\caO(D)$ \\ \aref{app:haar_random}}  \\
			\hline
			\hline
		\end{tabular}
	\end{table*}
	
	\section{\label{sec:preliminaries}Background}
	Given a finite-dimensional Hilbert space $\caH$, denote by $\caL(\caH)$, 
	$\caL^\rmH(\caH)$, and  $\caL_0^\rmH(\caH)$ the set of linear operators, the set of Hermitian operators, and the set of traceless Hermitian operators on $\caH$, respectively. The spectrum (with  multiplicities taken into account) and the minimum eigenvalue of a Hermitian operator $O$ are denoted by $\sigma(O)$ and $\sigma_{\min}(O)$, respectively. 
	
	\subsection{Schmidt number and Schmidt-number witnesses}
	Consider a bipartite system with  Hilbert space $\caH_{\rmA\rmB} = \caH_\rmA\otimes\caH_\rmB$. Let $d_\rmA\equiv \dim(\caH_\rmA)$, $d_\rmB\equiv\dim(\caH_\rmB)$, $D\equiv d_\rmA d_\rmB$, and $d\equiv \min\{d_\rmA,d_\rmB\}$. Without loss of generality, we assume $d_\rmA\leq d_\rmB$ throughout this paper, which means $d=d_\rmA$. 
	The \emph{Schmidt number} (also known as Schmidt rank) of a pure state  $|\psi\rangle$ in $\caH_{\rmA\rmB}$ is defined as
	\begin{equation}
		\SN(|\psi\rangle) \equiv \mathrm{rank}(\rho_{\rmA}),
	\end{equation}
	where $\rho_{\rmA}=\Tr_{\rmB}(|\psi\rangle\langle\psi|)$. 
	Note that $\SN(|\psi\rangle)=1$ iff $|\psi\>$ is a product state. The \emph{Schmidt spectrum} of $|\psi\rangle$, denoted by $\sigma(\psi)$,  is the set of eigenvalues of $\rho_{\rmA}$ (with multiplicities taken into account)   \cite{nielsen2010quantum}, that is, $\sigma(\psi)\equiv\sigma(\rho_{\rmA})$. 
	
	We assume $\{|i\>_{\rmA}\},\{|i\>_{\rmB}\}$ are the computational bases on $\caH_\rmA,\caH_\rmB$ until otherwise specified. If  $|\psi\rangle$ has Schmidt decomposition
	\begin{equation}
		|\psi\rangle = \sum_{i=0}^{d-1}\sqrt{\lambda_i}|i\>_{\rmA}\otimes|i\>_{\rmB},
	\end{equation}
	then its Schmidt spectrum reads $\sigma(\psi)=\{\lambda_i\}_{i=0}^{d-1}$. In addition, the vector $\blambda=(\lambda_0,\lambda_1,\ldots,\lambda_{d-1})$ is called the Schmidt vector of $|\psi\>$. 
	This  vector  can be regarded as
	a point on the $(d-1)$-dimensional probability simplex
	\begin{equation}\label{eq:simplex}
		\!\! \Delta_{d} \equiv \left\{(x_0,x_1,\ldots,x_{d-1}) : \sum_{i=0}^{d-1} x_i = 1, x_i \ge 0\right\}.
	\end{equation}
	If $\lambda_0=\cdots =\lambda_{r-1}=1/r$ and $\lambda_r=\cdots =\lambda_{d-1}=0$, then we get the maximally entangled state with Schmidt number $r$,
	\begin{equation}\label{eq:rank_r_maximally_entangled}
		|+_r\rangle \equiv \frac{1}{\sqrt{r}}\sum_{i=0}^{r-1}|i\>_{\rmA}\otimes|i\>_{\rmB}.
	\end{equation}

	The concept of Schmidt number can be generalized to mixed states. Denote by $\caS_r$
	the set of quantum states on $\caH_{\rmA\rmB}$ that can be expressed as  convex combinations of pure states with Schmidt numbers no larger than $r$. By definition we have the following hierarchy relation:
	\begin{align}\label{eq:schmidt_hierarchy}
		\caS_1\subset \caS_2 \subset \cdots \subset \caS_{d} = \caS(\caH_{\rmA\rmB}),
	\end{align}
	where $\caS(\caH_{\rmA\rmB})$ denotes the set of all density operators on $\caH_{\rmA\rmB}$. The Schmidt number $\SN(\rho)$ of $\rho\in \caS(\caH_{\rmA\rmB})$ is defined as the smallest integer $r$ such that $\rho\in \caS_r$. Alternatively, $\SN(\rho)$ can be expressed as follows \cite{terhal2000schmidt}:
	\begin{equation}\label{eq:def_schmidt_number_2}
		\SN(\rho) \equiv \min_{\mb{D}(\rho)}\max_{|\phi_i\> \in \mb{D}(\rho)}\SN(\ket{\phi_i}),
	\end{equation}
	where $\bbD(\rho)$ denotes  a pure-state decomposition of $\rho$. A quantum state with Schmidt number $r$ can be prepared from $|+_r\rangle$ using local operations and classical communication (LOCC), but cannot be prepared from any state with a lower Schmidt rank. Hence, the Schmidt number of a mixed state $\rho$ quantifies its entanglement dimension.
	
	Similar to entanglement witnesses, an observable $O_{\rmW}$ is termed a \textit{Schmidt-number-$k$ witness} ($k\ge 2$) \cite{sanpera2001schmidt} if $\Tr(\rho O_{\rmW}) \ge 0$ for all $\rho\in \caS_{k-1}$, and there exists at least one $\rho\in \caS_{k}$ such that $\Tr(\rho O_{\rmW}) < 0$. The most common Schmidt-number-$k$ witness is
	\begin{equation} \label{eq:common_witness}
		\frac{k-1}{d}I - |+_d\>\<+_d|.
	\end{equation}
	For more general constructions, we have the following well-known result \cite{terhal2000schmidt,weilenmann2020entanglement}. Suppose $|\psi\>\in\caH_{\rmA\rmB}$. Then
	$cI - |\psi\>\<\psi|$
	is a Schmidt-number-$k$ witness iff
	\begin{equation}
		\sn(|\psi\>) \ge k,\quad \sum_{i=0}^{k-2}\lambda^{\downarrow}_i \le c < \sum_{i=0}^{k-1} \lambda^{\downarrow}_i,
	\end{equation}
	where $\lambda^{\downarrow}_i$ is the $i$-th largest element in $\sigma(\psi)$.
	
	\subsection{Schmidt number certification via the correlation matrix}
	
	Here we review the CM method for certifying the Schmidt number \cite{imai2021bound, liu2023characterizing, wyderka2023probing}, assuming that $d_\rmA=d_\rmB=d$. The CM criterion is a weaker but more concise version of the covariance matrix criterion \cite{liu2024bounding}. Let $\{P^{(\rmA)}_i\}_{i=1}^{d^2-1}$ and $\{P^{(\rmB)}_i\}_{i=1}^{d^2-1}$ be  orthogonal operator bases of $\caL_0^\rmH(\caH_\rmA)$ and $\caL_0^\rmH(\caH_\rmB)$, respectively, that satisfy the conditions
	\begin{equation}\label{eq:OperatorBasesCon}
		\begin{aligned}
			\Tr(P^{(\rmA)}_i) &=\Tr(P^{(\rmB)}_i)  =0,\\   
			\Tr(P^{(\rmA)}_i P^{(\rmA)}_{j}) &= \Tr(P^{(\rmB)}_i P^{(\rmB)}_{j}) = d \delta_{i,j}. 
		\end{aligned}
	\end{equation}
	The CM of the state with respect to these bases is the $(d^2-1)\times (d^2-1)$ matrix with entries
	\begin{equation}
		T_{ij} \equiv \frac{1}{d}\Tr\left[\rho \left(P^{(\rmA)}_i\otimes P^{(\rmB)}_j\right)\right].
	\end{equation}
	Suppose $T$ has singular values $\{v_i\}_{i=1}^{d^2-1}$. Then its  Schatten $p$-norm reads
	\begin{equation}
		\|T\|_p \equiv \left(\sum_{i=1}^{d^2-1}v_i^p\right)^{\frac{1}{p}} =\left\{\Tr\Bigl[(TT^\top)^{\frac{p}{2}}\Bigr]\right\}^{\frac{1}{p}}.
	\end{equation}
	The CM criterion states that \cite{liu2023characterizing, wyderka2023probing}
	\begin{equation}\label{eq:correlation_matrix}
		\text{ if } \SN(\rho) \le k,\text{ then } \|T\|_1 \le k- \frac{1}{d}.
	\end{equation}
	Note that this criterion is independent of the specific choices of local operator bases. 
	If the second inequality in \eref{eq:correlation_matrix} is violated, then we can conclude that $\SN(\rho)$ is larger than $k$.
	
	In practice, it is difficult to estimate the singular values of $T$ and the 1-norm $\|T\|_1$. Nevertheless, the second and fourth moments of the singular values, that is, $(\|T\|_2^2, \|T\|_4^4)$, 
	can be estimated from the second and fourth order Haar-randomized correlators $\caC^{(n)}$ constructed using a special observable $O$:
	\begin{gather}\label{eq:correlator_moment}
		\!\!\caC^{(n)} \equiv \int dU_{\rmA} \int dU_{\rmB} \left\<U_{\rmA}OU_{\rmA}^\dag\otimes U_{\rmB} O U_{\rmB}^\dag\right\>_\rho^n,
	\end{gather}
	where $\<O\>_\rho \equiv \Tr(\rho O)$. More precisely, we have the following relations:
	\cite{imai2021bound}:
	\begin{equation}
		\begin{aligned}
			\caC^{(2)} &= \frac{d^2}{(d-1)^2}\|T\|_2^2,\\
			\caC^{(4)} &= \frac{2d^4}{3(d-1)^4}\|T\|_4^4 + \frac{1}{3}\left(\caC^{(2)}\right)^2.
		\end{aligned}
	\end{equation}
	Define the region
	\begin{equation}
		\caT_k \equiv \left\{\left(\|T\|_2^2,\|T\|_4^4\right): \|T\|_1 \le k-d^{-1}\right\},
	\end{equation}
	whose  boundary can be determined by solving an optimization problem. 
	Then a moment-based CM criterion can be formulated as follows:
	\begin{equation}\label{eq:CMM4}
		\text{ if } \SN(\rho) \le k,\text{ then } \left(\|T\|_2^2,\|T\|_4^4\right)\in \caT_k.
	\end{equation}
	This criterion is referred to as the \textit{moment-based CM criterion} henceforth.

	In addition, by virtue of the H\"older inequality we can deduce that
	\begin{align}
		\!\!	\|T\|_2^2 &=\sum_{i=1}^{d^2-1}v_i^{\frac{2}{3}} v_i^{\frac{4}{3}}
		\leq \left(\sum_{i=1}^{d^2-1}v_i \right)^{\frac{2}{3}} \left(\sum_{i=1}^{d^2-1} v_i^4\right)^{\frac{1}{3}}\nonumber\\
		&=\|T\|_1^{\frac{2}{3}}\|T\|_4^{\frac{4}{3}},
	\end{align}
	which implies that
	\begin{align}
		\|T\|_1\geq \frac{\|T\|_2^3}{\|T\|_4^2}.
	\end{align}
	In conjunction with \eref{eq:correlation_matrix} we can deduce an alternative moment-based criterion
	\begin{equation}\label{eq:CMM4b}
		\text{ if } \SN(\rho) \le k,\text{ then } \frac{\|T\|_2^3}{\|T\|_4^2} \le k- \frac{1}{d},
	\end{equation}
	which is potentially weaker but simpler than the counterpart in \eref{eq:CMM4}.

	Currently, it is still quite challenging to apply the moment-based criteria in real experiments, although it is easier than the original CM criterion. To estimate the fourth order correlator $\caC^{(4)}$ experimentally, we need to implement unitary transformations in  a unitary $4$-design.
	Unfortunately, the Clifford group widely used in quantum information processing  fails to be a unitary $4$-design \cite{webb2015clifford,zhu2017multiqubit,ZhuKGG16}, and it is very challenging to sample from a unitary 4-design in experiments. Therefore, it is desirable to construct alternative certification protocols that do not rely on unitary 4-designs, which is a main motivation behind the current work. 
	
	\section{\label{sec:k_spectrum}Schmidt number certification via the \texorpdfstring{$k$}{empty}-reduction map}
	
	\subsection{\texorpdfstring{$k$}{empty}-positive maps}
	Given a finite-dimensional Hilbert space $\caH$, 
	an operator $O \in \caL(\caH)$ is positive (semidefinite),  denoted as $O\geq 0$, if $\langle\psi|O|\psi\rangle \ge 0$ for all $|\psi\rangle \in \caH$. A linear map $\caM:\caL(\caH) \to \caL(\caH)$ is \emph{positive} if 
	\begin{equation}
		\caM(O) \geq 0 \quad \forall\, O\geq 0.
	\end{equation}
	A linear map $\caM : \caL(\caH) \to \caL(\caH)$ is  \emph{$k$-positive} if  the map $\caI_k \otimes \caM$ is positive, where $\caI_k$ is the identity map on an auxiliary Hilbert space of dimension $k$ \cite{choi1972positive, takasaki1983geometry, tomiyama1985geometry, chruscinski2014entanglement}. 
	
	Next, we turn to a bipartite system with Hilbert space $\caH_{\rmA\rmB}=\caH_\rmA\otimes \caH_\rmB$ and explain the applications of $k$-positive  maps in certifying high-dimensional entanglement.  Let $\caI_\rmA$ be the identity map on $\caL(\caH_\rmA)$. 
	A linear map $\caM$ on $\caL(\caH_\rmB)$ is $k$-positive iff $\caI_\rmA\otimes\caM(|\psi\rangle\langle\psi|)$ is positive  for every pure state $|\psi\rangle \in \caH_{\rmA\rmB}$ with $\sn(|\psi\>)\leq k$ \cite{terhal2000schmidt, VidaW02}. Then we also have $\caI_\rmA\otimes\caM (\rho) \geq 0$ for every mixed state $\rho$ with $\SN(\rho)\leq k$, but the converse is not necessarily true. Therefore, if $\caI_\rmA\otimes\caM (\rho)$ is not positive, then we can conclude that $\SN(\rho)> k$. This fact is instrumental to constructing protocols for certifying high-dimensional entanglement.

	\subsection{\label{sec:kReductionC}\texorpdfstring{$k$}{empty}-reduction map and \texorpdfstring{$k$}{empty}-reduction criterion}
	
	One of the most important and common $k$-positive maps is the $k$-reduction map \cite{takasaki1983geometry, tomiyama1985geometry}, which is defined as follows:
	\begin{equation}
		\caR_k(O) \equiv k\Tr(O)I - O.
	\end{equation}
	Here we are mainly interested in the map
	$\caI_\rmA\otimes\caR_k$ acting on $\caL(\caH_\rmA\otimes \caH_\rmB)$. 
	To simplify the following discussion, the expression $\caI_\rmA\otimes\caR_k(\rho)$ 
	will be abbreviated as  $\caR_k(\rho)$ and called the \emph{$k$-reduced operator} associated with $\rho$. When $\rho=|\psi\>\<\psi|$ is a pure state, $\caR_k(\rho)$ can further be abbreviated as $\caR_k(\psi)$.  
	Our approach for certifying the Schmidt number is based on the following well-known result \cite{terhal2000schmidt}: 
	\begin{proposition}\label{prop:positive_k_reduction}
		If $\SN(\rho)\le k$, then 
		\begin{equation}\label{eq:k_reduction_criterion}
			\caR_k(\rho)=k\rho_{\rmA}\otimes I_{\rmB}-\rho \geq 0.
		\end{equation}
		When $\rho$ is a pure state, the condition $\SN(\rho)\le k$ is also necessary to guarantee \eref{eq:k_reduction_criterion}. 
	\end{proposition}
	If the condition in \eref{eq:k_reduction_criterion}, known as  the $k$-\emph{reduction condition},  is violated, 
	then we can conclude that $\te{SN}(\rho)> k$. The basic idea of this criterion is analogous to the famous PPT criterion. By considering the expectation value of the $k$-reduced operator with respect to the maximally entangled state $|+_d\rangle$ we can deduce that 
	\begin{equation}
		\text{if } \sn(\rho) \le k, \quad \text{ then } \langle +_d|\rho|+_d\rangle \le \frac{k}{d},
		\label{eq:fidelitybased}
	\end{equation}
	which reproduces a fidelity-based criterion.

	By analogy with the entanglement negativity 
	\cite{VidaW02}, here we  define the \emph{$k$-reduction negativity} of $\rho$, denoted by $\caN_k(\rho)$, as the absolute value of the sum of the negative eigenvalues of $\caR_k(\rho)$. When $\rho=|\psi\>\<\psi|$ is a pure state, $\caN_k(\rho)$ is abbreviated as $\caN_k(\psi)$. Alternatively, $\caN_k(\rho)$ can be expressed as follows:
	\begin{align}
		\caN_k(\rho) & = \frac{1}{2}\left(\|\caR_k(\rho)\|_1 - \Tr[\caR_k(\rho)]\right) \nonumber\\
		& = \frac{1}{2}\left(\|\caR_k(\rho)\|_1 -kd_\rmB+1 \right). \label{eq:k_negativity}
	\end{align}
	It measures the degree to which the $k$-reduction criterion is violated. With this definition, the $k$-reduction criterion can also be formulated as follows:
	\begin{equation}
		\text{if } \SN(\rho) \le k, \quad \text{ then }\caN_k(\rho) = 0.
	\end{equation}
	In other words, if $\caN_k(\rho) > 0$, then we can deduce that $\SN(\rho) > k$.

	\subsection{\texorpdfstring{$k$}{empty}-reduction negativity of pure states}\label{sec:pure_state_negativity}
	The $k$-reduction negativity plays an important role in certifying high-dimensional entanglement as we have seen in \sref{sec:kReductionC}. Here we determine the $k$-reduction negativity of a general bipartite pure state $|\psi\>$ in $\caH_{\rmA\rmB}$ and clarify its basic properties. To this end, we need to introduce some additional notations and results. 
	
	Given two real vectors  $\bfx$ and $\bfy$ in $\bbR^d$, we can rearrange the entries of each vector in nonincreasing order to obtain $\bfx^{\downarrow} = (x'_{0},x'_{1},\ldots, x'_{d-1})$ and $\bfy^{\downarrow} = (y'_{0},y'_{1},\ldots, y'_{d-1})$. If
	\begin{equation}
		\sum_{i=0}^{\ell-1}x'_{i} \ge \sum_{i=0}^{\ell-1}y'_{i}\quad \forall \ell,
	\end{equation}
	and the inequality is saturated when $\ell = d$, then we say $\bfx$ \emph{majorizes} $\bfy$, denoted by $\bfx \succ \bfy$ or $\bfy\prec\bfx$. A function $F$  on a subset of $\bbR^d$ is \emph{Schur concave} if the condition $\bfx \succ \bfy$ implies that $F(\bfx) \le F(\bfy)$ \cite{roberts1993convex,nielsen2010quantum}.

	Given $\blambda \in \Delta_d$, let $\Omega_k(\blambda)$ be the $d\times d$ matrix with entries 
	\begin{align}\label{eq:Omegaklambda}
		\Omega_k(\blambda)_{ij}=k\lambda_i \delta_{ij}-\sqrt{\lambda_i\lambda_j}
	\end{align}
	and define
	\begin{align}\label{eq:thetaklambda}
		\theta_k(\blambda)\equiv\max\{0,-\sigma_{\min}(\Omega_k(\blambda))\}. 
	\end{align}
	This function is interesting because of its close relation with the $k$-reduction negativity as we shall see shortly. The basic properties of $\Omega_k(\blambda)$ and $\theta_k(\blambda)$ are summarized in \lsref{lem:thetaOmegalambda} and \ref{lem:Tomega} in \aref{app:kRNegativitySN}. 
	
	Now, we can  clarify the basic properties of the $k$-reduced operator $\caR_k(\psi)$ 
	as summarized in the following theorem and proved in \aref{app:kRNegativitySN}. Part of the results has been proved in previous works \cite{tomiyama1985geometry,terhal2000schmidt}.

	\begin{theo}\label{thm:pure_state}
		Suppose  $|\psi\>\in  \caH_{\rmA\rmB}$ is a pure state with Schmidt vector $\blambda$, Schmidt number $r$, and $\tilde{r}$ distinct nonzero Schmidt coefficients. Then the $k$-reduced operator $\caR_k(\psi)$ has spectrum
		\begin{align}\label{eq:RkrhoSpectrum}
			\sigma(\caR_k(\psi))&=\sigma(\Omega_k(\blambda)) \cup [k \sigma(\psi)]^{\times (d_\rmB-1)},   
		\end{align}
		where the superscript $d_\rmB-1$ denotes the multiplicity. In addition, $\caR_k(\psi)$ has $2\tilde{r}$  distinct nonzero eigenvalues when $k\neq r$ and $2\tilde{r}-1$  distinct nonzero eigenvalues when $k= r$.  If $k \ge r$, then $\Omega_k(\blambda)\geq 0$ and  $\caR_k(\psi)\geq 0$.  If instead $1 \le k<r$, then both $\Omega_k(\blambda)$ and  $\caR_k(\psi)$ have exactly one negative eigenvalue. 
	\end{theo}
	
	\thref{thm:pure_state} implies that 
	\begin{align}
		\caN_k(\psi)=\max\{0,-\sigma_{\min}(\caR_k(\psi)) \} = \theta_k(\blambda).
	\end{align}
	In conjunction with \lref{lem:thetaOmegalambda} in \aref{app:kRNegativitySN}, we can further clarify the properties of  $\caN_k(\psi)$. The following theorem is also proved in 
	\aref{app:kRNegativitySN}. 
	\begin{theo}\label{thm:k_negativity}
		Suppose  $|\psi\>\in  \caH_{\rmA\rmB}$ is a pure state with Schmidt vector $\blambda$.  Then $\theta_k(\blambda)$ is Schur concave in $\blambda$ and nonincreasing in $k$. If $\sn(|\psi\>)=r$, then $\caN_k(\psi) \le 1 - k/r$, and the inequality is saturated when $|\psi\rangle$ is a maximally entangled state with $\sn(|\psi\>) = r$. 
	\end{theo}
	
	Suppose $|\psi\>$ and $|\psi'\>$ are two pure states in $\caH_{\rmA\rmB}$. 
	According to  Nielsen's majorization criterion  \cite{nielsen1999conditions},
	$|\psi\rangle$ can be transformed into $|\psi'\rangle$ under LOCC iff the Schmidt vector of $|\psi\rangle$ is majorized by that of $|\psi'\rangle$. In conjunction with  \thref{thm:k_negativity}, we can deduce that the $k$-reduction negativity is nonincreasing under LOCC and thus can serve as a high-dimensional entanglement monotone for pure states.
	
	By virtue of \thref{thm:k_negativity} we can further derive tight bounds on the $k$-reduction negativity $\caN_k(\rho)$ and eigenvalues of $\caR_k(\rho)$ for a general mixed state~$\rho$. 
	\begin{corollary}\label{coro:spectrum}
		Suppose $\rho \in \caS(\caH_{\rmA\rmB})$ with $d_\rmA = d_\rmB = d$. Then
		\begin{equation}
			\caN_k(\rho) \leq 1-\frac{k}{d},\quad \left(\frac{k}{d}-1\right) I\leq \caR_k(\rho) \leq  k I.
		\end{equation}
	\end{corollary}
	
	\subsection{\label{sec:depolarization}\texorpdfstring{$k$}{empty}-reduction negativity of pure states with depolarizing noise}
	As a generalization, in this section we determine the $k$-reduction negativity of an arbitrary  pure state with depolarizing noise. Suppose $|\psi\>\in \caH_{\rmA\rmB}$ is a pure state with $\sn(|\psi\>)=r$  and $\rho$ has the form
	\begin{equation}\label{eq:depolarizing}
		\rho = (1-\varepsilon)|\psi\rangle\langle \psi| + \varepsilon\frac{I}{d_\rmA d_\rmB},\quad 
		\varepsilon \in [0,1]. 
	\end{equation}
	Then 
	\begin{align}
		\caR_k(\rho) &= (1-\varepsilon)\caR_k(\psi) + \varepsilon\caR_k\left(\frac{I}{d_\rmA d_\rmB}\right)\nonumber \\
		&=(1-\varepsilon)\caR_k(\psi) + \frac{\varepsilon(kd_\rmB-1)}{d_\rmA d_\rmB}I, 
	\end{align}
	and there exists a simple linear relation between the spectrum of $\caR_k(\rho)$ and the spectrum of $\caR_k(\psi)$. Notably, the minimal eigenvalue of $\caR_k(\rho)$ equals
	\begin{align}
		(1-\varepsilon)\sigma_{\min}(\caR_k(\psi))+ \frac{\varepsilon(kd_\rmB-1)}{d_\rmA d_\rmB}.
	\end{align}
	In conjunction with \thsref{thm:pure_state} and \ref{thm:k_negativity}, it is now straightforward to deduce the following result.
	\begin{proposition}\label{prop:depolarizing}  The $k$-reduction negativity of $\rho$ in  \eref{eq:depolarizing} reads
		\begin{equation}
			\caN_k(\rho) = \begin{cases}
				0, & \varepsilon \ge \varepsilon^*,\\
				(1-\varepsilon)\caN_k(\psi) - \frac{\varepsilon(d_\rmB k-1)}{d_\rmA d_\rmB}, & \varepsilon < \varepsilon^*,
			\end{cases}
		\end{equation}
		where 
		\begin{gather}
			\varepsilon^* \equiv \frac{d_\rmA d_\rmB \caN_k(\psi) }{kd_\rmB-1+d_\rmA d_\rmB \caN_k(\psi)}.
		\end{gather}
	\end{proposition}
	
	\begin{figure}
		\includegraphics[width = 0.45\textwidth]{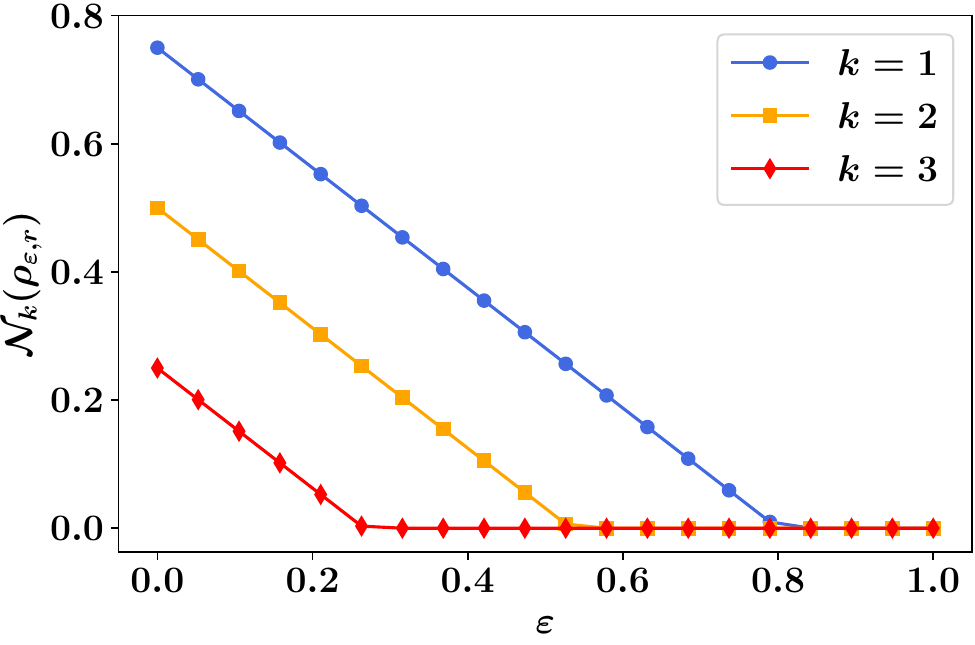}
		\caption{\label{fig:rank_r}The $k$-reduction negativity $\caN_k(\rho_{\varepsilon,r})$ as a function of $\varepsilon$ for $k = 1,2,3$. Here we set $r=d_\rmA = d_\rmB = 4$. The $k$-reduction negativity is nonincreasing in $k$ and $\varepsilon$ as predicted in \pref{prop:depolarizing}. }
	\end{figure}
	
	Next, we turn to  the special case in which
	$|\psi\rangle$ is a maximally entangled state of Schmidt rank $r$. Define
	\begin{equation}\label{eq:eps_r}
		\rho_{\varepsilon,r} \equiv (1-\varepsilon)|+_r\rangle\langle +_r| + \varepsilon\frac{I}{d_\rmA d_\rmB},\quad \varepsilon \in [0,1].
	\end{equation}
	The dependence of $\caN_k(\rho_{\varepsilon,r})$ on $k$ and $\varepsilon$ is illustrated in \fref{fig:rank_r}.
	Define $\varepsilon^{\rmRM}_c$ as the minimum value 
	of $\varepsilon\geq 0$ such that $\caN_{r-1}(\rho_{\varepsilon,r})= 0$. By virtue of \pref{prop:depolarizing} with $\caN_{r-1}(+_r)=1/r$ we can deduce that
	\begin{equation}\label{eq:epr2}
		\varepsilon^{\rmRM}_c = \left(1 + \frac{r^2 - r}{d_\rmA} - \frac{r}{d_\rmA d_\rmB}\right)^{-1}.
	\end{equation}
	As the local dimensions $d_\rmA, d_\rmB$ increase, $\varepsilon^{\rmRM}_c$ approaches 1, so the $k$-reduction criterion can certify the Schmidt number of
	$\rho_{\varepsilon,r}$ with arbitrary noise strength asymptotically.
	
	Using the $k$-reduction criterion, we can further derive the following informative lower and upper bounds for the Schmidt number of $\rho_{\varepsilon,r}$ as proved in \aref{app:proof_theo_5}.
	\begin{proposition}	\label{prop:rank_r_MEDP}
		The Schmidt number of $\rho_{\varepsilon,r}$ in \eref{eq:eps_r} satisfies
		\begin{equation}\label{eq:rank_r_MEDPSN}
			\left\lceil \frac{(1 + u)r}{1 + d ru}\right\rceil \le \SN(\rho_{\varepsilon,r}) \le \left\lceil  \frac{(1 + u)r}{1 + r^2u}\right\rceil,
		\end{equation}
		where $u=\varepsilon/[(1-\varepsilon)d_\rmA d_\rmB]$. 
		If $r \le \sqrt{d}$ and $\varepsilon < 1/2$, then $\SN(\rho_{\varepsilon,r}) = r$.
	\end{proposition}
	
	\subsection{\label{sec:RMCCMCcom}Comparison with the CM criterion}
	
	In this  section we compare the $k$-reduction criterion in \eref{eq:k_reduction_criterion} with the CM criterion in \eref{eq:correlation_matrix} under various situations. From \pref{prop:positive_k_reduction}, we know that the $k$-reduction criterion can certify 
	the Schmidt number of any pure state. By contrast,  the CM criterion cannot. For example, consider the following state with Schmidt rank $r = 4$:
	\begin{equation}\label{eq:example}
		\sqrt{\frac{4}{5}}|00\rangle + \sqrt{\frac{1}{15}}|11\rangle + \sqrt{\frac{1}{15}}|22\rangle + \sqrt{\frac{1}{15}}|33\rangle,
	\end{equation}
	assuming that $d_\rmA = d_\rmB = d = 4$. The 1-norm of its CM is approximately 2.7231, which is smaller than the upper bound  $r-1-d^{-1} = 11/4$ in \eref{eq:correlation_matrix} with $k=r-1$. Thus, from the CM criterion, we can only conclude that its Schmidt number is at least 3, instead of the actual Schmidt number $4$.
	
	\begin{figure}[t]
		\centering
		\includegraphics[width = 0.45\textwidth]{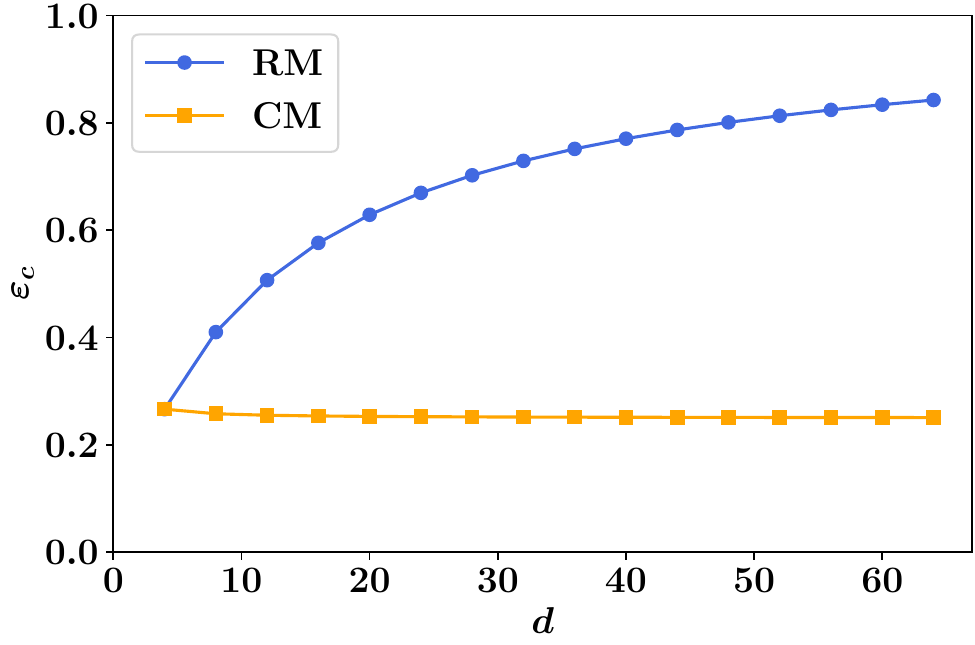}
		\caption{\label{fig:cmc_eps}Noise thresholds $\varepsilon^{\rmRM}_c$ and $\varepsilon^{\rmCM}_c$ associated with
			the $k$-reduction criterion	and CM criterion in certifying the Schmidt number of the state $\rho_{\varepsilon,r}$ in \eref{eq:eps_r} with $r=4$. Here $d_\rmA=d_\rmB=d$, and $\varepsilon^{\rmRM}_c$ is given in \eref{eq:epr2}.}	
	\end{figure}
	
	Next,  we consider the state $\rho_{\varepsilon,r}$ defined in \eref{eq:eps_r}. According to the discussion in \sref{sec:depolarization}, the $k$-reduction criterion can certify that the Schmidt number of $\rho_{\varepsilon,r}$ is $r$ when $\varepsilon$ is smaller than the threshold  $\varepsilon^{\rmRM}_c$ in \eref{eq:epr2}. As an analogy, we define $\varepsilon^{\rmCM}_c$ as the maximum value of $\varepsilon$ satisfying $\|T\|_1 \ge r-1-d^{-1}$. The relations between the two thresholds and the local dimension $d$ are illustrated in \fref{fig:cmc_eps}, which implies that the $k$-reduction criterion can tolerate much higher noise strength than the CM criterion.
	
	\begin{figure}[b]
		\centering
		\includegraphics[width = 0.45\textwidth]{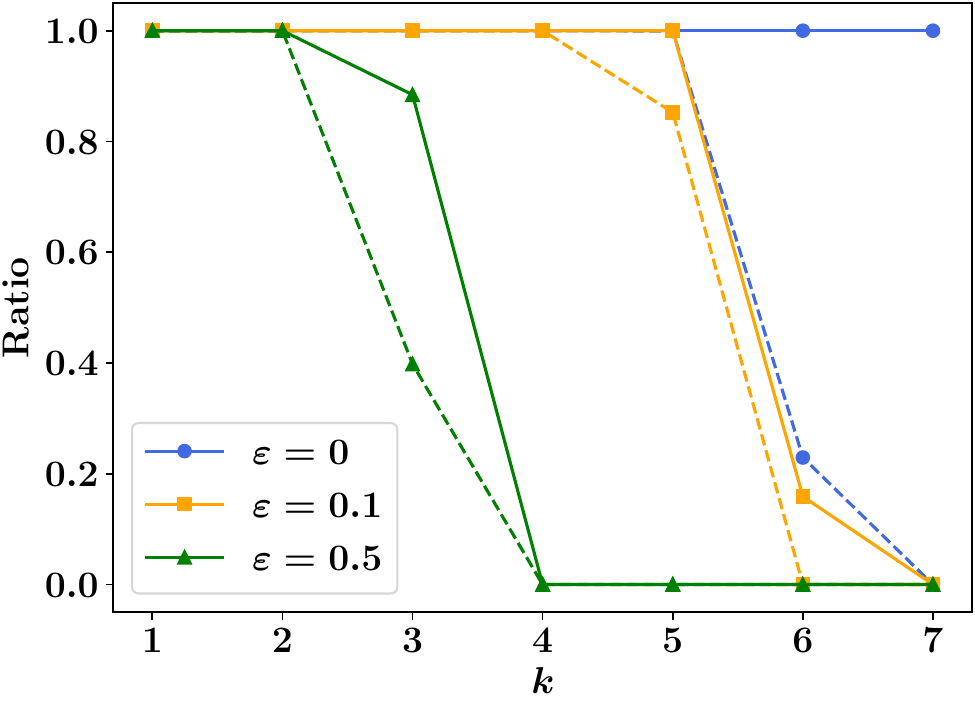}
		\caption{\label{fig:haar_random}Detection ratios of the $k$-reduction criterion and CM criterion for the ensemble of states in  \eref{eq:Haar_ensemble}. The solid lines correspond to the $k$-reduction criterion in \eref{eq:k_reduction_criterion}, while the dashed lines correspond to the CM criterion in \eref{eq:correlation_matrix}. 		
			Here $d_\rmA = d_\rmB = 8$ and three different noise strengths are indicated in the legend.  The performance of the $k$-reduction criterion is better than the CM criterion in all cases under consideration.}
	\end{figure}
	
	We then consider Haar-random pure states with depolarizing noise. Define the ensemble
	\begin{equation}\label{eq:Haar_ensemble}
		\caE_{\Haar,\varepsilon} \equiv \left\{(1-\varepsilon)|v\rangle\langle v| + \varepsilon \frac{I}{d^2} : |v\> \in \caH_{\rmA\rmB}\right\},
	\end{equation}
	where $d_\rmA = d_\rmB = d$ and $0\leq \varepsilon\leq 1$. When $\varepsilon$ is sufficiently small, a random sample $\rho$ from this ensemble almost surely has the maximum Schmidt number, that is, $\sn(\rho) = d$. To compare the performance of 
	the $k$-reduction criterion and CM criterion, \fref{fig:haar_random} illustrates their  detection ratios determined by numerical simulations as functions of $k$ for three different noise strengths $\varepsilon$. 
	Here the detection ratio of a criterion  denotes the percentage of quantum states in a given  ensemble that is certified to have $\sn(\rho) > k$ by this criterion. The $k$-reduction criterion has higher detection ratios than the CM criterion in all cases under consideration.  In other words, the $k$-reduction criterion is  more effective in certifying the Schmidt numbers of pure states with depolarizing noise.
	
	Finally, we consider random mixed states with  induced measures \cite{zyczkowski2001induced}. Let $\caH_\rmC$ be a $K$-dimensional Hilbert space. 
	Denote by $\caE_{D,K}$ the ensemble of mixed states on $\caH_{\rmA\rmB}$ induced by Haar-random pure states on $\caH_{\rmA\rmB}\otimes \caH_\rmC$.   \Tsref{tab:k_reduction_induced} and \ref{tab:correlation_matrix_induced} summarize the detection ratios of the  $k$-reduction criterion and CM criterion for various choices of the parameters $K$ and $k$ as determined by  numerical simulations with $d_\rmA=d_\rmB=16$. 
	According to the two tables, the CM criterion is stronger than the $k$-reduction criterion for large $K$, in which case the states in the ensemble $\caE_{D,K}$ tend to have low purities. In the special case $K=2$, the two criteria have comparable performance: both of them can certify that almost all states in the ensemble $\caE_{256,2}$ have Schmidt numbers at least $9$.
	
	\begin{center}
		\begin{table} 
			\caption{\label{tab:k_reduction_induced}Detection ratio of the $k$-reduction criterion in  \eref{eq:k_reduction_criterion}
				for the ensemble $\caE_{256,K}$. Here $d_\rmA=d_\rmB=16$. }
			\begin{tabular}{c|ccccc}
				\hline
				\hline
				\backslashbox{$k$}{\raisebox{-0.7ex}{$K$}} & 2 & 3 & 4 & 5 & 6 \\
				\hline
				\\[-0.8em]
				2 & 1 & 1 & 1 & 1 & 1\\
				\\[-0.8em]
				3 & 1 & 1 & 1 & 1 & 0.9977\\
				\\[-0.8em]
				4 & 1 & 1 & 1 & 0.0168 & 0\\
				\\[-0.8em]
				5 & 1 & 1 & 0.0009 & 0 & 0 \\
				\\[-0.8em]
				6 & 1 & 0.1606 & 0 & 0 &  0\\
				\\[-0.8em]
				7 & 1 & 0 & 0 & 0 &  0\\
				\\[-0.8em]
				8 & 1 & 0 & 0 & 0 &  0\\
				\\[-0.8em]
				9 & 0.0010 & 0 & 0 & 0 &  0\\
				\hline
				\hline
			\end{tabular}

			\caption{\label{tab:correlation_matrix_induced}Detection ratio of the CM criterion in \eref{eq:correlation_matrix}
				for the ensemble $\caE_{256,K}$. Here $d_\rmA=d_\rmB=16$.}			
			\begin{tabular}{c|ccccc}
				\hline
				\hline
				\backslashbox{$k$}{\raisebox{-0.7ex}{$K$}} & 2 & 3 & 4 & 5 & 6 \\
				\hline
				\\[-0.8em]
				5 & 1 & 1 & 1 & 1 & 1 \\
				\\[-0.8em]
				6 & 1 & 1 & 1 & 0 & 0 \\
				\\[-0.8em]
				7 & 1 & 1 & 0 & 0 & 0 \\
				\\[-0.8em]
				8 & 1 & 0 & 0 & 0 & 0 \\
				\\[-0.8em]
				9 & 0.0051 & 0 & 0 & 0 & 0 \\
				\hline
				\hline
			\end{tabular}
		\end{table}
	\end{center}

	\subsection{\texorpdfstring{$k$}{empty}-reduction Schmidt-number witnesses}
	\label{sec:kreductionwitness}
	
	Here we introduce a family of Schmidt-number-$(k+1)$ witnesses by virtue of the $k$-reduction map, which  will be useful
	for certifying the entanglement dimensionality of $(k+1)$-unfaithful states as shown in the next subsection. 
	
	\begin{lemma}\label{lem:SNwitness}
		Suppose $\rho,\varrho\in\caS(\caH_{\rmA\rmB})$; then we have $\Tr[\caR_k(\rho)\varrho]=	\Tr[\rho\caR_k(\varrho)]$. If $\Tr[\caR_k(\rho)\varrho]<0$, then  both $\rho$ and $\varrho$ have Schmidt numbers at least $k+1$; both $\caR_k(\rho)$ and  $\caR_k(\varrho)$ are Schmidt-number-$(k+1)$ witnesses, which can certify $\varrho$ and $\rho$, respectively. 
	\end{lemma}
	\Lref{lem:SNwitness} follows from the facts that the $k$-reduction map is self-adjoint and that $\caR_k(\rho)\geq 0$ whenever $\rho\in \caS_k$. 
	Now, suppose $\caR_k(\rho) \not \ge 0$, then any state supported in a subspace associated with negative eigenvalues of $\caR_k(\rho)$  has Schmidt number at least $k+1$. In particular, any eigenstate of  $\caR_k(\rho)$ with a negative eigenvalue has Schmidt number at least $k+1$. The following theorem is a simple corollary of this observation.
	\begin{theo}\label{thm:kreductionwitness}
		Suppose $\rho\in\caS(\caH_{\rmA\rmB})$, $\caR_k(\rho) \not \ge 0$, and $|\Psi\>$ is an eigenstate of $\caR_k(\rho)$ with a negative eigenvalue. Then $\caR_k(\rho)$ and  $\caR_k(\Psi)$ are Schmidt-number-$(k+1)$ witnesses, which can certify $|\Psi\>$ and $\rho$, respectively. 
	\end{theo}
	
	A Schmidt-number-$(k+1)$ witness of the form $\caR_k(\rho)$ with $\caR_k(\rho) \not \ge 0$ is henceforth called a \emph{$k$-reduction Schmidt-number witness}. Thanks to \lref{lem:SNwitness} and \thref{thm:kreductionwitness}, any state $\rho$ certifiable by the $k$-reduction map can also be certified by a $k$-reduction Schmidt-number witness. Note that the witness based on the maximally entangled state $|+_d\>$ in \eref{eq:common_witness} can also  be regarded as a $(k-1)$-reduction Schmidt-number witness. In addition, by applying the witness  $\caR_k(\rho)$ to the state $\rho$ itself we can obtain a purity-based criterion:
	\begin{equation}
		\textrm{if } \sn(\rho) \le k,\quad \textrm{then } \Tr\left(\rho_{\rmA}^2\right) \ge \frac{1}{k}\Tr(\rho^2)
	\end{equation}
	with $\rho_{\rmA} = \Tr_{\rmB}(\rho)$.
	
	\subsection{Certification of entanglement dimensionality of \texorpdfstring{$k$}{empty}-unfaithful states}
	\label{sec:k_unfaithful}
	
	A quantum state $\rho\in\caS(\caH_{\rmA\rmB})$ is  \textit{$k$-unfaithful} if $\Tr(\rho O_{\rmW}) \ge 0$ for all Schmidt-number-$k$ witnesses of the form $O_{\rmW} = cI - |\psi\>\<\psi|$ \cite{weilenmann2020entanglement}. By definition, any state with Schmidt number at most $k-1$ is $k$-unfaithful. By contrast, a state $\rho$ is  $k$-\textit{faithful} if there exists a Schmidt-number-$k$ witness of the form $O_{\rmW} = cI - |\psi\>\<\psi|$  that satisfies $\Tr(\rho O_{\rmW}) < 0$; in other words, the state $\rho$ is $k$-faithful iff there exists $|\psi\>\in \caH_{\rmA\rmB}$ with $\sn(|\psi\>) \ge k$ such that 
	\begin{equation}\label{eq:kunfaithfulcondition}
		\<\psi|\rho|\psi\> > \sum_{i=0}^{k-2} \sigma(\psi)_i^{\downarrow}.
	\end{equation}
	By convention, 2-unfaithful (faithful) is also referred to as unfaithful (faithful). It is known that the state $\rho$ is faithful iff there exist local unitary operators $U_{\rmA}, U_{\rmB}$ such that \cite{guhne2021geometry} 
	\begin{equation}\label{eq:criterion_unfaithful}
		\<+_d|\left(U_{\rmA}^\dag\otimes U_{\rmB}^\dag\right)\rho \left(U_{\rmA} \otimes U_{\rmB}\right)|+_d\> > \frac{1}{d}.
	\end{equation}
	However, a similar result does not hold when $k \ge 3$. The following proposition highlights the power of the $1$-reduction map and $1$-reduction Schmidt-number witnesses in certifying faithful states.
	\begin{proposition}
		Any faithful state in $\caS(\caH_{\rmA\rmB})$ can be certified by the $1$-reduction map and a $1$-reduction Schmidt-number witness.
	\end{proposition}

	\begin{proof}
		Suppose $\rho\in\caS(\caH_{\rmA\rmB})$ is faithful. Then, by  \eref{eq:criterion_unfaithful},  there exists a state $|\Psi\> = \left(U_{\rmA}\otimes U_{\rmB}\right)|+_d\>$ in $\caH_{\rmA\rmB}$ such that
		\begin{equation}
			\<\Psi|\caR_1(\rho)|\Psi\>  = \frac{1}{d} - \<\Psi|\rho|\Psi\> < 0,
		\end{equation}
		which implies that  $\caR_1(\rho)\not\ge 0$. So $\rho$ can be certified by the $1$-reduction map and a $1$-reduction Schmidt-number witness by \thref{thm:kreductionwitness}.
	\end{proof}
	
	Next, we  show that the 1-reduction map and $1$-reduction Schmidt-number witnesses can certify unfaithful states. Suppose $\caH_{\rmA\rmB}$ is a Hilbert space of local dimensions $d_\rmA=d_\rmB=d$,
	$|\Phi^+\> = (|00\> + |11\>)/\sqrt{2}$ is a Bell state in $\caH_{\rmA\rmB}$, and 
	\begin{equation}
		\rho = x|\Phi^+\>\<\Phi^+| + (1-x)\frac{I}{d^2},\quad x\in [0,1].
	\end{equation}
	Then the fidelity between $\rho$ and any 
	maximally entangled state in $\caH_{\rmA\rmB}$ is upper bounded by 
	\begin{equation}
		\frac{2x}{d} + \frac{1-x}{d^2}.
	\end{equation}
	Therefore, $\rho$ is unfaithful
	when $x \le (d-1)/(2d+1)$ according to \eref{eq:criterion_unfaithful}. In addition, the minimum eigenvalue of $\caR_1(\rho)$ reads
	\begin{equation}
		-\frac{1}{2}x + \frac{(1-x)(d-1)}{d^2}.
	\end{equation}
	If $x > 2/(d+2)$, then this eigenvalue is negative, which means $\caR_1(\rho) \not\ge 0$. Now, suppose $d > 4$ and
	\begin{equation}
		x \in \bigg(\frac{2}{d+2}, \frac{d-1}{2d+1}\bigg];
	\end{equation}
	then the state $\rho$ is unfaithful but certifiable by the 1-reduction map and some $1$-reduction Schmidt-number witnesses.
	
	In general,  $k$-unfaithfulness can be certified via semi-definite programming \cite{weilenmann2020entanglement}. To be specific, a state $\rho\in \caS(\caH_{\rmA\rmB})$ is $k$-unfaithful if there exist positive operators $E_{\rmA}\in\caL(\caH_\rmA), E_{\rmB}\in\caL(\caH_\rmB)$ and non-negative numbers $\mu_{\rmA},\mu_{\rmB}$ fulfilling the following conditions:
	\begin{gather}
		E_{\rmA} \otimes I_{\rmB} + I_{\rmA} \otimes E_{\rmB} \ge \rho,\nonumber\\
		\mu_{\rmA} + \mu_{\rmB} = 1,\nonumber\\ \label{eq:sdp}
		\Tr(E_{\rmA}) = \mu_{\rmA}(k-1),\\ 
		\Tr(E_{\rmB}) = \mu_{\rmB}(k-1),\nonumber\\
		E_{\rmA} \le \mu_{\rmA} I_{\rmA},\quad E_{\rmB} \le \mu_{\rmB} I_{\rmB}.\nonumber
	\end{gather}
	For instance, the following state is 3-unfaithful:
	\begin{equation}\label{eq:unfaithful}
		\begin{aligned}
			\rho_{\UF} &\equiv \frac{1}{2}|+_3\>\<+_3| + \frac{1}{2}|\Phi_2\>\<\Phi_2|,\\
			|\Phi_2\> &\equiv \frac{1}{\sqrt{2}}(|23\> + |32\>).
		\end{aligned}
	\end{equation}
	As a short proof, the conditions in \eref{eq:sdp} can be fulfilled by the following choices:
	\begin{equation}
		\begin{aligned}
			\mu_{\rmA} = \mu_{\rmB} = \frac{1}{2},\quad
			E_{\rmA} = E_{\rmB} = \frac{1}{4}\sum_{i=0}^3|i\>\<i|.
		\end{aligned}
	\end{equation}
	Moreover, the state $\rho_{\UF}$ can be certified by a $2$-reduction Schmidt-number witness introduced in \sref{sec:kreductionwitness}. Direct calculation shows that $\caR_2(\rho_{\UF})$ has a unique negative eigenvalue, and the corresponding eigenstate reads
	\begin{equation}
		\begin{aligned}
			|\Psi\> &= \alpha|00\> + \alpha|11\> + \sqrt{1-2\alpha^2}\lsp|22\>,\\
			\alpha &= \frac{\sqrt{3} + \sqrt{2}}{2\sqrt{3 + \sqrt{6}}} \approx 0.67.
		\end{aligned}
	\end{equation}
	In addition, $\Tr[\rho_{\UF}\caR_2(\Psi)] =\Tr[\caR_2(\rho)|\Psi\>\<\Psi|]< 0$, so the $2$-reduction Schmidt-number witness $\caR_2(\Psi)$ can certify that the state $\rho$ has Schmidt number at least 3.
	In experiments, the expectation value of the witness $\caR_2(\Psi)$ can be determined by measuring  $\<\Psi|\rho_\UF|\Psi\>$ and $\Tr[\Tr_{\rmB}(\Psi)\Tr_{\rmB}(\rho_\UF)]$.
	
	\section{\label{sec:moment}Schmidt number certification via \texorpdfstring{$k$}{empty}-reduction moments} 
	
	In this section we systematically develop \emph{the moment method}, which is designed to characterize the subset of positive operators within  a given set of Hermitian operators, and apply this method to certifying
	the Schmidt number. 
	
	\subsection{\label{sec:classical_shadow}PT moments and \texorpdfstring{$k$}{empty}-reduction moments}
	
	Here we give an overview of how the moment method is utilized in entanglement certification \cite{elben2020mixed}, which has connection but is different from the method in \cite{bohnet2012}. 
	
	The well-known PPT criterion states that if the partial transpose $\rho^{\top}$ (with respect to either party) of a bipartite state $\rho$ is not a positive operator, then $\rho$ is necessarily entangled.
	This criterion is proved to be useful for a large class of quantum states. However, non-complete positive maps, including the partial transpose and $k$-reduction map, are non-physical. We cannot  implement them directly in the laboratory. One solution to this problem is quantum tomography. After reconstructing the entire density matrix from the outcomes of experimental measurements, we can directly test the PPT criterion or $k$-reduction  criterion using a classical computer. Unfortunately, full quantum state tomography requires a huge measurement budget \cite{gross2010quantum,haah2016sample,o2016efficient}, which is too prohibitive for practical applications.  
	
	Recently, researchers have found more resource-efficient variants of the PPT criterion \cite{aaronson2018shadow,aaronson2019gentle,huan2020}. To apply the PPT criterion, we need to determine whether the spectrum of $\rho^{\top}$ contains a negative eigenvalue, which  can be certified by the partial transpose (PT) moments:
	\begin{equation}
		p^{\top}_n\equiv\Tr[(\rho^{\top})^n],\quad  n = 1,2,\ldots.
	\end{equation} 
	The criterion that employs the first $n$ PT moments is referred to as the $p^\top_n$-PPT criterion \cite{elben2020mixed, neven2021symmetry,YuIG21}. For example, the $p^\top_3$-PPT criterion states that
	\begin{equation}
		\mathrm{if}\ \rho \ \mathrm{is \ separable},\quad \mathrm{then}\ p^{\top}_3 \le \left(p^{\top}_2\right)^2.
	\end{equation}
	In practice,  $p_n^{\top}$ can be estimated efficiently  using  randomized measurements \cite{elben2023randomized}.

	Likewise, we can devise protocols for certifying the  Schmidt number using the following sequence of \emph{k-reduction moments}:
	\begin{align}\label{eq:moment_sequence_rk}
		Q \equiv (q_n)_{n\in\bbN_0},\quad
		q_n \equiv \Tr[\caR_k(\rho)^n].
	\end{align}
	According to \thref{thm:pure_state}, if $\rho$ is a pure state with Schmidt vector $\blambda = (\lambda_0,\lambda_1,\ldots,\lambda_{d-1})$, then
	\begin{equation}\label{eq:qn_analytic}
		q_n = \Tr[\Omega_k(\blambda)^n] + (d_\rmB - 1) k^n \sum_{i=0}^{d-1}\lambda_i^n,
	\end{equation}
	where $\Omega_k(\blambda)$ is defined in \eref{eq:Omegaklambda}. 
	We can certify that  $\sn(\rho) > k$ if this moment sequence satisfies certain conditions as detailed below.
	
	\subsection{The moment method}

	Here  we recapitulate several standard results on the moment method \cite{schmudgen2017moment}.
	Denote by $\bbN$ the set of natural numbers (positive integers) and by $\bbN_0$  the set of nonnegative integers.

	\emph{The moment problem} concerns the following question: given a real sequence $S=(s_n)_{n\in \bbN_0}$ and a closed subset $\bbS\subset\bbR$, when does there exist a Radon measure $\mu$ such that  
	\begin{align}
		s_n=\int_{\bbS}x^n d\mu(x)\quad \forall n\in\bbN_0. 
	\end{align}
	Here the integration can be replaced by a summation when the measure $\mu$ is supported on discrete points.  If such a measure exists, then $\mu$ is called a \textit{representing measure} of $S$. The moment method refers to a systematic approach for solving the moment problem. The sequence $S$ is called an \emph{$[a,b]$-moment sequence} if $\bbS = [a,b]$.  Accordingly, the moment problem is called an $[a,b]$-moment problem. When the sequence is finite, the corresponding moment problem is called a \emph{truncated $[a,b]$-moment problem}. Given $N,N_1,N_2\in \bbN_0$ with $N_2 \ge N_1$, we shall use the following notation to represent a finite sequence truncated from  an infinite sequence $S$:
	\begin{align}
		S_{N_1, N_2} \equiv (s_{N_1}, s_{N_1 + 1}, \ldots, s_{N_2}),\quad 
		S_N \equiv S_{0,N}.
	\end{align}
	The arithmetic operations between sequences are defined by the corresponding operations on each entry:
	\begin{equation}
		(S \pm S')_i \equiv S_i \pm S'_i.
	\end{equation}

	An important tool for solving the moment problem is the \emph{Hankel matrix}. Given a nonnegative integer $n$ and a finite sequence $S_{2n}$, the Hankel matrix $H(S_{2n})$ is defined as follows:
	\begin{equation}
		\left[H(S_{2n})\right]_{ij} \equiv s_{i+j}\quad i,j = 0,1,\ldots,n.
	\end{equation}
	By definition, $H(S_{2n})$ is an $(n+1)\times(n+1)$ real symmetric matrix. 
	The positivity of a Hankel matrix is determined by the corresponding moment sequence. The following lemma is a standard result on the moment problem. 
	
	\begin{lemma}[Theorems 10.1 and 10.2 in \cite{schmudgen2017moment}]\label{lem:std}
		Suppose $N\in\bbN$ is even. Then $S_{N}$ is a truncated $[a,b]$-moment sequence iff
		\begin{equation}\label{eq:even_case}
			H(S_{N}) \geq 0,\quad  H(\overline{S}_{N-2}) \geq 0,
		\end{equation}
		where
		\begin{equation}
			\overline{S}_N \equiv (a+b)S_{1,N+2} - S_{2,N+3} - ab S_{N+1}.
		\end{equation}
		Suppose $N\in \bbN$ is odd. Then $S_{N}$ is a truncated $[a,b]$-moment sequence iff 
		\begin{equation}\label{eq:odd_case}
			H(S_{1,N} - aS_{N-1}) \geq 0,\;\; 
			H(b S_{N-1} - S_{1,N}) \geq 0.
		\end{equation}
	\end{lemma}
	
	\subsection{\label{sec:main_criterion}  \texorpdfstring{$k$}{empty}-reduction moment criteria}
	
	To formulate the certification problem as a moment problem, we need to delve into the moment sequence $Q=(q_n)_{n\in \bbN_0}$ generated by the moments of the $k$-reduced operator $\caR_k(\rho)$ as defined in \eref{eq:moment_sequence_rk}. Thanks to  \coref{coro:spectrum}, this sequence can  be regarded as a $[-1,k]$-moment sequence, so $Q_N$ is a truncated $[-1,k]$-moment sequence. If in addition $\caR_k(\rho) \ge 0$, then $Q_N$ is a truncated $[0,k]$-moment sequence. Now, it remains to address the technical problem: given a truncated $[-1,k]$-moment sequence, when can we conclude that it is not a truncated $[0,k]$-moment sequence?  
	
	Define 
	\begin{equation}\label{eq:BN}
		B_N[\rho,k] \equiv \begin{cases}
			H(Q_{1,N}), & N \text{ odd},\\
			H(kQ_{1,N-1} - Q_{2,N}), & N \text{ even},
		\end{cases} 
	\end{equation}
	which can be abbreviated as $B_N$ if there is no danger of confusion. The following theorem proved in \aref{app:proofmain} is the basis of our approach for certifying the entanglement dimensionality. 
	\begin{theo}\label{thm:main}
		Suppose $\rho \in \caS(\caH_{\rmA\rmB})$ and $\SN(\rho) \le k$. Then $ B_N[\rho,k] \geq 0$ for all $N$.
	\end{theo}
	
	\begin{remark}
		Note that the condition $B_N \ge 0$ comes directly from the standard result \lref{lem:std}, but only contains the second condition in \eref{eq:even_case} and the first one in \eref{eq:odd_case}. The condition $H(Q_N) \ge 0$ holds for any moment sequence of the form in \eref{eq:moment_sequence_rk}, and the condition $H(kQ_{N-1} - Q_{1,N}) \ge 0$ holds automatically because $\caR_k(\rho) \le kI$ by \coref{coro:spectrum}.
	\end{remark}
	
	If $B_N \not\geq 0$ for some positive integer $N$, then we can conclude that $Q_N$ is not a truncated $[0,k]$-moment sequence, which means $\caR_k(\rho) \not\geq 0$ and $\SN(\rho) > k$. In this way, we can  construct a sequence of criteria for certifying the Schmidt number. The \emph{$N$-th order moment-based $k$-reduction criterion} can be formulated as follows:
	\begin{equation}\label{eq:k_moment_criteria}
		\text{ if } \SN(\rho) \le k, \quad \text{then } B_N[\rho,k] \geq 0.
	\end{equation}
	Note that $k$-reduction moments are invariant under local unitary transformations, so our criteria are also invariant under local unitary transformations, in sharp contrast with fidelity-based methods.
	
	\begin{figure}
		\begin{algorithm}[H]
			{\small
				\hspace{-146pt}\textbf{Input:} $\rho$, $N^*$, and $k$.\\
				\hspace{-76pt} \textbf{Output:} if return yes, then $\sn(\rho) \ge k$.
				\begin{algorithmic}[1]			\caption{\label{alg:detection_protocol}Certification of Schmidt number}
					\State Let $q_1 = (k-1)d_\rmB - 1$.
					\State Estimate the second moment $q_2$ of $\caR_{k - 1}(\rho)$.
					\For{$N=3,\ldots,N^*$}
					\State Estimate the $N$-th moment $q_N$
					of $\caR_{k - 1}(\rho)$.
					\State Construct the Hankel matrix $B_N[\rho,k-1]$ with $\{q_n\}_{n=1}^N$. 
					\State If $B_N[\rho,k-1] \not\geq 0$, return yes.
					\EndFor
					\State Return no.
					
				\end{algorithmic}
			}
		\end{algorithm}
		
		\begin{algorithm}[H]
			{\small
				\hspace{-146pt}\textbf{Input:} $\rho$, $N^*$, and $r$.\\
				\hspace{-95pt} \textbf{Output:} a lower bound for $\SN(\rho)$.
				\begin{algorithmic}[1]
					\caption{\label{alg:full_protocol}Optimal certification of Schmidt number by the $N^*$-th order moment}
					\State $s = 1$.
					\For{$k=2,\ldots,r$}
					\State Input $(\rho,N^*,k)$ to Algorithm \ref{alg:detection_protocol}.
					\If{output is no,}
					\State Return $s$.
					\Else{\ $s = k$.}
					\EndIf
					\EndFor
					\State Return $s$.
					
				\end{algorithmic}
			}
		\end{algorithm}
		
	\end{figure} 
	
	The conditions $B_1 \geq 0, B_2 \geq 0$ are trivial because $B_1 = q_1$ and $B_2 = kq_1 - q_2$  contain only one entry each, so we usually start with $N = 3$. The simplest odd-order condition is thus
	\begin{equation}
		B_{3}=\left(\begin{array}{cc}
			q_1 & q_2\\
			q_2 & q_3
		\end{array}\right) \geq 0,
	\end{equation}
	and the simplest even-order condition is
	\begin{equation}
		B_4= \begin{pmatrix}
			kq_1 - q_2 & kq_2 - q_3 \\
			kq_2 - q_3 & kq_3 - q_4
		\end{pmatrix} \geq 0.
	\end{equation} 
	Note that $B_N$ is a principal submatrix of $B_{N+2}$, which means $B_N\geq 0$ whenever $B_{N+2}\geq 0$, so the $(N+2)$-th order moment-based $k$-reduction criterion is usually stronger than the $N$-th order moment-based criterion.
	
	By virtue of \thref{thm:main}, we can devise  a simple algorithm (Algorithm~\ref{alg:detection_protocol}) for certifying whether $\SN(\rho)\geq k$. Here we set a truncation value for $N$ because 
	it is in general more difficult to determine higher-order moments accurately. On this basis, we can further devise an algorithm (Algorithm~\ref{alg:full_protocol}) for constructing the best lower bound for the Schmidt number $\SN(\rho)$ based on $k$-reduction moments. 
	
	To illustrate the detection capabilities of moment-based $k$-reduction criteria, it is instructive to consider the following family of two-qutrit pure states:
	\begin{align}\label{eq:x1x2}
		\sqrt{x_1}|00\rangle + \sqrt{x_2}|11\rangle + \sqrt{1 - x_1 - x_2}|22\rangle,
	\end{align} 
	where $x_1,x_2 > 0$ and $x_1 + x_2 < 1$; all these states have Schmidt numbers equal to 3.  \Fref{fig:triangle} illustrates the  detectable region of the $N$-th order moment-based $k$-reduction criterion with $k=2$ and $N=3,4,5,6,7$. Note that the detection capability gets stronger and stronger as $N$ increases. When $N = 7$, the moment-based criterion can certify the Schmidt numbers of almost all two-qutrit pure states.
	
	Next, we show that the $N$-th order moment-based $k$-reduction criterion is equivalent to the original $k$-reduction criterion when $N$ is sufficiently large. 
	The following theorem is proved in \aref{app:proofmain}.
	\begin{theo}\label{thm:maximal_order}
		Suppose $\rho \in \caS(\caH_{\rmA\rmB})$, $\caR_k(\rho)$ has $\chi$ distinct nonzero eigenvalues, and $N \ge 2\chi-1$. Then $B_N[\rho,k] \geq 0$ iff $\caR_k(\rho) \geq 0$.
	\end{theo}
	\Thref{thm:maximal_order} applies to arbitrary quantum states in $\caS(\caH_{\rmA\rmB})$.
	See \sref{sec:PureBipartite} for additional results on  pure states. 
	
	\begin{figure}
		\centering
		\includegraphics[width = 0.45\textwidth]{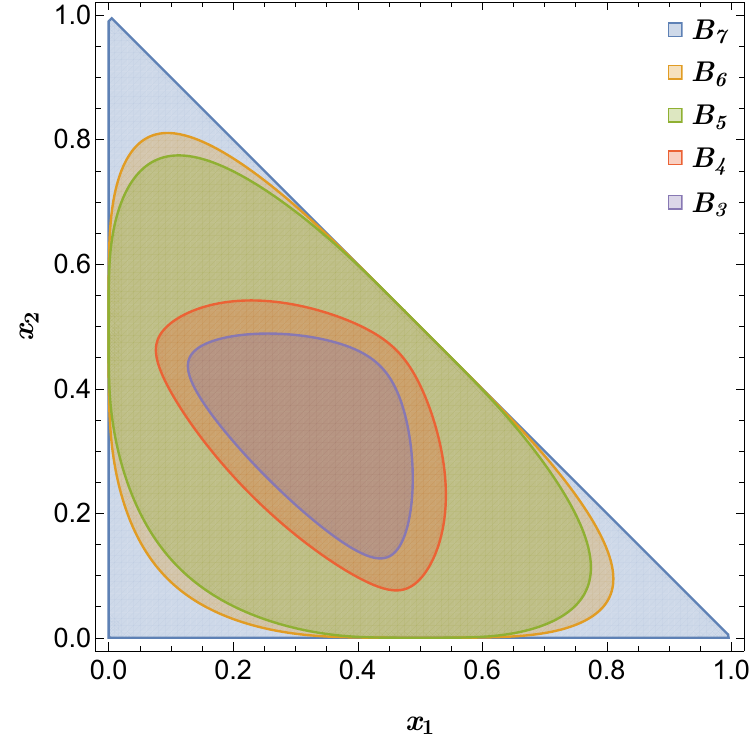}
		\caption{\label{fig:triangle}Detectable regions of $k$-reduction moment criteria with $k=2$. Each point represents a state with Schmidt number $3$ in the form of \eref{eq:x1x2}. The region of states whose Schmidt numbers can be detected by the $N$-th order moment-based criterion (that is, the set of states satisfying $B_N[\psi,2] \not\geq 0$) is illustrated in the plot. The higher the order is, the larger the detectable region becomes. When $N=7$, almost all states can be detected.}
	\end{figure}
	
	\subsection{\label{sec:protocol}The third order moment-based criterion}
	
	In this subsection, we provide more details on the third order moment-based $k$-reduction  criterion, that is, $N=3$. Note that $q_1 = kd_\rmB-1 > 0$, so $B_3\geq 0$ iff $\det(B_3[\rho,k]) \ge 0$. Accordingly, the criterion in \eref{eq:k_moment_criteria} can be simplified as follows:
	\begin{equation}\label{eq:B3}
		\textrm{ if } \SN(\rho) \le k, \quad \textrm{then }  \det(B_3[\rho,k]) \ge 0.
	\end{equation}
	Define
	\begin{equation}\label{eq:pnant2}
		\begin{gathered}
			p_n \equiv \Tr(\rho^n),\quad  a_n \equiv \Tr(\rho_{\rmA}^n),\\
			t_2 \equiv \Tr[\rho_{\rmA}\Tr_{\rmB}(\rho^2)].
		\end{gathered}
	\end{equation}
	Then  $\det(B_3[\rho,k])$ can be expressed as a polynomial of $k$ as follows:
	\begin{equation}\label{eq:frk}
		\det(B_3[\rho,k])= \beta_4 k^4 + \beta_3 k^3 + \beta_2 k^2 + \beta_1 k + \beta_0,
	\end{equation}
	where 
	\begin{equation}
		\begin{aligned}
			\beta_4 &= d_\rmB^2(a_3 - a_2^2),\\
			\beta_3 &= -4d_\rmB(a_3 - a_2^2),\\
			\beta_2 &= d_\rmB(3t_2 - 2a_2 p_2) - 4 a_2^2 + 3a_3,\\
			\beta_1 &= 4a_2 p_2 - d_\rmB p_3 - 3 t_2,\\
			\beta_0 &= p_3 - p_2^2.
		\end{aligned}
	\end{equation}
	Hence, to apply the criterion in \eref{eq:B3}, it remains to estimate the moments $p_2, p_3, a_2, a_3, t_2$. 
	
	\subsection{\label{sec:PureBipartite} \texorpdfstring{$k$}{empty}-reduction moment criteria for pure states}
	
	Recall that the $k$-reduction criterion can certify the Schmidt number of any pure state. Here we clarify the performance of moment-based $k$-reduction criteria.  The following theorem proved in  \aref{app:proof_of_theorem5} shows that the $N$-th order moment-based $k$-reduction criterion can certify the Schmidt number of any pure state in $\caH_{\rmA\rmB}$ when $N\geq 4d-1$, where $d=\min\{d_\rmA,d_\rmB\}$. In addition, it is necessary to employ the $N$-th order moment-based $k$-reduction criterion with $N\geq 2d$ for certain pure states in $\caH_{\rmA\rmB}$.
	
	\begin{theo}\label{thm:pure_state_detect}
		Suppose $|\psi\rangle\in\caH_{\rmA\rmB}$ has Schmidt rank $r$ and $\tilde{r}$ distinct nonzero Schmidt coefficients denoted by $\{\ell_j\}_{j=0}^{\tilde{r}-1}$. 
		If $N \ge 4\tilde{r}-1$, then the $N$-th order moment criterion can  detect the  state $|\psi\>$ as
		\begin{equation}\label{eq: theorem6statement1}
			B_N[\psi,k] \not\geq 0\quad \forall k=1,2,\ldots, r-1.
		\end{equation}
		Let $L = (l_n)_{n\in \bbN_0}$, $l_n = \sum_{j=0}^{\tilde{r}-1}\ell_j^n$, and
		\begin{equation}\label{eq:A3}
			A_0 \equiv \begin{cases}
				H(L_{1,N}), & N \textrm{ odd},\\
				H[(r-1)L_{1,N-1} - L_{2,N}], & N \textrm{ even}.
			\end{cases}
		\end{equation}
		If instead $N \le 2\tilde{r}$ and
		\begin{equation}\label{eq:dBLBcondition}
			d_\rmB > 1 + \frac{r\tilde{r}\caN_{r-1}(\psi)}{\sigma_{\min}(A_0)},
		\end{equation}
		then the $N$-th order moment criterion  cannot detect the  state $|\psi\>$ as
		\begin{equation}\label{eq: theorem6statement2}
			B_N[\psi,r-1] \geq 0.
		\end{equation}
	\end{theo}

	When $|\psi\>=|+_r\rangle$ is the  maximally entangled state of Schmidt rank $r$, we have $\tilde{r}=1$. In this case, 
	it suffices to apply the third order moment-based $k$-reduction criterion because \thref{thm:pure_state_detect} implies the following proposition.
	
	\begin{proposition} \label{prop:mes}
		Given $|+_r\rangle$ defined in \eref{eq:rank_r_maximally_entangled} and $1 \le k < r$, we have
		$B_N[+_r,k] \not\geq 0$ for all $N \ge 3$.
	\end{proposition}

	\subsection{Certification of entanglement dimensionality of isotropic states}
	
	Given a Hilbert space $\caH_{\rmA\rmB}$ with $d_\rmA = d_\rmB = d$, an isotropic state on $\caH_{\rmA\rmB}$ is a state of the form
	\begin{equation}\label{eq:isotropic_state}
		\rho_F \equiv \frac{1-F}{d^2-1}I + \frac{d^2 F-1}{d^2-1}|+_d\rangle\langle+_d|, \quad F \in [0,1].
	\end{equation}
	It has been proved that $\SN(\rho_F) = \lceil dF \rceil$, and the $k$-reduction map with $k = \lceil dF\rceil -1$ can certify the isotropic state \cite{terhal2000schmidt}. In the next proposition, we prove that $\rho_F$ can be certified by fidelity-based method as well.
	
	\begin{proposition}
		The isotropic state $\rho_F$ is $\lceil dF\rceil$-faithful.
	\end{proposition}
	\begin{proof}
		Note that
		\begin{gather}
			\sn(|+_d\>) = d \ge \lceil dF\rceil,\\
			\<+_d|\rho_F|+_d\> = F > \frac{1}{d}(\lceil dF\rceil-1), \\ \sigma_i^\downarrow(+_d) = \frac{1}{d}\quad \forall i.
		\end{gather}
		In conjunction with the sufficient condition of $k$-faithfulness in \eref{eq:kunfaithfulcondition}, the proposition is proved.
	\end{proof}
	
	Here we show that the third  order moment-based criterion is enough to certify all isotropic states as well. Meanwhile, the CM criterion and its moment-based variants in \eqsref{eq:CMM4}{eq:CMM4b} are equally effective. \Psref{prop:isotropic} and \ref{prop:cmc_isotropic} below are  proved in \asref{app:k_reduction_detect} and \ref{app:cmc_detect}, respectively.
	\begin{proposition}\label{prop:isotropic}
		Suppose $\rho_F$ is the isotropic state defined in \eref{eq:isotropic_state} and $ 1 \le k \le \lceil dF\rceil - 1$. Then $B_3[\rho_F,k] \not\geq 0$.
	\end{proposition}
	
	\begin{proposition}\label{prop:cmc_isotropic}
		Suppose $\rho_F$ is the isotropic state defined in \eref{eq:isotropic_state}, and $T$ is its CM with respect to given local operator bases that satisfy the conditions in \eref{eq:OperatorBasesCon}. Then all $d^2-1$ singular values of $T$ are equal to $(d^2 F - 1)/[d(d^2-1)]$ and 
		\begin{equation}
			\begin{aligned}
				\|T\|_1 &= dF - d^{-1}, \quad \|T\|_2^2 = \frac{(d^2 F - 1)^2}{d^2(d^2-1)},\\
				\|T\|_4^4 &= \frac{(d^2 F - 1)^4}{d^4(d^2-1)^3}.
			\end{aligned}
		\end{equation}
	\end{proposition}
	In conjunction with  the fact $\SN(\rho_F) = \lceil dF \rceil$ and \eref{eq:correlation_matrix}
	we can conclude that the CM criterion can certify the Schmidt number of any isotropic state as well. Even if only the second and fourth moments are available, we can still deduce that $\|T\|_1\geq \|T\|_2^3/\|T\|_4^2= dF - d^{-1}$. 
	Therefore, the moment-based criterion in \eref{eq:CMM4b} is equally effective, and so is the alternative in \eref{eq:CMM4}.

	%In \sref{sec:k_unfaithful}, we have introduced basic properties of the $k$-unfaithful states.
	
	\subsection{Certification of entanglement dimensionality of \texorpdfstring{$k$}{empty}-unfaithful states}

	In this subsection, we show that the $k$-reduction moment method
	can  certify entanglement dimensionality of $(k+1)$-unfaithful states.
	
	First,  we construct a family of $k$-unfaithful states on $\caH_{\rmA\rmB}$, assuming that $d_\rmA=d_\rmB=d\geq k$. Let $\chi$ be a positive integer that satisfies $1\leq \chi\leq \min\{d,2k-2\}$ and let $\pi$ and $\tau$ be two permutations on $\{0,\ldots,d-1\}$. 
	Consider the following state on $\caH_{\rmA\rmB}$:
	\begin{equation}\label{eq:special_unfaithful}
		\begin{aligned}
			\rho &=\frac{1}{2}|+_{k}\>\<+_{k}| + \frac{1}{2}|\Phi\>\<\Phi|,\\
			|\Phi\> &= \sum_{j=0}^{\chi-1}\sqrt{\phi_j}\lsp|\pi(j)\>\otimes|\tau(j)\>,
		\end{aligned}
	\end{equation}
	where $\sum_{j=0}^{\chi-1}\phi_j=1$. Note that the state $\rho$ can be regarded as a generalization of $\rho_\UF$ in \eref{eq:unfaithful}. The following proposition proved  in \aref{app:proof_of_prop_unfaithful} is helpful for constructing $k$-unfaithful states of the form \eref{eq:special_unfaithful}
	that have Schmidt number at least $k$.
	
	\begin{proposition}\label{prop:unfaithful}
		Suppose the state $\rho\in \caS(\caH_{\rmA\rmB})$ has
		the form  \eref{eq:special_unfaithful}, where $1\leq \chi\leq \min\{d,2k-2\}$,  $\pi(j) \neq \tau(j)$ for each $j = 0,\ldots,k-1$, and the following index sets
		\begin{equation}
			\begin{aligned}
				\caI_L &\equiv \{j\}_{j=0}^{k-1}\cup\{\pi(j)\}_{j=0}^{\chi-1},\\
				\caI_R &\equiv \{j\}_{j=0}^{k-1}\cup\{\tau(j)\}_{j=0}^{\chi-1}
			\end{aligned}
		\end{equation} 
		satisfy $|\caI_L|,|\caI_R| \le 2k-2$. Then $\rho$ is $k$-unfaithful.
		If in addition
		\begin{equation} \label{eq: extra_condition}
			\sum_{i=0}^{k-1}\frac{1}{1 + k\phi_{\pi^{-1}(i)}} > k-1,
		\end{equation}
		then $\caR_{k-1}(\rho)\not\ge 0$.
	\end{proposition} 
	
	Now, suppose $k\ge 3$,  $d>2k-2$, $\chi=2$, $\phi_j = 1/2$ for $j = 0,1$, and
	\begin{equation}
		\pi(j) = k-1+j,\quad \tau(j) = k-j\quad  \forall \lsp j,
	\end{equation}
	where the arithmetic is modulo $d$. Then the state $\rho$ in \eref{eq:special_unfaithful} 
	simplifies to 
	\begin{equation}
		\begin{aligned}
			\rho_{\UF,k} &\equiv \frac{1}{2}|+_k\>\<+_k| + \frac{1}{2}|\Phi^{+}_k\>\<\Phi^{+}_k|,\\
			|\Phi^{+}_k\> &\equiv \frac{1}{\sqrt{2}}(|k-1,k\> + |k,k-1\>),
		\end{aligned}
	\end{equation}
	and it satisfies all the conditions in \pref{prop:unfaithful} given that 
	\begin{equation}
		\sum_{j=0}^{k-1}\frac{1}{1 + k\phi_{j+1-k}} = k-1 + \frac{1}{1 + k/2} > k-1.
	\end{equation}
	So $\rho_{\UF,k}$  has Schmidt number  $k$ but is $k$-unfaithful.
	
	Next, we use Algorithm~\ref{alg:full_protocol} to certify the Schmidt number of the state $\rho_{\UF,k}$ with  $d_\rmA = d_\rmB = 16$. In \tref{tab:unfaithful} we summarize the $(k-1)$-reduction negativity of $\rho_{\UF,k}$ and the minimum moment order $N$ such that  $B_N[\rho_{\UF,k}, k-1] \not\ge 0$ for $k=3,4,\ldots, 10$. 
	These results highlight the versatility of  $k$-reduction moments in certifying the entanglement dimensionality of 
	$k$-unfaithful states in addition to $k$-faithful states.
	
	\begin{center}
		\begin{table}
			\caption{\label{tab:unfaithful} Certification of the entanglement dimensionality of the $k$-unfaithful state $\rho_{\UF,k}$ using $k$-reduction moments. The second column shows the $(k-1)$-reduction negativity of $\rho_{\UF,k}$, and the third column shows the minimum order $N$ such that $B_N[\rho_{\UF,k}, k-1] \not\ge 0$, which can be determined by virtue of  Algorithm~\ref{alg:full_protocol}.}
			
			\begin{tabular}{c|cc}
				\hline
				\hline
				\rule{0pt}{1em}
				$k\ $ & $\caN_{k-1}\left(\rho_{\UF,k}\right)$ & minimum order \\[0.2em]
				\hline
				\\[-0.8em]
				3 & $0.0749 $ & 7 \\
				\\[-0.8em]
				4 & $0.0449 $ & 7 \\
				\\[-0.8em]
				5 & $0.0301$ & 7 \\
				\\[-0.8em]
				6 & $0.0216$ & 7 \\
				\\[-0.8em]
				7 & $0.0163$ & 9 \\
				\\[-0.8em]
				8 & $0.0128$ & 9 \\
				\\[-0.8em]
				9 & $0.0103$ & 9 \\
				\\[-0.8em]
				10 & $0.0085$ & 9 \\
				\hline
				\hline
			\end{tabular}
			
		\end{table}
	\end{center}

	\section{\label{sec:performance}Numerical results on detection ratios}
	To complement the theoretical analysis in \sref{sec:moment}, here we provide extensive numerical results on the  detection ratios of  the $N$-th order moment-based $k$-reduction criterion in \eref{eq:k_moment_criteria} and moment-based CM criterion in \eref{eq:CMM4}. Given an ensemble $\caE$ of quantum states,  the \emph{detection ratio}  of the $N$-th order moment-based $k$-reduction criterion is defined as
	\begin{equation}\label{eq:detect_ratio}
		\text{Pr}\left\{ B_N[\rho, k] \not\geq 0 \ : \  \rho\in\caE\right\}.
	\end{equation}
	Similarly, the detection ratio of the moment-based CM criterion is defined as
	\begin{equation}
		\text{Pr}\left\{ \left(\|T\|_2^2,\|T\|_4^4\right)\not\in \caT_k : \  \rho\in\caE\right\}.
	\end{equation}
	Not surprisingly, the detection ratios may strongly depend on the ensemble $\caE$ under consideration. Nevertheless, for any given ensemble, a higher detection ratio is an indication of a higher performance. In the rest of this section we shall consider several concrete ensembles of quantum states, including the ensemble of pure states with a fixed Schmidt number,  the ensemble of Haar-random pure states with depolarizing noise, and the ensemble of mixed states with an induced measure.
	
	\subsection{\label{sec:pure_state_ensemble}Ensemble of pure states with a fixed Schmidt number}
	
	Define the ensemble 
	\begin{gather}\label{eq:pure_ensemble}
		\caE_r \equiv \left\{|\psi\rangle = \sum_{i=0}^{r-1} \sqrt{\lambda_i}|i\>_{\rmA}\otimes|i\>_{\rmB}\right\},
	\end{gather}
	where $\{|i\>_{\rmA}\},\{|i\>_{\rmB}\}$ are elements of two orthonormal  bases of $\caH_\rmA$ and $\caH_\rmB$, respectively, and $\{\lambda_i\}_{i=0}^{r-1}$ is uniformly distributed in the interior of a $(r-1)$-dimensional probability simplex [cf. \eref{eq:simplex}],
	which is also known as the Dirichlet distribution. By definition all states in $\caE_r$ have Schmidt numbers equal to $r$. Here it is not necessary to consider local unitary transformations  because the moment-based  $k$-reduction  criteria and  CM criterion are both invariant under local unitary transformations.  In numerical simulations we choose $d=d_\rmA=d_\rmB$.
	
	\begin{figure}
		\centering
		\includegraphics[width=0.4\textwidth]{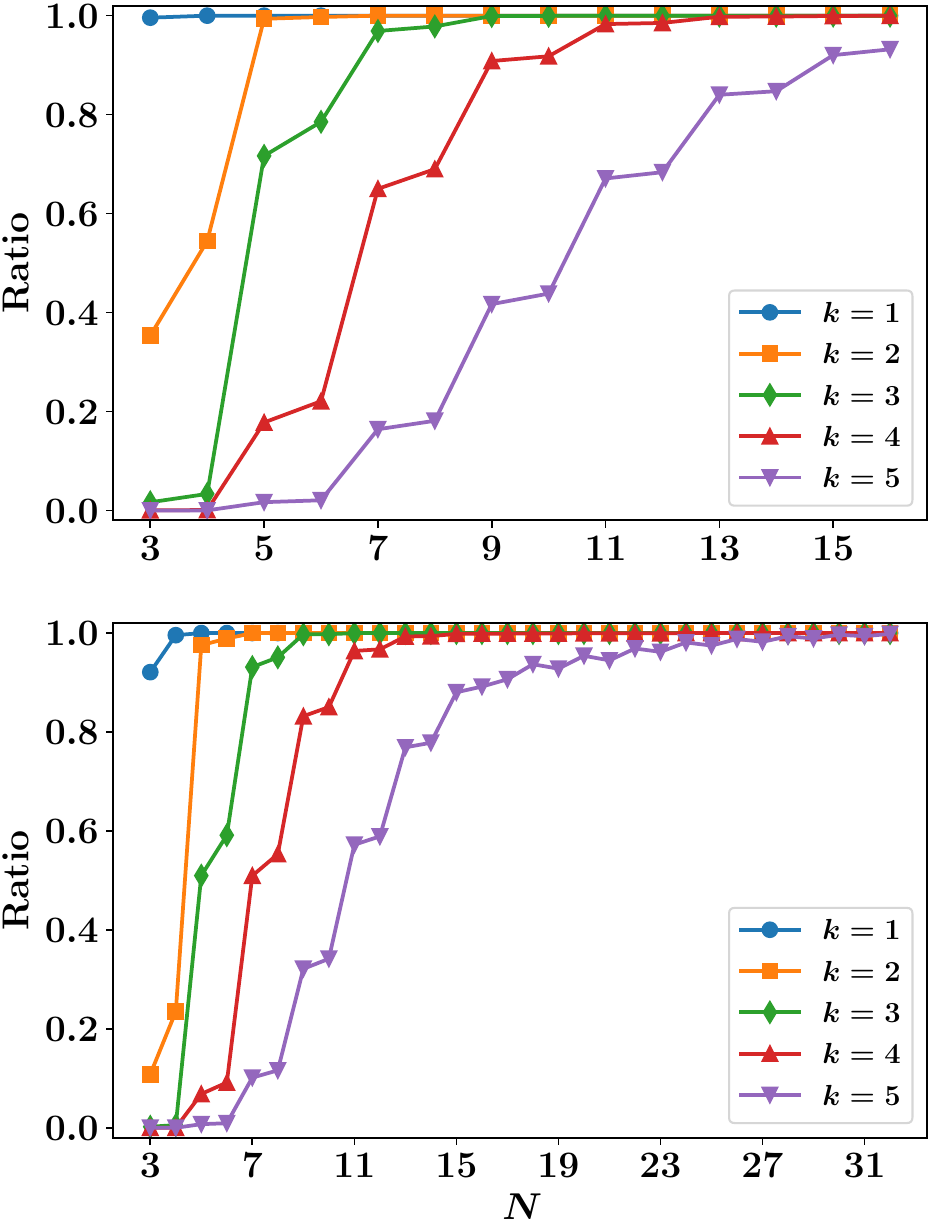}\\
		\caption{\label{fig:bounded_k}Detection ratio of the $N$-th order moment-based $k$-reduction criterion in \eref{eq:k_moment_criteria} for the pure state ensemble $\caE_6$ in \eref{eq:pure_ensemble} as a function of $N$.
			Here $d=8$ (upper plot) and $d=16$ (lower plot).  The detection ratio increases monotonically with $N$, but decreases monotonically with $k$. }
	\end{figure}
	
	\begin{figure}
		\centering
		\includegraphics[width=0.4\textwidth]{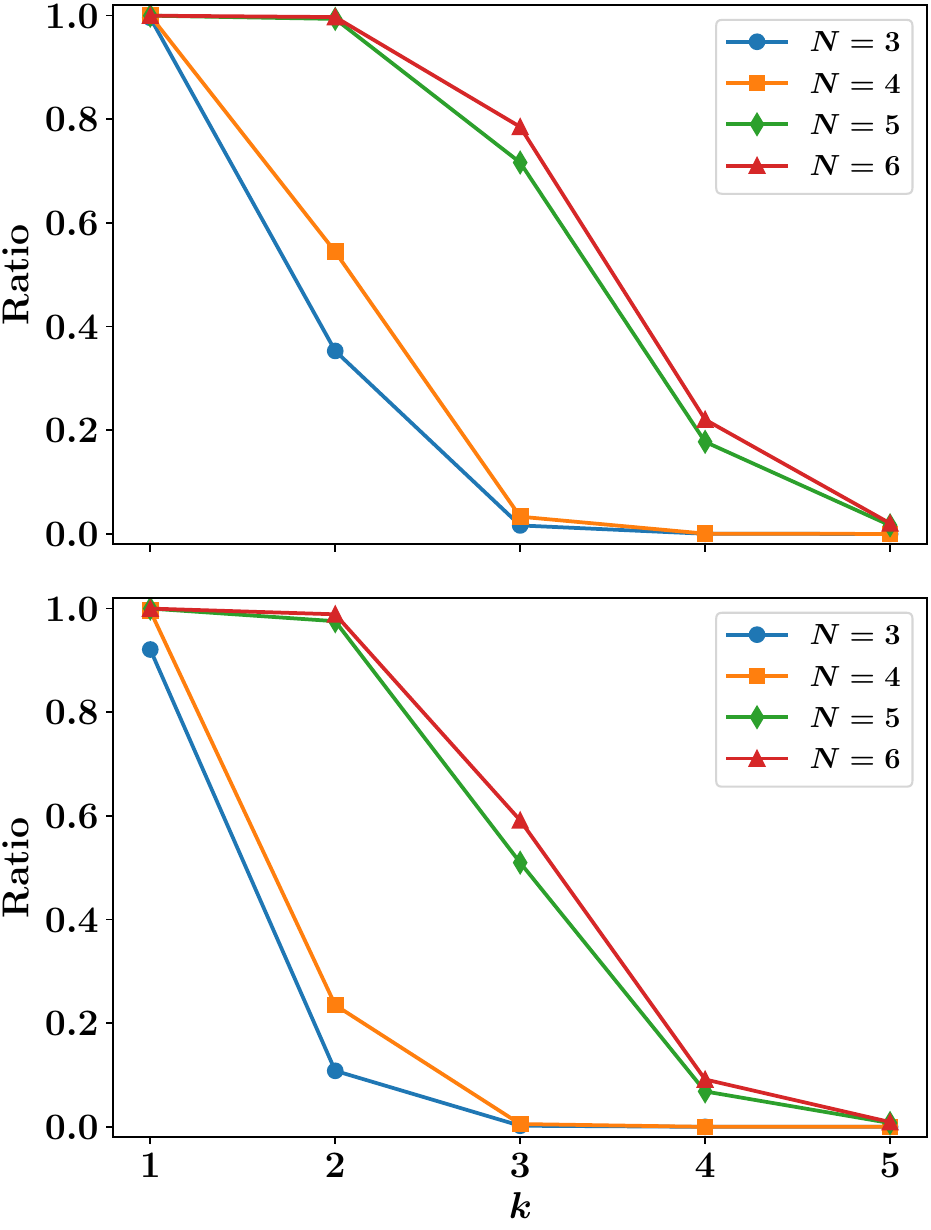}\\
		\caption{\label{fig:bounded_k2}Detection ratio of the $N$-th order moment-based $k$-reduction criterion  for the pure state ensemble $\caE_6$  as a function of $k$. 
			Here $d=8$ (upper plot) and $d=16$ (lower plot) as in \fref{fig:bounded_k}.  }
	\end{figure}
	
	\begin{figure}
		\centering  
		\includegraphics[width=0.45\textwidth]{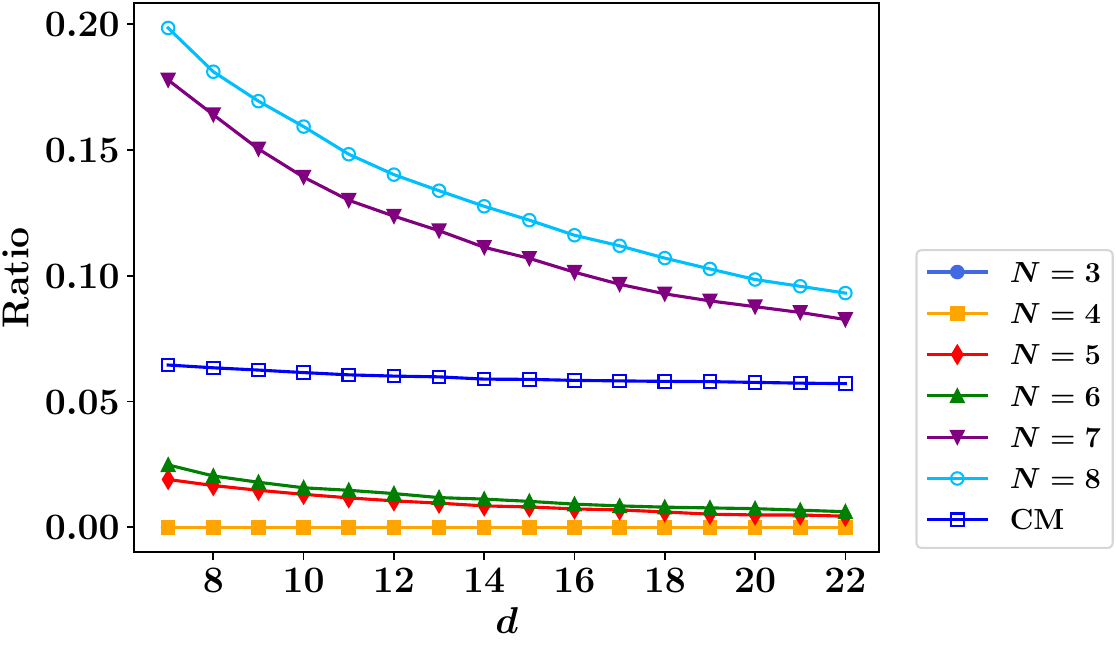}
		\caption{\label{fig:decay}Detection ratios of the $N$-th order moment-based $k$-reduction criterion in \eref{eq:k_moment_criteria} and moment-based CM criterion in \eref{eq:CMM4} with $k=5$ for the ensemble $\caE_6$ in  \eref{eq:pure_ensemble} with $7\leq d\leq 22$. The $N$-th order moment-based $k$-reduction criterion has a larger detection ratio than the moment-based CM criterion when $N\geq 7$. 
		}
	\end{figure}
	
	Figures \ref{fig:bounded_k} and \ref{fig:bounded_k2} illustrate the detection ratio  of the $N$-th order moment-based $k$-reduction criterion in \eref{eq:k_moment_criteria} as a function of $N$ and $k$ for the ensemble $\caE_6$ with two different local dimensions, namely, $d=8$ and $d=16$.  As expected the detection ratio  decreases monotonically with $k$, but increases monotonically with the order $N$ and eventually approaches~1 when $N$ is sufficiently large.
	
	Now, we compare the detection ratios of the $N$-th order moment-based $k$-reduction criterion in \eref{eq:k_moment_criteria} and moment-based CM criterion in
	\eref{eq:CMM4}. \Fref{fig:decay} illustrates the results on the ensemble $\caE_6$ as functions of the order $N$ and local dimension $d$, where $k=5$.
	According to this figure,  the detection ratio of the  moment-based CM criterion is almost independent of the local dimension, 
	but the  detection ratios of moment-based $k$-reduction criteria tend to decrease with the local dimension. Nevertheless, the $N$-th order moment-based $k$-reduction criterion has a larger detection ratio when $N$ is sufficiently large.

	\subsection{Ensemble of Haar-random pure states with depolarizing noise}
	
	\begin{figure}
		\centering
		\includegraphics[width = 0.4\textwidth]{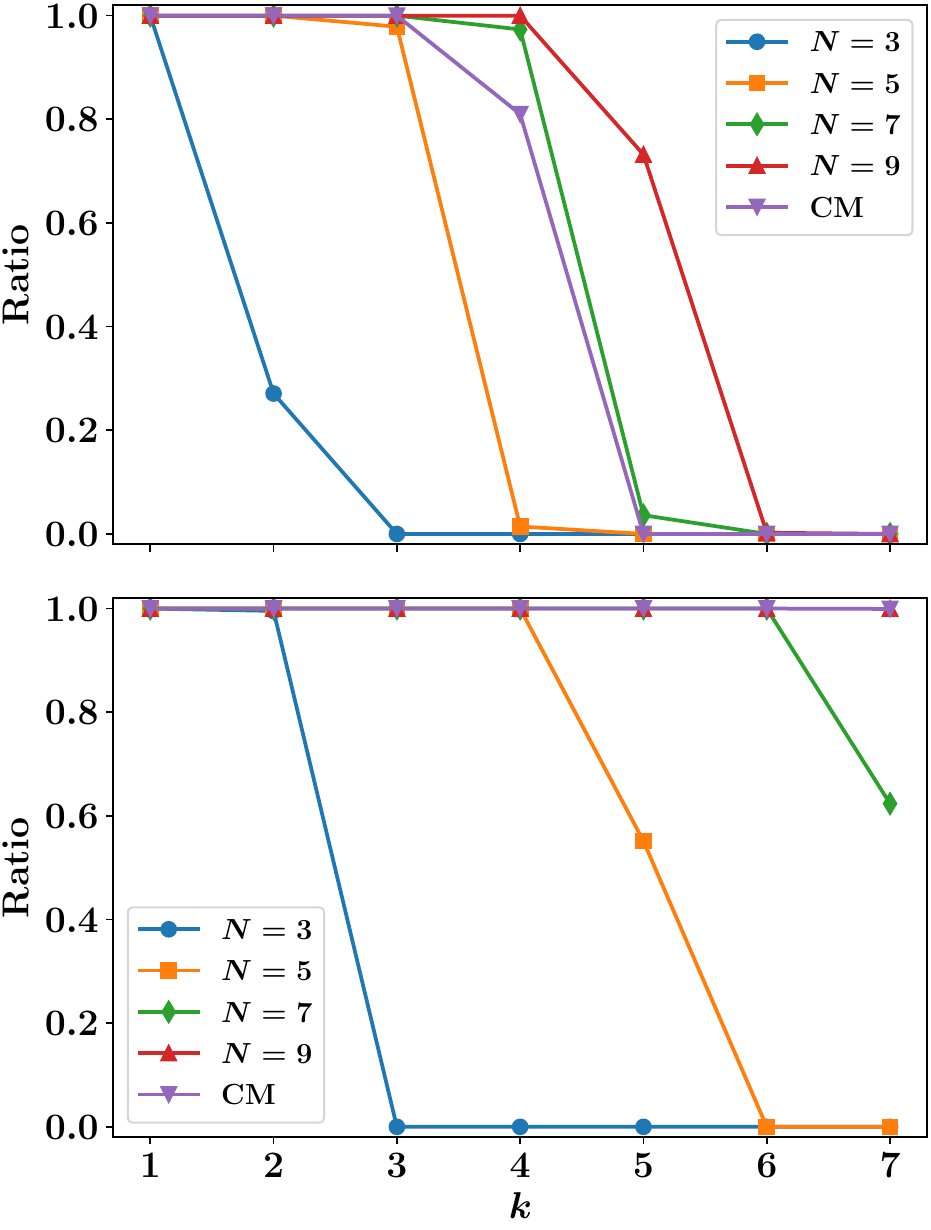}
		\caption{\label{fig:depolarized_e0}Detection ratios of different $k$-reduction moment criteria in \eref{eq:k_moment_criteria} and moment-based CM criterion in \eref{eq:CMM4} for the Haar-random ensemble  $\caE_{\Haar,0}$ in \eref{eq:Haar_ensemble}. Here $d=8$ (upper plot) and $d=16$ (lower plot) as in \fref{fig:bounded_k}. }
	\end{figure}
	
	\begin{figure}
		\centering
		\includegraphics[width = 0.4\textwidth]{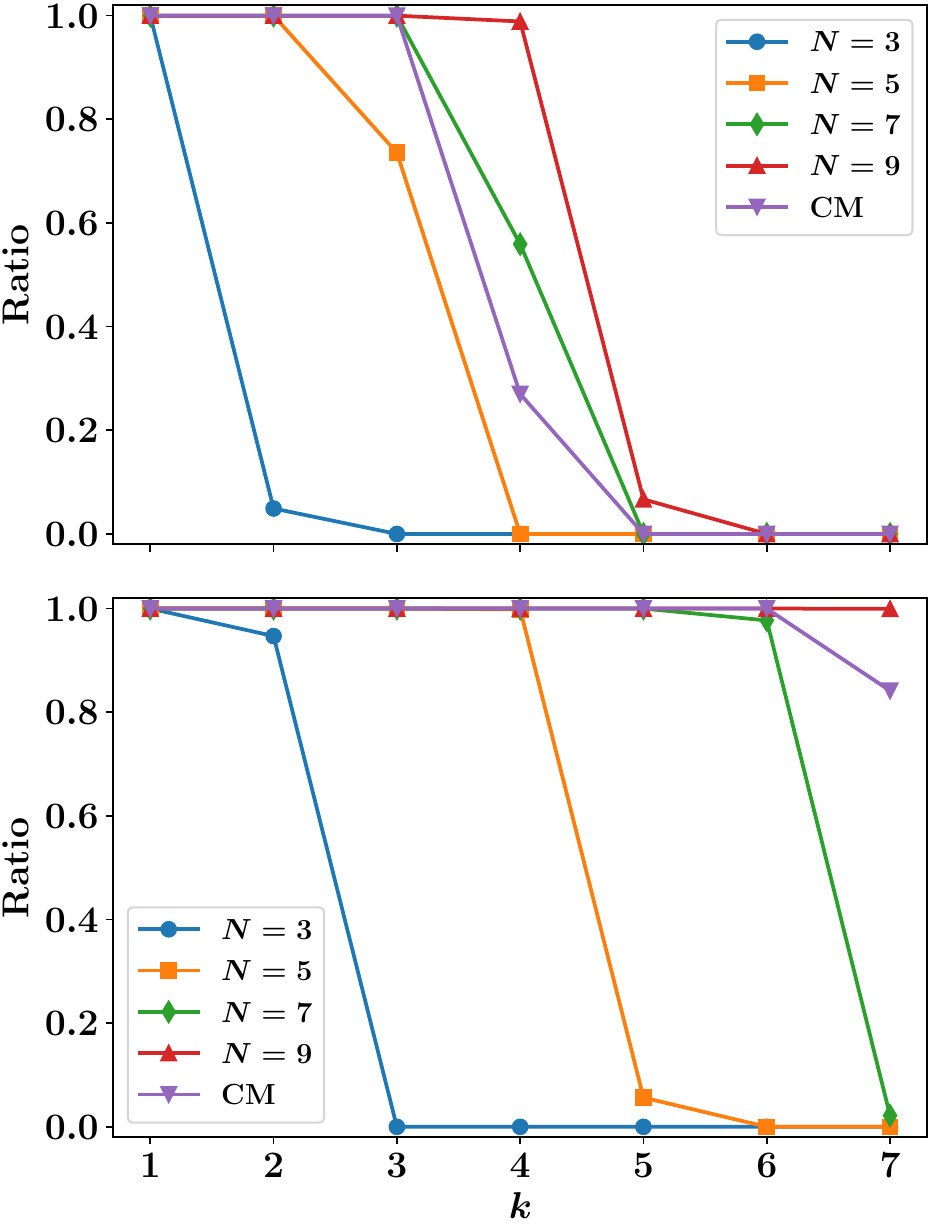}\caption{\label{fig:depolarized_e0.1}Detection ratios of different $k$-reduction moment criteria in \eref{eq:k_moment_criteria} and moment-based CM criterion in \eref{eq:CMM4} for the Haar-random ensemble $\caE_{\Haar,0.1}$ in \eref{eq:Haar_ensemble}. Here $d=8$ (upper plot) and $d=16$ (lower plot) as in \fref{fig:bounded_k}. }
	\end{figure}
	
	\begin{figure}
		\centering
		\includegraphics[width = 0.4\textwidth]{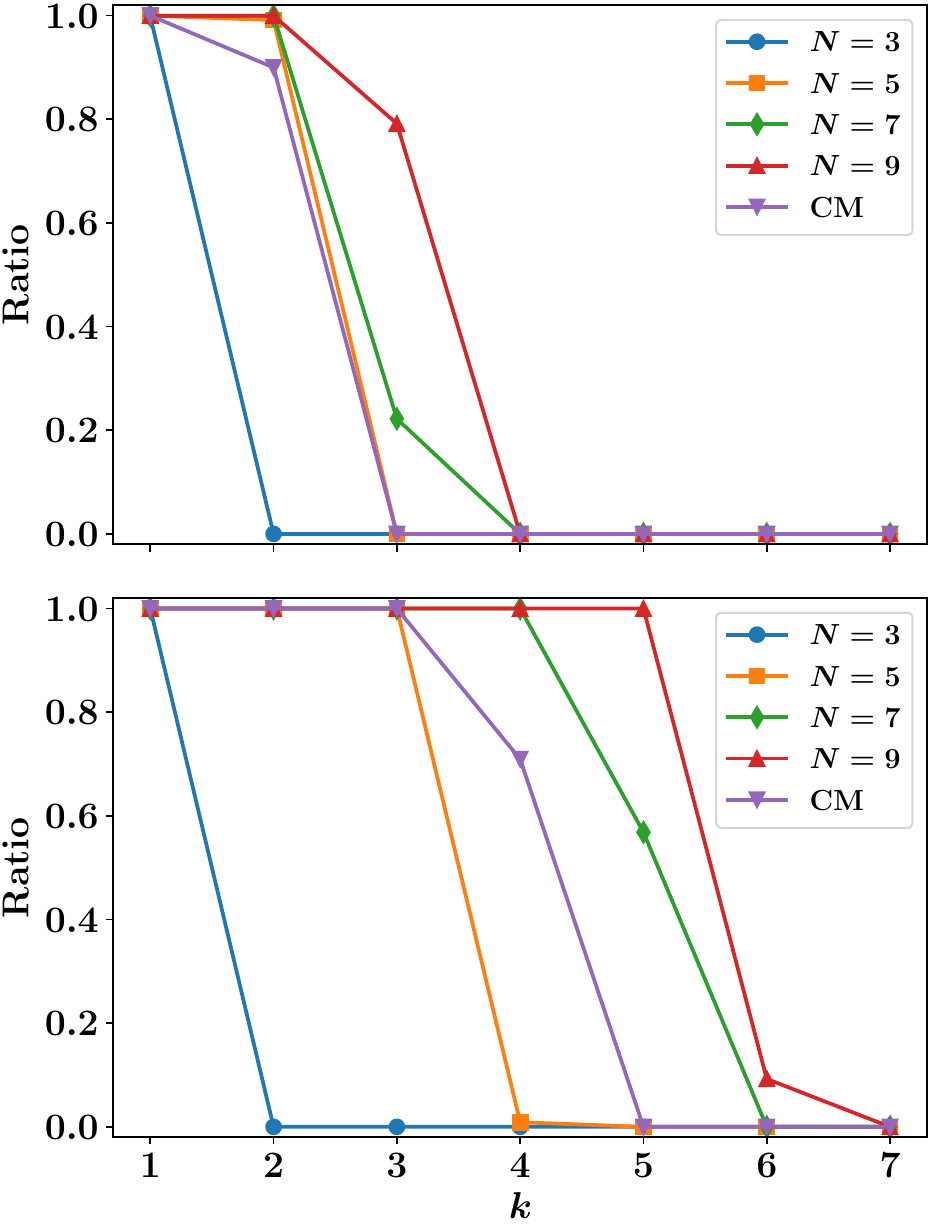}\caption{\label{fig:depolarized_e0.5}Detection ratios of different $k$-reduction moment criteria in \eref{eq:k_moment_criteria} and moment-based CM criterion in \eref{eq:CMM4} for the Haar-random ensemble $\caE_{\Haar,0.5}$ in \eref{eq:Haar_ensemble}. Here $d=8$ (upper plot) and $d=16$ (lower plot) as in \fref{fig:bounded_k}. }
	\end{figure}
	
	Next, we turn to the ensemble $\caE_{\Haar, \varepsilon}$ of Haar-random states with depolarizing noise as defined in  \eref{eq:Haar_ensemble}, where the parameter $\varepsilon$ characterizes the noise strength. \Fsref{fig:depolarized_e0}-\ref{fig:depolarized_e0.5}
	illustrate the 
	detection ratios of the $N$-th order moment-based $k$-reduction criterion in \eref{eq:k_moment_criteria} and moment-based CM criterion in
	\eref{eq:CMM4} for  the ensemble $\caE_{\Haar, \varepsilon}$ with two different local dimensions and three different noise strengths. Not surprisingly,  the detection ratios increase monotonically with $N$ but decrease monotonically with $k$ and $\varepsilon$. When $d=8$ and $\varepsilon=0, 0.1$ for example, the $7$-th order moment-based $k$-reduction criterion has a larger detection ratio than the moment-based CM criterion, but the situation is different when $d=16$.

	\subsection{Ensemble of mixed states with an induced measure}
	
	Finally, we consider the ensemble $\caE_{D,K}$ of mixed states with an induced measure, which is defined in \sref{sec:RMCCMCcom}.
	In numerical simulations we choose  $D=d^2=16^2$ and $K=2,3,4,5,6$. \Tsref{tab:k_induce_1} and \ref{tab:k_induce_2} show the detection ratios of states with 
	various lower bounds for the Schmidt numbers. The detection ratios in  \Tref{tab:k_induce_1} are determined by moment-based $k$-reduction criteria described by Algorithm \ref{alg:full_protocol} with maximal order $N^*=16$ and $r=16$. The detection ratios in \Tref{tab:k_induce_2} are determined by moment-based CM criterion in \eref{eq:CMM4}.
	From the numerical results, we can see that the $k$-reduction moment criteria
	are weaker than moment-based CM criterion for $K=4,5,6$, but becomes stronger for $K=3$. 
	
	\begin{center}
		\begin{table}
			
			\caption{\label{tab:k_induce_1}
				Detection ratio of the $N$-th order moment-based $k$-reduction criterion in \eref{eq:k_moment_criteria} with $N=16$ (corresponding to Algorithm \ref{alg:full_protocol} with $N^*=16$) for the ensemble $\caE_{D,K}$ with  $d_\rmA = d_\rmB = 16$ and $D = d_\rmA d_\rmB = 256$.
			}
			\begin{tabular}{c|ccccc}
				\hline
				\hline
				\backslashbox{$k$}{\raisebox{-0.7ex}{$K$}} & 2\quad & 3\quad & 4\quad & 5\quad & 6\quad \\
				\hline
				\\[-0.8em]
				2 & 1 & 1 & 1 & 1 & 1\\
				\\[-0.8em]
				3 & 1 & 1 & 1 & 1 & 0.9559\\
				\\[-0.8em]
				4 & 1 & 1 & 1 & 0.0049 & 0\\
				\\[-0.8em]
				5 & 1 & 1 & 0.0001 & 0 & 0 \\
				\\[-0.8em]
				6 & 1 & 0.0146 & 0 & 0 &  0\\
				\\[-0.8em]
				7 & 0.9999 & 0 & 0 & 0 &  0\\
				\\[-0.8em]
				8 & 0.0441 & 0 & 0 & 0 &  0\\
				\hline
				\hline
			\end{tabular}
			\caption{\label{tab:k_induce_2}
				Detection ratio of the moment-based CM criterion in \eref{eq:CMM4} for the ensemble $\caE_{D,K}$. Here $d_\rmA = d_\rmB = 16$ and $D = d_\rmA d_\rmB = 256$ as in \tref{tab:k_induce_1}. 
			}
			\begin{tabular}{c|ccccc}
				\hline
				\hline
				\backslashbox{$k$}{\raisebox{-0.7ex}{$K$}} & 2 & 3 & 4 & 5\quad & 6 \quad\\
				\hline
				\\[-0.8em]
				4 & 1 & 1 & 1 & 1 & 1 \\
				\\[-0.8em]
				5 & 1 & 1 & 0.9337 & 0 & 0 \\
				\\[-0.8em]
				6 & 0.9996 & 0.0016 & 0 & 0 &  0\\
				\\[-0.8em]
				7 & 0.0059 & 0 & 0 & 0 &  0\\
				\hline
				\hline
			\end{tabular}
			
		\end{table}
	\end{center}

	\section{\label{sec:moment_estimation} Estimation of \texorpdfstring{$k$}{empty}-reduction moments}
	
	To certify the  Schmidt number using moment-based $k$-reduction criteria, we need to estimate the $k$-reduction moments $q_n = \caR_k(\rho)^n$ that appear in the entries of the Hankel matrices. Simple analysis shows that these moments can be expressed as follows:
	\begin{equation}
		\Tr\left[\text{poly}(\rho,\rho_{\rmA}\otimes I_{\rmB})\right].
	\end{equation}
	In this section, we introduce two approaches for estimating these moments.

	\subsection{\label{sec:moment_estimation_cs}Moment estimation by randomized measurements}
	Here we discuss moment estimation based on 
	the classical shadow \cite{huan2020}, which is a prototypical method based on randomized measurements. This method consists of two steps, namely, data acquisition and prediction. 
	
	Suppose $\rho\in\caS(\caH_{\rmA\rmB})$, and $\dim(\caH_{\rmA\rmB}) = D$. In the phase of data acquisition, we first randomly sample a unitary operator  $U$ from a given unitary ensemble $\caU$ on the Hilbert space $\caH_{\rmA\rmB}$ and apply it to the state $\rho$ under consideration. Then we perform a measurement in the computational basis and record the measurement outcome $b$. Based on $U$ and $b$, we can create a classical shadow of the unknown state as follows:
	\begin{equation}
		\hr = \caM^{-1}_{\mathrm{cs}}\left(U^\dag|b\rangle\langle b|U\right).
	\end{equation}
	Here $\caM_{\mathrm{cs}}$ is the measurement channel, which depends on the ensemble $\caU$, and $\caM^{-1}_{\mathrm{cs}}$ is the inverse of $\caM_{\mathrm{cs}}$ and is known as the reconstruction map. If $\caU$ corresponds to the global Clifford group, then the above equation reduces to
	\begin{align}
		\hr=(D+1)U^\dag |b\>\<b|U -I.
	\end{align}
	
	In the prediction phase, we need to estimate  the expectation value of a given observable $O$ (generalization to more observables is straightforward). Given $M$ samples of the  classical shadow we can  construct an estimator based on the empirical mean as follows:
	\begin{equation}
		\hat{o} = \frac{1}{M}\sum_{m=1}^{M} \Tr(\hr_m O).
	\end{equation}
	The variance of this estimator is controlled by the \emph{shadow norm} introduced in \cite{huan2020}. The worst-case variance is minimized when the ensemble $\caU$ forms a unitary 3-design. 
	Alternatively, we can employ the median-of-means to suppress the probability of large deviation. 
	
	In addition, 
	the  classical shadow method can be applied to estimating quadratic functions of $\rho$ that have the form  $o=\Tr[O(\rho\otimes \rho)]$, where $O$ is now an operator on $\caH^{\otimes 2}_{\rmA\rmB}$. Here is a specific estimator \cite{huan2020}:
	\begin{equation}
		\hat{o}=\frac{1}{M(M-1)}\sum_{m=1}^{M}\sum_{n=1,n\neq m}^M\Tr[O (\hr_{m}\otimes \hr_{n})].
	\end{equation} 
	For example, the purity $\Tr(\rho^2)$ can be estimated by
	\begin{align}
		&\quad \frac{1}{M(M-1)}\sum_{m=1}^{M}\sum_{n=1,n\neq m}^M\Tr[\bbW_{(12)}(\hr_m\otimes\hr_n)]\nonumber\\
		&= \frac{1}{M(M-1)}\sum_{m=1}^{M}\sum_{n=1,n\neq m}^M \Tr(\hr_m \hr_n),
	\end{align}
	where $\bbW_{(12)}$ is the swap operator on $\caH^{\otimes 2}_{\rmA\rmB}$. Other nonlinear functions can be estimated in a similar way.

	Next, we turn to the estimation of reduction moments that are required for certifying the Schmidt number. The following theorem clarifies the sample complexities of estimating the moments involved in the third order moment-based $k$-reduction criterion; see \aref{app:haar_random} for a proof. 
	\begin{theo}\label{thm:sample_complexity}
		Suppose $\rho\in \caS(\caH_{\rmA\rmB})$, $d_\rmA = d_\rmB = \sqrt{D}$, and we are given access to a unitary 3-design. Then the sample complexity of estimating the moments $p_2, p_3, a_2, a_3, 
		t_2$ defined in \eref{eq:pnant2}
		up to accuracy $\caO(1)$ using the classical shadow method is $\caO(D)$.
	\end{theo}
	
	Note that the randomized measurements required to apply our certification protocol can be realized using a unitary 3-design, say the  global Clifford group \cite{webb2015clifford,zhu2017multiqubit}, which is quite appealing to practical applications. By contrast, to apply the CM protocol \cite{liu2023characterizing,wyderka2023probing}, it is necessary to use unitary 4-designs on $\caH_\rmA$ and $\caH_\rmB$; moreover, a unitary 8-design is required to 
	bound the variance.
	
	\subsection{\label{sec:swap_test}Moment estimation by permutation tests}
	
	Next, we show that the sample complexity of estimating the $k$-reduction moments can be reduced to $\caO(1)$ if permutation tests based on collective measurements are accessible. Suppose $\rho_1, \rho_2, \ldots, \rho_N\in \caS(\caH_{\rmA\rmB})$; then $\bigotimes_{n=1}^N \rho_n\in \caS(\caH^{\otimes N}_{\rmA\rmB})$. Let $\bbW_\pi$ be the unitary operator tied to the standard cyclic permutation.  Add an ancillary qubit prepared in the state $|+\>$ and apply the controlled cyclic shift gate $|0\rangle\langle0|\otimes I + |1\rangle\langle 1|\otimes\bbW_\pi$ on the entire system then yields
	\begin{equation}
		\begin{aligned}
			&\quad \frac{|0\rangle\langle 0|}{2} \otimes \bigotimes_{n=1}^N \rho_n + \frac{|1\rangle\langle 1|}{2}\otimes \bbW_\pi \bigotimes_{n=1}^N \rho_n \bbW^{-1}_\pi\\
			& + \frac{|0\rangle\langle 1|}{2} \otimes \bigotimes_{n=1}^N \rho_n \bbW^{-1}_\pi + \frac{|1\rangle\langle 0|}{2} \otimes \bbW_\pi\bigotimes_{n=1}^N \rho_n.
		\end{aligned}
	\end{equation}

	The reduced state on the ancilla qubit reads
	\begin{equation}
		\frac{I}{2} + \frac{1}{2}\Tr(\rho_1 \rho_2\cdots  \rho_N) X.
	\end{equation}
	Measure the ancilla qubit in the Pauli $X$ basis, then the probabilities of obtaining outcomes $\pm$ are given by 
	\begin{equation}
		P_\pm = \frac{1}{2} \pm  \frac{1}{2}\Tr(\rho_1 \rho_2\cdots  \rho_N),
	\end{equation}
	which means 
	\begin{equation}
		P_+ - P_- = \Tr(\rho_1 \rho_2\cdots  \rho_N).
	\end{equation}
	The full protocol is termed as the \textit{permutation test}, which leads to a simple recipe for estimating the $N$-th moment $\Tr(\rho^N)$ of a general quantum state $\rho\in \caS(\caH_{\rmA\rmB})$, and the sample complexity is independent of the system size. This method can readily be applied to estimating all the relevant parameters featuring in the moment-based $k$-reduction criteria. 
	
	In practice, the controlled cyclic shift gate can be implemented as follows. A controlled-SWAP gate (also known as the Fredkin gate) can be constructed using Toffoli and CNOT gates (see \cite{cruz2024shallow} for a detailed decomposition). Furthermore, any full permutation of $N$ elements can be expressed as a sequence of adjacent transpositions. Denote the controlled-SWAP between the $i$-th and $(i+1)$-th qubits by $|0\rangle\langle0|\otimes I + |1\rangle\langle1|\otimes \bbW_{(i,i+1)}$, then a controlled cyclic shift gate acting on $N$ qubits is given by the product
	\begin{equation}
		\prod_{i=1}^{N-1}\left[|0\>\<0|\otimes I + |1\>\<1|\otimes \bbW_{(i,i+1)}\right].
	\end{equation}
	
	Incidentally, a relevant work about the quantum advantage enabled by permutation tests can be found in \cite{liu2024separation}.

	\section{\label{sec:conclusion}Conclusion}
	
	By combining the $k$-reduction map and the moment method, we proposed a simple yet powerful approach for certifying 
	the  Schmidt number via $k$-reduction moments. These moments can be efficiently estimated by virtue of randomized measurements based on 3-designs, such as the Clifford group, which are significantly more practical than previous approaches based on 4-designs. Our approach does not require any prior assumption on the target state, is applicable to both $k$-faithful and $k$-unfaithful states, and thus has a much wider scope of applications than fidelity-based methods. The efficacy of the proposed method is substantiated by rigorous theoretical guarantees and comprehensive numerical benchmarks across a range of practically relevant settings.
	
	When $k$-reduction moments with sufficiently high orders are accessible, our criteria become equivalent to the $k$-reduction criterion itself, enabling accurate certification of the entanglement dimensionality of all pure states. Moreover,  only  the first three $k$-reduction moments are required for certifying the entanglement dimensionality of isotropic states. In the course of study, we  clarified the fundamental properties of the spectrum of the $k$-reduced operator and showed that  the $k$-reduction negativity is monotonic under LOCC for pure states.

	Our findings open several avenues for further explorations. First, is the $k$-reduction negativity monotonic under LOCC for general mixed states? Second, are all $k$-faithful states certifiable by the $k$-reduction map? Intuitively, the $k$-reduction criterion is more powerful than  fidelity-based methods, but a strict proof is still absent. Third, how robust are the $k$-reduction moment criteria to statistical fluctuation? To clarify the third question would require knowledge about the perturbation theory of Hankel matrices, which we  leave for future study. Finally, can we construct better criteria for certifying the entanglement dimensionality by virtue of  other $k$-positive maps in conjunction with the moment method? We anticipate that these questions can stimulate further progress on the study of high-dimensional entanglement.
	
	\emph{Note added.} After posting this paper, a recent paper working on a similar topic \cite{mallick2025detecting} was posted, which considered other $k$-positive maps and the application of Schmidt number certification on the task of quantum channel discrimination.
	
	The data used in this study are included in \cite{code2025}.

	\bigskip
	
	\section*{Acknowledgement}
	We thank Xiaodong Yu, Zihao Li, Datong Chen, Zhenhuan Liu, and Shuheng Liu for helpful discussions. This work is supported by Shanghai Science and Technology Innovation Action Plan (Grant No.~24LZ1400200), the
	National Key Research and Development Program
	of China (Grant No.~2022YFA1404204),  Shanghai
	Municipal Science and Technology Major Project
	(Grant No.~2019SHZDZX01), and National Natural
	Science Foundation of China (Grant No.~92165109).

	\bibliography{main}

@article{zyczkowski2001induced,
  title={Induced measures in the space of mixed quantum states},
  author={Zyczkowski, Karol and Sommers, Hans-J{\"u}rgen},
  journal={J. Phys. A: Math. Gen.},
  volume={34},
  number={35},
  pages={7111},
  year={2001},
  publisher={IOP Publishing},
  doi={10.1088/0305-4470/34/35/335},
  url={https://doi.org/10.1088/0305-4470/34/35/335}
}

@article{liu2024separation,
  title = {Separation between entanglement criteria and entanglement detection protocols},
  author = {Liu, Zhenhuan and Wei, Fuchuan},
  journal = {Phys. Rev. Res.},
  volume = {7},
  issue = {3},
  pages = {033121},
  numpages = {18},
  year = {2025},
  month = {Aug},
  publisher = {American Physical Society},
  doi = {10.1103/hm8j-wgqm},
  url = {https://link.aps.org/doi/10.1103/hm8j-wgqm}
}

@article{tomiyama1985geometry,
  title={On the geometry of positive maps in matrix algebras. {II}},
  author={Tomiyama, Jun},
  journal={Linear Algebra Appl.},
  volume={69},
  pages={169--177},
  year={1985},
  publisher={Elsevier},
  doi={10.1016/0024-3795(85)90074-6},
  url={https://doi.org/10.1016/0024-3795(85)90074-6}
}

@article{morelli2023resource,
  title={Resource-efficient high-dimensional entanglement detection via symmetric projections},
  author={Morelli, Simon and Huber, Marcus and Tavakoli, Armin},
  journal={Phys. Rev. Lett.},
  volume={131},
  number={17},
  pages={170201},
  year={2023},
  publisher={APS},
  doi={10.1103/PhysRevLett.131.170201},
  url={https://doi.org/10.1103/PhysRevLett.131.170201}
}

@article{bavaresco2018measurements,
  title={Measurements in two bases are sufficient for certifying high-dimensional entanglement},
  author={Bavaresco, Jessica and Herrera Valencia, Natalia and Kl{\"o}ckl, Claude and Pivoluska, Matej and Erker, Paul and Friis, Nicolai and Malik, Mehul and Huber, Marcus},
  journal={Nat. Phys.},
  volume={14},
  number={10},
  pages={1032--1037},
  year={2018},
  publisher={Nature Publishing Group UK London},
  doi={10.1038/s41567-018-0203-z},
  url={https://doi.org/10.1038/s41567-018-0203-z}
}

@article{horodecki2009quantum,
  title = {Quantum entanglement},
  author = {Horodecki, Ryszard and Horodecki, Pawe\l{} and Horodecki, Micha\l{} and Horodecki, Karol},
  journal = {Rev. Mod. Phys.},
  volume = {81},
  issue = {2},
  pages = {865--942},
  numpages = {0},
  year = {2009},
  month = {Jun},
  publisher = {American Physical Society},
  doi = {10.1103/RevModPhys.81.865},
  url = {https://link.aps.org/doi/10.1103/RevModPhys.81.865}
}

@article{nielsen1999conditions,
  title={Conditions for a class of entanglement transformations},
  author={Nielsen, Michael A},
  journal={Phys. Rev. Lett.},
  volume={83},
  number={2},
  pages={436},
  year={1999},
  publisher={APS},
  doi={10.1103/PhysRevLett.83.436},
  url={https://doi.org/10.1103/PhysRevLett.83.436}
}

@article{gross2007evenly,
    author = {Gross, D. and Audenaert, K. and Eisert, J.},
    title = {Evenly distributed unitaries: {On} the structure of unitary designs},
    journal = {J. Math. Phys.},
    volume = {48},
    number = {5},
    pages = {052104},
    year = {2007},
    month = {05},
    issn = {0022-2488},
    doi = {10.1063/1.2716992},
    url = {https://doi.org/10.1063/1.2716992}
}

@incollection{roberts1993convex,
  title={Convex functions},
  author={Roberts, Arthur Wayne},
  booktitle={Handbook of convex geometry},
  pages={1081--1104},
  year={1993},
  publisher={Elsevier, Amsterdam, Netherlands},
  doi={10.1016/C2009-0-15705-7},
  url={https://doi.org/10.1016/C2009-0-15705-7}
}

@article{zhu2017multiqubit,
  title={Multiqubit {C}lifford groups are unitary 3-designs},
  author={Zhu, Huangjun},
  journal={Phys. Rev. A},
  volume={96},
  number={6},
  pages={062336},
  year={2017},
  publisher={APS},
  doi={10.1103/PhysRevA.96.062336},
  url={https://doi.org/10.1103/PhysRevA.96.062336}
}

@article{webb2015clifford,
author = {Webb, Zak},
title = {The {Clifford} group forms a unitary 3-design},
year = {2016},
issue_date = {November 2016},
publisher = {Rinton Press, Incorporated},
address = {Paramus, NJ},
volume = {16},
number = {15–16},
issn = {1533-7146},
journal = {Quantum Info. Comput.},
month = nov,
pages = {1379–1400},
numpages = {22},
doi={10.5555/3179439.3179447},
url={https://doi.org/10.5555/3179439.3179447}
}

@article{kueng2017low,
  title={Low rank matrix recovery from rank one measurements},
  author={Kueng, Richard and Rauhut, Holger and Terstiege, Ulrich},
  journal={Appl. Comput. Harmon. Anal.},
  volume={42},
  number={1},
  pages={88--116},
  year={2017},
  publisher={Elsevier},
  doi={10.1016/j.acha.2015.07.007},
  url={https://doi.org/10.1016/j.acha.2015.07.007}
}

@article{grier2022sample,
  title={Sample-optimal classical shadows for pure states},
  author={Grier, Daniel and Pashayan, Hakop and Schaeffer, Luke},
  journal={Quantum},
  volume={8},
  pages={1373},
  year={2024},
  publisher={Verein zur F{\"o}rderung des Open Access Publizierens in den Quantenwissenschaften},
  doi={10.22331/q-2024-06-17-1373},
  url={https://doi.org/10.22331/q-2024-06-17-1373}
}

@article{choi1972positive,
  title={Positive linear maps on {C}*-algebras},
  author={Choi, Man-Duen},
  journal={Can. J. Math.},
  volume={24},
  number={3},
  pages={520--529},
  year={1972},
  publisher={Cambridge University Press},
  doi={10.4153/CJM-1972-044-5},
  url={https://doi.org/10.4153/CJM-1972-044-5}
}

@article{takasaki1983geometry,
  title={On the geometry of positive maps in matrix algebras},
  author={Takasaki, Toshiyuki and Tomiyama, Jun},
  journal={Math. Zeitschrift},
  volume={184},
  pages={101--108},
  year={1983},
  publisher={Springer},
  doi={10.1007/BF0116200},
  url={https://doi.org/10.1007/BF01162009}
}

@misc{schmudgen2020lectures,
      title={Ten {L}ectures on the {M}oment {P}roblem}, 
      author={Konrad Schmüdgen},
      year={2020},
      eprint={2008.12698},
      archivePrefix={arXiv},
      primaryClass={math.FA}
}

@book{schmudgen2017moment,
  title={The {Moment} {Problem}},
  author={Schm{\"u}dgen, Konrad},
  year={2017},
  publisher={Springer, Cham, Switzerland},
  doi={10.1007/978-3-319-64546-9},
  url={https://doi.org/10.1007/978-3-319-64546-9}
}

@article{YuIG21,
  title = {Optimal Entanglement Certification from Moments of the Partial Transpose},
  author = {Yu, Xiao-Dong and Imai, Satoya and G\"uhne, Otfried},
  journal = {Phys. Rev. Lett.},
  volume = {127},
  issue = {6},
  pages = {060504},
  numpages = {6},
  year = {2021},
  month = {Aug},
  publisher = {American Physical Society},
  doi={10.1103/PhysRevLett.127.060504},
  url={https://doi.org/10.1103/PhysRevLett.127.060504}
}

@article{VidaW02,
  title = {Computable measure of entanglement},
  author = {Vidal, G. and Werner, R. F.},
  journal = {Phys. Rev. A},
  volume = {65},
  issue = {3},
  pages = {032314},
  numpages = {11},
  year = {2002},
  month = {Feb},
  publisher = {American Physical Society},
  doi={10.1103/PhysRevA.65.032314},
  url={https://doi.org/10.1103/PhysRevA.65.032314}
}

@book{nielsen2010quantum,
  title={Quantum computation and quantum information: 2nd edition},
  author={Nielsen, Michael A and Chuang, Isaac L},
  year={2010},
  publisher={Cambridge University Press, Combriage, UK},
  doi={10.1017/CBO9780511976667},
  url={https://doi.org/10.1017/CBO9780511976667}
}

@article{liu2023characterizing,
  title={Characterizing entanglement dimensionality from randomized measurements},
  author={Liu, Shuheng and He, Qiongyi and Huber, Marcus and G{\"u}hne, Otfried and Vitagliano, Giuseppe},
  journal={PRX Quantum},
  volume={4},
  number={2},
  pages={020324},
  year={2023},
  publisher={APS},
  doi={10.1103/PRXQuantum.4.020324},
  url={https://doi.org/10.1103/PRXQuantum.4.020324}
}

@article{wyderka2023probing,
  title={Probing the geometry of correlation matrices with randomized measurements},
  author={Wyderka, Nikolai and Ketterer, Andreas},
  journal={PRX Quantum},
  volume={4},
  number={2},
  pages={020325},
  year={2023},
  publisher={APS},
  doi={10.1103/PRXQuantum.4.020325},
  url={https://doi.org/10.1103/PRXQuantum.4.020325}
}

@article{elben2023randomized,
  title={The randomized measurement toolbox},
  author={Elben, Andreas and Flammia, Steven T and Huang, Hsin-Yuan and Kueng, Richard and Preskill, John and Vermersch, Beno{\^\i}t and Zoller, Peter},
  journal={Nat. Rev. Phys.},
  volume={5},
  number={1},
  pages={9--24},
  year={2023},
  publisher={Nature Publishing Group UK London},
  doi={10.1038/s42254-022-00535-2},
  url={https://doi.org/10.1038/s42254-022-00535-2}
}

@article{brydges2019probing,
  title={Probing {R}{\'e}nyi entanglement entropy via randomized measurements},
  author={Brydges, Tiff and Elben, Andreas and Jurcevic, Petar and Vermersch, Beno{\^\i}t and Maier, Christine and Lanyon, Ben P and Zoller, Peter and Blatt, Rainer and Roos, Christian F},
  journal={Science},
  volume={364},
  number={6437},
  pages={260--263},
  year={2019},
  publisher={American Association for the Advancement of Science},
  doi={10.1126/science.aau4963},
  url={https://doi.org/10.1126/science.aau4963}
}

@article{neven2021symmetry,
  title={Symmetry-resolved entanglement detection using partial transpose moments},
  author={Neven, Antoine and Carrasco, Jose and Vitale, Vittorio and Kokail, Christian and Elben, Andreas and Dalmonte, Marcello and Calabrese, Pasquale and Zoller, Peter and Vermersch, Benoit and Kueng, Richard and others},
  journal={npj Quantum Inf.},
  volume={7},
  number={1},
  pages={152},
  year={2021},
  publisher={Nature Publishing Group UK London},
  doi={10.1038/s41534-021-00487-y},
  url={https://doi.org/10.1038/s41534-021-00487-y}
}

@article{liu2024bounding,
  title={Bounding entanglement dimensionality from the covariance matrix},
  author={Liu, Shuheng and Fadel, Matteo and He, Qiongyi and Huber, Marcus and Vitagliano, Giuseppe},
  journal={Quantum},
  volume={8},
  pages={1236},
  year={2024},
  publisher={Verein zur F{\"o}rderung des Open Access Publizierens in den Quantenwissenschaften},
  doi={10.22331/q-2024-01-30-1236},
  url={https://doi.org/10.22331/q-2024-01-30-1236}
}

@article{zhou2020single,
  title={Single-copies estimation of entanglement negativity},
  author={Zhou, You and Zeng, Pei and Liu, Zhenhuan},
  journal={Phys. Rev. Lett.},
  volume={125},
  number={20},
  pages={200502},
  year={2020},
  publisher={APS},
  doi={10.1103/PhysRevLett.125.200502},
  url={https://doi.org/10.1103/PhysRevLett.125.200502}
}

@article{elben2020mixed,
  title={Mixed-state entanglement from local randomized measurements},
  author={Elben, Andreas and Kueng, Richard and Huang, Hsin-Yuan Robert and van Bijnen, Rick and Kokail, Christian and Dalmonte, Marcello and Calabrese, Pasquale and Kraus, Barbara and Preskill, John and Zoller, Peter and others},
  journal={Phys. Rev. Lett.},
  volume={125},
  number={20},
  pages={200501},
  year={2020},
  publisher={APS},
  doi={10.1103/PhysRevLett.125.200501},
  url={https://doi.org/10.1103/PhysRevLett.125.200501}
}

@article{terhal2000schmidt,
  title={Schmidt number for density matrices},
  author={Terhal, Barbara M and Horodecki, Pawe{\l}},
  journal={Phys. Rev. A},
  volume={61},
  number={4},
  pages={040301},
  year={2000},
  publisher={APS},
  doi={10.1103/PhysRevA.61.040301},
  url={https://doi.org/10.1103/PhysRevA.61.040301}
}

@article{weilenmann2020entanglement,
  title={Entanglement detection beyond measuring fidelities},
  author={Weilenmann, Mirjam and Dive, Benjamin and Trillo, David and Aguilar, Edgar A and Navascu{\'e}s, Miguel},
  journal={Phys. Rev. Lett.},
  volume={124},
  number={20},
  pages={200502},
  year={2020},
  publisher={APS},
  doi={10.1103/PhysRevLett.125.159903},
  url={https://doi.org/10.1103/PhysRevLett.125.159903}
}

@article{chruscinski2014entanglement,
  title={Entanglement witnesses: construction, analysis and classification},
  author={Chru{\'s}ci{\'n}ski, Dariusz and Sarbicki, Gniewomir},
  journal={J. Phys. A: Math. Theor.},
  volume={47},
  number={48},
  pages={483001},
  year={2014},
  publisher={IOP Publishing},
  doi={10.1088/1751-8113/47/48/483001},
  url={https://doi.org/10.1088/1751-8113/47/48/483001}
}

@article{guhne2009entanglement,
  title={Entanglement detection},
  author={G{\"u}hne, Otfried and T{\'o}th, G{\'e}za},
  journal={Phys. Rep.},
  volume={474},
  number={1-6},
  pages={1--75},
  year={2009},
  publisher={Elsevier},
  doi={10.1016/j.physrep.2009.02.004},
  url={https://doi.org/10.1016/j.physrep.2009.02.004}
}

@article{huan2020,
  title={Predicting many properties of a quantum system from very few measurements},
  author={Huang, Hsin-Yuan and Kueng, Richard and Preskill, John},
  journal={Nat. Phys.},
  volume={16},
  number={10},
  pages={1050--1057},
  year={2020},
  publisher={Nature Publishing Group UK London},
  doi={10.1038/s41567-020-0932-7},
  url={https://doi.org/10.1038/s41567-020-0932-7}
}

@article{friis2019entanglement,
  title={Entanglement certification from theory to experiment},
  author={Friis, Nicolai and Vitagliano, Giuseppe and Malik, Mehul and Huber, Marcus},
  journal={Nat. Rev. Phys.},
  volume={1},
  number={1},
  pages={72--87},
  year={2019},
  publisher={Nature Publishing Group UK London},
  doi={10.1038/s42254-018-0003-5},
  url={https://doi.org/10.1038/s42254-018-0003-5}
}

@article{sanpera2001schmidt,
  title={Schmidt-number witnesses and bound entanglement},
  author={Sanpera, Anna and Bru{\ss}, Dagmar and Lewenstein, Maciej},
  journal={Phys. Rev. A},
  volume={63},
  number={5},
  pages={050301},
  year={2001},
  publisher={APS},
  doi={10.1103/PhysRevA.63.050301},
  url={https://doi.org/10.1103/PhysRevA.63.050301}
}

@article{imai2021bound,
  title={Bound entanglement from randomized measurements},
  author={Imai, Satoya and Wyderka, Nikolai and Ketterer, Andreas and G{\"u}hne, Otfried},
  journal={Phys. Rev. Lett.},
  volume={126},
  number={15},
  pages={150501},
  year={2021},
  publisher={APS},
  doi={10.1103/PhysRevLett.126.150501},
  url={https://doi.org/10.1103/PhysRevLett.126.150501}
}

@inproceedings{aaronson2018shadow,
  title={Shadow tomography of quantum states},
  author={Aaronson, Scott},
  booktitle={Proc. Annu. ACM Symp. Theory Comput.},
  pages={325--338},
  year={2018},
  doi={10.1145/3188745.318880},
  url={https://doi.org/10.1145/3188745.318880}
}

@inproceedings{aaronson2019gentle, 
    author = {Aaronson, Scott and Rothblum, Guy N.}, 
    title = {Gentle measurement of quantum states and differential privacy}, 
    year = {2019}, isbn = {9781450367059}, 
    url = {https://doi.org/10.1145/3313276.3316378}, 
    doi = {10.1145/3313276.3316378}, 
    booktitle = {Proc. Annu. ACM Symp. Theory Comput.}, 
    pages = {322--333}
}

@inproceedings{o2016efficient,
  title={Efficient quantum tomography},
  author={O'Donnell, Ryan and Wright, John},
  booktitle={Proc. Annu. ACM Symp. Theory Comput.},
  pages={899--912},
  year={2016},
  doi={10.1145/2897518.2897544},
  url={https://doi.org/10.1145/2897518.2897544}
}

@article{gross2010quantum,
  title={Quantum state tomography via compressed sensing},
  author={Gross, David and Liu, Yi-Kai and Flammia, Steven T and Becker, Stephen and Eisert, Jens},
  journal={Phys. Rev. Lett.},
  volume={105},
  number={15},
  pages={150401},
  year={2010},
  publisher={APS},
  doi={10.1103/PhysRevLett.105.150401},
  url={https://doi.org/10.1103/PhysRevLett.105.150401}
}

@inproceedings{haah2016sample,
  title={Sample-optimal tomography of quantum states},
  author={Haah, Jeongwan and Harrow, Aram W and Ji, Zhengfeng and Wu, Xiaodi and Yu, Nengkun},
  booktitle={Proc. Annu. ACM Symp. Theory Comput.},
  pages={913--925},
  year={2016},
  doi={10.1109/TIT.2017.2719044},
  url={https://doi.org/10.1109/TIT.2017.2719044}
}

@article{aubrun2012phase,
  title={Phase transitions for random states and a semicircle law for the partial transpose},
  author={Aubrun, Guillaume and Szarek, Stanis{\l}aw J and Ye, Deping},
  journal={Phys. Rev. A},
  volume={85},
  number={3},
  pages={030302},
  year={2012},
  publisher={APS},
  doi={10.1103/PhysRevA.85.030302},
  url={https://doi.org/10.1103/PhysRevA.85.030302}
}

@article{Pere96,
  title={Separability criterion for density matrices},
  author={Peres, Asher},
  journal={Phys. Rev. Lett.},
  volume={77},
  number={8},
  pages={1413},
  year={1996},
  publisher={APS},
  doi={10.1103/PhysRevLett.77.1413},
  url={https://doi.org/10.1103/PhysRevLett.77.1413}
}

@article{Horo01,
  title={Separability of mixed states: {Necessary} and sufficient conditions in terms of linear maps},
  author={Horodecki, Micha{\l} and Horodecki, Pawe{\l} and Horodecki, Ryszard},
  journal={Phys. Lett. A},
  volume={283},
  number={1-2},
  pages={1--7},
  year={2001},
  publisher={Elsevier},
  doi={10.1016/S0375-9601(01)00142-6},
  url={https://doi.org/10.1016/S0375-9601(01)00142-6},
}

@misc{ZhuKGG16,
      title={The {Clifford} group fails gracefully to be a unitary 4-design}, 
      author={Zhu, Huangjun and Kueng, Richard and Grassl, Markus and Gross, David},
      year={2016},
      eprint={1609.08172},
      archivePrefix={arXiv},
      primaryClass={quant-ph}
}

@article{huang2016high,
  title={High-dimensional entanglement certification},
  author={Huang, Zixin and Maccone, Lorenzo and Karim, Akib and Macchiavello, Chiara and Chapman, Robert J and Peruzzo, Alberto},
  journal={Sci. Rep.},
  volume={6},
  number={1},
  pages={27637},
  year={2016},
  publisher={Nature Publishing Group UK London},
  doi={10.1038/srep27637},
  url={https://doi.org/10.1038/srep27637}
}

@article{bannai2020unitary,
  title={Unitary $t$-groups},
  author={Bannai, Eiichi and Navarro, Gabriel and Rizo, Noelia and Tiep, Pham Huu},
  journal={J. Math. Soc. Japan},
  volume={72},
  number={3},
  pages={909--921},
  year={2020},
  publisher={The Mathematical Society of Japan},
  doi={10.2969/jmsj/82228222},
  url={https://doi.org/10.2969/jmsj/82228222}
}

@article{ciesidg2024analysing,
  title={Analysing quantum systems with randomised measurements},
  author={Cie{\'s}li{\'n}ski, Pawe{\l} and Imai, Satoya and Dziewior, Jan and G{\"u}hne, Otfried and Knips, Lukas and Laskowski, Wies{\l}aw and Meinecke, Jasmin and Paterek, Tomasz and V{\'e}rtesi, Tam{\'a}s},
  journal={Phys. Rep.},
  volume={1095},
  pages={1--48},
  year={2024},
  publisher={Elsevier},
  doi={10.1016/j.physrep.2024.09.009},
  url={https://doi.org/10.1016/j.physrep.2024.09.009}
}

@article{guhne2021geometry,
  title={Geometry of faithful entanglement},
  author={G{\"u}hne, Otfried and Mao, Yuanyuan and Yu, Xiao-Dong},
  journal={Phys. Rev. Lett.},
  volume={126},
  number={14},
  pages={140503},
  year={2021},
  publisher={APS},
  doi={10.1103/PhysRevLett.126.140503},
  url={https://doi.org/10.1103/PhysRevLett.126.140503}
}

@article{li2025high,
  title={High-dimensional entanglement witnessed by correlations in arbitrary bases},
  author={Li, Nicky Kai Hong and Huber, Marcus and Friis, Nicolai},
  journal={npj Quantum Inf.},
  volume={11},
  number={1},
  pages={50},
  year={2025},
  publisher={Nature Publishing Group UK London},
  doi={10.1038/s41534-025-00990-6},
  url={https://doi.org/10.1038/s41534-025-00990-6}
}

@article{lib2024experimental,
  title = {Experimental Certification of High-Dimensional Entanglement with Randomized Measurements},
  author = {Lib, Ohad and Liu, Shuheng and Shekel, Ronen and He, Qiongyi and Huber, Marcus and Bromberg, Yaron and Vitagliano, Giuseppe},
  journal = {Phys. Rev. Lett.},
  volume = {134},
  issue = {21},
  pages = {210202},
  numpages = {6},
  year = {2025},
  month = {May},
  publisher = {American Physical Society},
  doi = {10.1103/PhysRevLett.134.210202},
  url = {https://link.aps.org/doi/10.1103/PhysRevLett.134.210202}
}

@article{mallick2025detecting,
  title = {Higher-dimensional-entanglement detection and quantum-channel characterization using moments of generalized positive maps},
  author = {Mallick, Bivas and Maity, Ananda G. and Ganguly, Nirman and Majumdar, A. S.},
  journal = {Phys. Rev. A},
  volume = {112},
  issue = {1},
  pages = {012416},
  numpages = {12},
  year = {2025},
  month = {Jul},
  publisher = {American Physical Society},
  doi = {10.1103/nzrc-8yrt},
  url = {https://link.aps.org/doi/10.1103/nzrc-8yrt}
}

@article{cruz2024shallow,
    author = {M. Q. Cruz, Pedro and Murta, Bruno},
    title = {Shallow unitary decompositions of quantum {Fredkin} and {Toffoli} gates for connectivity-aware equivalent circuit averaging},
    journal = {APL Quantum},
    volume = {1},
    number = {1},
    pages = {016105},
    year = {2024},
    month = {03},
    issn = {2835-0103},
    doi = {10.1063/5.0187026},
    url = {https://doi.org/10.1063/5.0187026},
}

@article{bohnet2012,
  title = {Entanglement and the truncated moment problem},
  author = {Bohnet-Waldraff, F. and Braun, D. and Giraud, O.},
  journal = {Phys. Rev. A},
  volume = {96},
  issue = {3},
  pages = {032312},
  numpages = {12},
  year = {2017},
  month = {Sep},
  publisher = {American Physical Society},
  doi = {10.1103/PhysRevA.96.032312},
  url = {https://link.aps.org/doi/10.1103/PhysRevA.96.032312}
}

@misc{code2025,
  publisher = {GitHub},
  journal = {GitHub repository},
  howpublished = {\url{https://github.com/CYI1995/PRXQ_SN}},
}
	
	\clearpage
	\onecolumngrid
	
	\appendix

	%\section{\label{app:A1}Spectra of $k$-reduced operators}
	
	\section{\label{app:kRNegativitySN}Proofs of \thsref{thm:pure_state}, \ref{thm:k_negativity} and \pref{prop:rank_r_MEDP}}
	
	In this appendix and the following appendices we prove
	\thsref{thm:pure_state}, \ref{thm:k_negativity},  \ref{thm:maximal_order}-\ref{thm:sample_complexity} and \psref{prop:rank_r_MEDP}-\ref{prop:cmc_isotropic} presented in the main text. 
	
	%\subsection{\label{app:kRNegativitySN}Proofs of \thsref{thm:pure_state} and \ref{thm:k_negativity}}
	%
	%\subsection{$k$-reduced operators of pure states}

	\subsection{Auxiliary lemmas}
	Here we introduce two auxiliary lemmas to clarify the basic properties of the operator $\Omega_k(\blambda)$ defined in \eref{eq:Omegaklambda} and the function $\theta_k(\blambda)$ defined \eref{eq:thetaklambda}. Recall that $\Delta_d$ is the $(d-1)$-dimensional probability simplex defined in \eref{eq:simplex}.
	\begin{lemma}\label{lem:thetaOmegalambda}
		Suppose $k$ is a positive integer and $\blambda\in\Delta_d$. 
		Then $\Omega_k(\blambda)$ has at most one negative eigenvalue, and the function
		$\theta_k(\blambda)$ is Schur concave in $\blambda$ and nonincreasing in $k$. Let $r$ be the number of positive entries in $\blambda$; then 
		\begin{align}\label{eq:OmegaRank}
			\rank\left(\Omega_k(\blambda)\right)=\begin{cases}
				r-1, & k=r,\\
				r, & k\neq r.
			\end{cases}
		\end{align}
		If $k\geq r$, then $\Omega_k(\blambda)\geq 0$ and $\theta_k(\blambda)=0$. If $k< r$, then  $-\theta_k(\blambda)$
		is the unique negative eigenvalue of $\Omega_k(\blambda)$, and $\theta_k(\blambda)$ satisfies the following equation:
		\begin{equation}\label{eq:OmegaMinEig}
			\sum_{i=0}^{d-1} \frac{\lambda_i}{k\lambda_i + \theta_k(\blambda)} = 1.
		\end{equation}
	\end{lemma}
	\begin{proof}
		Let $\tblambda$ be the vector obtained from $\blambda$ by deleting the zero entries. Then $\Omega_k(\blambda)$ and $\Omega_k(\tblambda)$ have the same nonzero eigenvalues, including the multiplicities, which means $\theta_k(\blambda)=\theta_k(\tblambda)$.
		In addition, $\Omega_k(\blambda)$ and $\theta_k(\blambda)$  are continuous in $\blambda$. Therefore, it suffices to prove \lref{lem:thetaOmegalambda} under the assumption  $\blambda>0$, which means
		$\lambda_0,\lambda_1, \ldots, \lambda_{d-1}>0$ and $r=d$. 
		
		Let $\Lambda=\diag(\lambda_0,\lambda_1,\ldots, \lambda_{d-1})$; then $\Omega_k(\blambda)=\Lambda^{1/2}\Omega'\Lambda^{1/2}$, where $\Omega'$ is the $d\times d$ matrix with entries $\Omega'_{ij}=k\delta_{ij} -1$. Since $\Lambda$ is an invertible real diagonal matrix, $\Omega$ and $\Omega'$ have the same rank and the same number of positive (negative) eigenvalues. It is straightforward to verify that $\Omega'$ has $r-1$ eigenvalues equal to $k$ and one eigenvalue equal to $k-r$, which implies \eref{eq:OmegaRank}. 
		If $k\geq r=d$, then $\Omega_k(\blambda)\geq 0$ and $\theta_k(\blambda)=0$. If $k<r$, then $\Omega_k(\blambda)$ has exactly one negative eigenvalue, which is equal to the smallest eigenvalue $-\theta_k$.
		
		Suppose $k<r$ and $\boldsymbol{y} = (y_0, y_1,\dots, y_{d-1})^\top$ is  an eigenvector of $\Omega_k(\blambda)$ associated with the eigenvalue  $-\theta_k$. Then 
		\begin{align}\label{eq:A2}
			k\lambda_i y_i -\sum_j\sqrt{\lambda_i\lambda_j}\, y_j=-\theta_k y_i \quad \forall i,
		\end{align}
		which implies that
		\begin{align}
			y_i = \frac{\sqrt{\lambda_i}}{k\lambda_i + \theta_k} \left(\sum_{j=0}^{d-1} \sqrt{\lambda_j}\, y_j\right)\quad \forall i.
		\end{align}
		Note that $\sum_{j=0}^{d-1}\sqrt{\lambda_j}\, y_j \neq 0$, otherwise $\boldsymbol{y}$ is a zero vector. Now, we can deduce \eref{eq:OmegaMinEig} by 
		multiplying both sides with $\sqrt{\lambda_i}$ and summing over $i$.
		
		Next, we prove that $\theta_k(\blambda)$ is nonincreasing in $k$ and Schur concave in $\blambda$. If $k'\geq k$, then $\Omega_{k'}(\blambda)\geq \Omega_{k}(\blambda)$, which implies that $\theta_{k'}(\blambda)\leq \theta_{k}(\blambda)$, so $\theta_k(\blambda)$ is nonincreasing in $k$.

		If $k\geq d$, then $\theta_k(\blambda)=0$, so $\theta_k(\blambda)$ is Schur concave in $\blambda$.
		To address the case with $k<d$, we need to introduce an auxiliary function,
		\begin{equation}\label{eq:Gyk}
			G(y,k,\blambda) \equiv \sum_{i=0}^{d-1}\frac{\lambda_i}{k\lambda_i + y},\quad k,y >0,\quad \blambda>0.
		\end{equation}
		Then $\theta_k(\blambda)$ is the unique solution of $y>0$ to the equation $G(y,k,\blambda)=1$. Note that $G(y,k,\blambda)$ is nonincreasing in $y$, convex in each entry of $\blambda$, and invariant under any permutation of these entries. 
		Therefore, $G(y,k,\blambda)$ is Schur concave in $\blambda$. Suppose $\blambda'\in \Delta_d$ and $\blambda'\prec \blambda$;  let $y=\theta_k(\blambda)$, then $G(y,k,\blambda')\geq G(y,k,\blambda)=1$,
		which implies that $\theta_k(\blambda') \geq y=  \theta_k(\blambda)$. Therefore,  $\theta_k(\blambda)$ is Schur concave in $\blambda$, which completes the proof of \lref{lem:thetaOmegalambda}.	
	\end{proof}

	\begin{lemma}\label{lem:Tomega}
		Suppose $k$ is a positive integer, $\blambda\in\Delta_d$ has $r$ nonzero entries with $\tilde{r}$ distinct nonzero values labeled as $\{\ell_j\}_{j=0}^{\tilde{r}-1}$, where $\ell_j$ has multiplicity $m_j$. Let $\Tomega_k(\blambda)$ be the $\tilde{r}\times \tilde{r}$ matrix with entries $\Tomega_k(\blambda)_{ij}=k\ell_i \delta_{ij}-\sqrt{m_im_j \ell_i\ell_j}$. 
		Then the eigenvalues $\{\omega_j\}_{j=0}^{\tilde{r}-1}$ of $\Tomega_k(\blambda)$ are nondegenerate and satisfy 
		\begin{align}\label{eq:TomegaInterlacing}
			-1<\omega_{\tilde{r}-1}<k\ell_{\tilde{r}-1}<\omega_{\tilde{r}-2}<k\ell_{\tilde{r}-2}<\cdots 	<\omega_0<k\ell_0.
		\end{align}
		In addition,
		\begin{align}
			\rank\left(\Tomega_k(\blambda)\right)&=\begin{cases}
				\tilde{r}-1, & k=r,\\
				\tilde{r}, & k\neq r,
			\end{cases} \label{eq:TomegaRank}  \\
			\sigma(\Omega_k(\blambda)) &= \bigcup_{j=0}^{\tilde{r}-1}\{k\ell_j\}^{\times (m_j-1)}\cup \sigma\bigl(\Tomega_k(\blambda)\bigr) \cup\{0\}^{\times \left(d - r\right)}.\label{eq:spectrum_omegak}
		\end{align} 
	\end{lemma}
	
	\begin{proof}[Proof of \lref{lem:Tomega}]
		According to Cauchy Interlacing theorem, 
		the eigenvalues $\{\omega_j\}_{j=0}^{\tilde{r}-1}$ of $\Tomega_k(\blambda)$  satisfy 
		\begin{align}
			-1<	\omega_{\tilde{r}-1}\leq k\ell_{\tilde{r}-1}\leq \omega_{\tilde{r}-2}\leq k\ell_{\tilde{r}-2}\leq \cdots 	\leq \omega_0\leq k\ell_0.
		\end{align}
		Here the first inequality holds because $\ell_j>0$ and $\sum_{j=0}^{\tilde{r}-1} m_j\ell_j=1$. 
		In addition, using proof by contradiction, it is straightforward to verify that $k\ell_j$ is not an eigenvalue of $\Tomega_k(\blambda)$ for $j=0, 1,\ldots, \tilde{r}-1$. So all the inequalities in the above equation are necessarily strict, which confirms \eref{eq:TomegaInterlacing} and shows that the eigenvalues $\{\omega_j\}_{j=0}^{\tilde{r}-1}$ of $\Tomega_k(\blambda)$ are nondegenerate.
		
		Next,  we prove \eref{eq:spectrum_omegak} before \eref{eq:TomegaRank}. To this end, we need to introduce the following index sets:
		\begin{align}
			\caI_j \equiv \{i : \lambda_i = \ell_j\},\quad j=0,1,\ldots, \tilde{r}-1; \quad \caI_{\tilde{r}} \equiv \{i : \lambda_i = 0\}.
		\end{align}
		On this basis we can define the following subspaces of 
		$\bbR^d$:
		\begin{align}
			\caV_j=\biggl\{v\in \bbR^d : v_i=0\; \forall i\notin \caI_j,\; \sum_i v_i=0\biggr\}, \quad  j=0,1,\ldots, \tilde{r}-1;
			\quad \caV_{\tilde{r}}	=\bigl\{v\in \bbR^d : v_i=0\; \forall i\notin \caI_{\tilde{r}}\bigr\}.
		\end{align}
		Then every nonzero vector in $\caV_j$ is an eigenvector of $\Omega_k(\blambda)$ with eigenvalue $k\ell_j$, where $\ell_{\tilde{r}}=0$. In addition,
		\begin{gather}
			\dim \caV_j=|\caI_j|=m_j-1, \quad j=0,1,\ldots, \tilde{r}-1;  \quad \dim \caV_{\tilde{r}}=|\caI_{\tilde{r}}|=d-r,\\
			\caV_{j_1}\perp \caV_{j_2},\quad 0\leq j_1<j_2\leq \tilde{r}. 
		\end{gather}
		Define $\caW\equiv(\caV_1+\caV_2 +\cdots +\caV_{\tilde{r}})^\perp$. 
		Then $\caW$ is an invariant subspace of $\Omega_k(\blambda)$ that has dimension $\tilde{r}$ and is spanned by the following $\tilde{r}$ vectors in $\bbR^d$:
		\begin{align}
			(v_j)_i=\begin{cases}
				\frac{1}{\sqrt{m_j}}, &i\in \caI_j\\
				0, & i\not\in \caI_j
			\end{cases} \quad j=0, 1,\ldots, \tilde{r}-1. 
		\end{align}
		Now, consider the following map from  $\bbR^{\tilde{r}}$ to  $\bbR^d$:
		\begin{align}
			x\mapsto v_x \equiv\sum_{j=0}^{\tilde{r}} x_j v_j. 
		\end{align}
		Then $v_x$ is an eigenvector of $\Omega_k(\blambda)$ iff $x$ is an eigenvector of $\Tomega_k(\blambda)$ with the same eigenvalue. This observation implies \eref{eq:spectrum_omegak}. 
		
		\Eref{eq:TomegaRank} is a simple corollary of 
		\eqsref{eq:OmegaRank}{eq:spectrum_omegak} given that $\sum_{j=0}^{\tilde{r}-1}(m_j-1) =r-\tilde{r}$, which completes the proof of \lref{lem:Tomega}. 
	\end{proof}

	\subsection{\label{app:proof_theo_2}Proof of \thref{thm:pure_state}}

	\begin{proof}
		Without loss of generality we can assume that $|\psi\>$ has the following Schmidt decomposition: $|\psi\>=\sum_{i=0}^{d-1}\sqrt{\lambda_i} |ii\>$,
		where $d=\min\{d_\rmA,d_\rmB\}=d_\rmA$. Define $\caW_1 \equiv \mathrm{span}\left(\{ |ii\>\}_{i=0}^{d-1}\right)$ and $\caW_2=\caW_1^\perp$.
		Then $\caW_1$ and $\caW_2$ have dimensions $d$ and  $d(d_\rmB-1)$, respectively.
		The $k$-reduced operator $\caR_k(\rho)$ can be expressed as follows:
		\begin{align}
			\caR_k(\rho) =k \sum_{i=0}^{d-1} \lambda_i \ket{i}\bra{i}\otimes I_{\rmB} - \sum_{i,j=0}^{d-1} \sqrt{\lambda_i\lambda_j} \ket{ii}\bra{jj}=R_1+R_2, 
		\end{align}
		where
		\begin{align}
			R_1 \equiv k \sum_{i=0}^{d-1} \lambda_i \ket{ii}\bra{ii}- \sum_{i,j=0}^{d-1} \sqrt{\lambda_i\lambda_j} \ket{ii} \bra{jj},\quad 
			R_2 \equiv k\sum_{i=0}^{d-1}\sum_{j=0, j\ne i}^{d_\rmB-1} \lambda_i \ket{ij}\bra{ij}
		\end{align}
		are supported in $\caW_1, \caW_2$, respectively. 
		Note $\sigma(R_1)=\sigma(\Omega_k(\blambda))$ and $\sigma(R_2) =[k\sigma(\psi)]^{\times (d_\rmB-1)}$. Therefore,
		\begin{align}
			\sigma(\caR_k(\rho))=\sigma(R_1)\cup \sigma(R_2)
			=\sigma(\Omega_k(\blambda)) \cup \{k\sigma(\psi)\}^{\times (d_\rmB-1)},
		\end{align}
		which confirms \eref{eq:RkrhoSpectrum}. Since all eigenvalues of $R_2$ are non-negative, it follows that 
		$\caR_k(\rho)$ and $\Omega_k(\blambda)$ have the same negative spectrum, including multiplicities. According to \lref{lem:thetaOmegalambda}, if $k \ge r$, then $\Omega_k(\blambda)\geq 0$ and  $\caR_k(\rho)\geq 0$.  If instead $1 \le k<r$, then both $\Omega_k(\blambda)$ and  $\caR_k(\rho)$ have exactly one negative eigenvalue. 
		
		Let $\{\ell_j\}_{j=0}^{\tilde{r}-1}$ be the set of distinct nonzero Schmidt coefficients of $|\psi\>$, where $\ell_j$ has multiplicity $m_j$. In conjunction with \lref{lem:Tomega} we can further deduce that
		\begin{equation}
			\sigma(\caR_k(\rho)) = \bigcup_{j=0}^{\tilde{r}-1}\{k\ell_j\}^{\times (d_\rmB m_j-1)}\cup\sigma\bigl(\Tomega_k(\blambda)\bigr)\cup\{0\}^{\times d_\rmB(d- r)},
		\end{equation}
		where $\Tomega_k(\blambda)$ is the $\tilde{r}\times \tilde{r}$ matrix defined in \lref{lem:Tomega}, and its eigenvalues satisfy \eref{eq:TomegaInterlacing}.  Moreover, $\Tomega_k(\blambda)$ is nonsingular when $k\neq r$ and has exactly one eigenvalue equal to 0 when $k=r$. Therefore, $\caR_k(\rho)$ has  $2\tilde{r}$ distinct nonzero eigenvalues when $k\neq r$ and $2\tilde{r}-1$ distinct nonzero eigenvalues when $k=r$. In any case,  $\caR_k(\rho)$ has at most $2\tilde{r}$ distinct nonzero eigenvalues. This observation completes the proof of \thref{thm:pure_state}. 
	\end{proof}

	\subsection{Proof of \thref{thm:k_negativity}}
	\begin{proof}
		According to \thref{thm:pure_state} and 	\lref{lem:thetaOmegalambda}, $\caN_k(\psi)=\theta_k(\blambda)$ is Schur concave in  $\blambda$ and nonincreasing in $k$. Next, suppose $\sn(|\psi\>)=r$ and let $\blambda' \equiv \left(r^{-1},\ldots, r^{-1}, 0,\ldots, 0\right)\in \Delta_d$, where the value  $r^{-1}$ is repeated $r$ times. 
		Then $\blambda \succ\blambda'$, which implies that 
		$\caN_k(\blambda)=\theta_k(\blambda)\leq \theta_k(\blambda') = 1 - k/r$. 
	\end{proof}
	
	\subsection{\label{app:proof_theo_5}Proof of \pref{prop:rank_r_MEDP}}

	\begin{proof}
		Define $\caW_0 \equiv \mathrm{span}\left(\{|i\rangle_\rmA\otimes |j\rangle_\rmB\}_{i,j=0}^{r-1}\right)$. Let $I_{\caW_0}$ be the projector onto this subspace and  $I_{\caW_0^\perp} = I - I_{\caW_0}$, then
		\begin{align}
			\rho_{\varepsilon,r} = (1 - \varepsilon)|+_r\>\<+_r| + \varepsilon \frac{I_{\caW_0} + I_{\caW_0^\perp}}{D} = (1-\varepsilon')\rho^{\mathrm{ef}}_{\varepsilon,r}+\varepsilon'\frac{I_{\caW_0^\perp}}{D - r^2},
		\end{align}
		where $D \equiv d_\rmA d_\rmB$ and
		\begin{equation}
			\rho^{\mathrm{ef}}_{\varepsilon,r} \equiv \frac{1-\varepsilon}{1-\varepsilon'}|+_r\rangle\langle +_r| + \frac{\varepsilon}{1 - \varepsilon'}\frac{I_{\caW_0}}{D}, \quad\varepsilon' \equiv \left(1 -\frac{r^2}{D}\right) \varepsilon.
		\end{equation}
		Every pure state decomposition of $\rho^{\mathrm{ef}}_{\varepsilon,r}$ corresponds to a pure state decomposition of $\rho_{\varepsilon,r}$ with the same maximal Schmidt number. Thus, $\SN(\rho_{\varepsilon,r}) \le \SN\left(\rho^{\mathrm{ef}}_{\varepsilon,r}\right)$. Note that $\rho^{\mathrm{ef}}_{\varepsilon,r}$ is an isotropic state of a bipartite system with local dimensions both equal to $r$. It is known that the Schmidt number of
		\begin{equation}
			\frac{r^2 F-1}{r^2-1}|+_r\rangle\langle +_r|+  \frac{1-F}{r^2 - 1}I_{\caW_0} 
		\end{equation}
		equals $\lceil rF\rceil$ \cite{terhal2000schmidt}. Hence, if $F$ satisfies
		\begin{equation}
			\frac{r^2F - 1}{r^2 - 1} = \frac{1 - \varepsilon}{1 - \varepsilon'},
		\end{equation}
		then we have 
		\begin{equation}
			\SN(\rho_{\varepsilon,r}) \le \SN\left(\rho^{\mathrm{ef}}_{\varepsilon,r}\right)=\lceil rF\rceil= \left\lceil \frac{r^2-1}{r}\cdot \frac{D(1-\varepsilon)}{D(1-\varepsilon) + \varepsilon r^2} + \frac{1}{r}\right\rceil = \left\lceil \frac{(1 + u)r}{1 + r^2 u}\right\rceil,
		\end{equation}
		which confirms the upper bound in \eref{eq:rank_r_MEDPSN}.

		On the other hand, if 
		\begin{equation}
			k\leq\left \lceil \frac{(1 + u)r}{1 + d_\rmB ru}\right\rceil-1, 
		\end{equation}
		then $\caN_k(\rho_{\varepsilon,r})>0$ by
		\pref{prop:depolarizing}, which implies that
		\begin{equation}
			\SN(\rho_{\varepsilon,r})\geq  \max\left\{\frac{(1 + u)r}{1 + d_\rmB ru},\frac{(1 + u)r}{1 + d_\rmA ru}\right\} = \frac{(1 + u)r}{1 + d ru}.
		\end{equation}
		This confirms the lower bound in \eref{eq:rank_r_MEDPSN}. 
		
		Finally, if  $r \le \sqrt{d}$ and  $\varepsilon < 1/2$, then we have $u < (d_\rmA d_\rmB)^{-1}$ and
		\begin{equation}
			\frac{(1 + u)r}{1+dru} > r-1,
		\end{equation}
		which implies that $\SN(\rho_{\varepsilon,r}) = r$.
		This observation completes the proof of \pref{prop:rank_r_MEDP}. 
	\end{proof}

	\section{Proof of \pref{prop:unfaithful}}
	\label{app:proof_of_prop_unfaithful}

	\begin{proof}
		
		To prove that $\rho$ is $k$-unfaithful, we need to find one solution to \eref{eq:sdp}. Let $\mu_{\rmA} = \mu_{\rmB} = 1/2$, then $E_{\rmA},E_{\rmB}$ need to satisfy
		\begin{equation}
			E_{\rmA}\otimes I_{\rmB} + I_{\rmA} \otimes E_{\rmB} \ge \rho,\quad E_{\rmA}\le \frac{1}{2},\quad E_{\rmB} \le \frac{1}{2},\quad \Tr(E_{\rmA}) = \Tr(E_{\rmB}) = \frac{k-1}{2}.\label{eq: conditions}
		\end{equation}
		Because $\pi(j)\neq \tau(j)$ for all $j=0,\ldots,k-1$, we have $\<+_k|\Phi\> = 0$. Then the condition $E_{\rmA}\otimes I_{\rmB} + I_{\rmA} \otimes E_{\rmB} \ge \rho$ leads to
		\begin{gather}
			\frac{1}{k}\sum_{i=0}^{k-1}\<i|E_{\rmA}|i\> + \frac{1}{k}\sum_{i=0}^{k-1}\<i|E_{\rmB}|i\> \ge \<+_k|\rho|+_k\> = \frac{1}{2}, \label{eq:B1}\\
			\frac{1}{\chi}\sum_{j=0}^{\chi-1}\phi_j\<\pi(j)|E_{\rmA}|\pi(j)\> + \frac{1}{\chi}\sum_{j=0}^{\chi-1}\phi_j\<\tau(j)|E_{\rmB}|\tau(j)\> \ge \<\Phi|\rho|\Phi\> = \frac{1}{2}. \label{eq:B2}
		\end{gather}
		We assume $E_{\rmA}, E_{\rmB}$ are diagonal in the computational basis with diagonal entries equal to either $1/4$ or $0$. Then both operators have $2k-2$ non-zero entries according to the fact that $\Tr(E_{\rmA}) = \Tr(E_{\rmB}) = (k-1)/2$. The conditions in \eqsref{eq:B1}{eq:B2} can be fulfilled by setting
		\begin{gather}
			\<i|E_{\rmA}|i\> = \<i|E_{\rmB}|i\> = 1/4,\quad i = 0,1,\ldots,k-1;\\
			\<\pi(j)|E_{\rmA}|\pi(j)\> = \<\tau(j)|E_{\rmB}|\tau(j)\> = 1/4, \quad j = 0,1,\ldots,\chi-1,
		\end{gather}
		which is equivalent to
		\begin{equation}
			\<i|E_{\rmA}|i\> = \frac{1}{4}\quad  \forall i \in \caI_L;\quad \<i|E_{\rmB}|i\> = \frac{1}{4}\quad \forall i \in \caI_R.
		\end{equation}
		Then as long as $|\caI_L|,|\caI_R| \le 2k-2$, the operators $E_{\rmA}, E_{\rmB}$ can be constructed as
		\begin{equation}
			E_{\rmA} = \frac{1}{4}\sum_{i\in\caI_L}|i\>\<i| + \frac{1}{4}\sum_{i\in \caI'}|i\>\<i|,\quad E_{\rmB} = \frac{1}{4}\sum_{i\in\caI_R}|i\>\<i| + \frac{1}{4}\sum_{i\in \caI''}|i\>\<i|,
		\end{equation}
		where $\caI'$ can be any index set that has no overlap with $\caI_L$ and satisfies $|\caI_L| + |\caI'| = 2k-2$. The other set $\caI''$ can be construct in the same way. Clearly $E_{\rmA} ,E_{\rmB} \le 1/2$, hence all requirements in \eref{eq: conditions} are satisfied. This proves that $\rho$ is $k$-unfaithful.

		Now we prove that $\caR_{k-1}(\rho) \not\ge 0$ with the extra condition in \eref{eq: extra_condition}. The reduced density matrix reads
		\begin{equation}
			\rho_{\rmA} = \frac{1}{2k}\sum_{i=0}^{k-1}|i\>\<i| + \frac{1}{2}\sum_{j=0}^{\chi-1}\phi_j|\pi(j)\>\<\pi(j)| = \sum_{i=0}^{k-1}\frac{1 + k\phi_{\pi^{-1}(i)}}{2k}|i\>\<i| + \frac{1}{2}\sum_{j:\pi(j)\ge k}\phi_j |\pi(j)\>\<\pi(j)|.
		\end{equation}
		Let $\Pi = \sum_{i=0}^{k-1}|ii\>\<ii|$. Then
		\begin{equation}
			\Pi \caR_{k-1}(\rho)\Pi = (k-1)\sum_{i=0}^{k-1}\<i|\rho_{\rmA}|i\>|ii\>\<ii| - \frac{1}{2}|+_k\>\<+_k| - \frac{1}{2}\Pi|\Phi\>\<\Phi|\Pi,
		\end{equation}
		where $\Pi|\Phi\>\<\Phi|\Pi = 0$ as $\pi(j)\neq \tau(j)$ for all $j=0,\ldots,k-1$. Clearly, if $\Pi \caR_{k-1}(\rho)\Pi\not\ge 0$, then $\caR_{k-1}(\rho) \not\ge 0$. Suppose $\boldsymbol{y} = (y_j)_{j=0}^{k-1}$ is an eigenvector of $\Pi \caR_{k-1}(\rho)\Pi$ with eigenvalue $-\theta$. Then
		\begin{equation}
			(k-1)\<i|\rho_{\rmA}|i\>y_i - \frac{1}{2k}\left(\sum_{i=0}^{k-1}y_i\right) = -\theta y_i,\quad i = 0,\ldots,k-1.
		\end{equation}Using the same technique in the proof of \lref{lem:thetaOmegalambda}, we get
		\begin{equation}
			\sum_{i=0}^{k-1}\frac{1}{(k-1)\<i|\rho_{\rmA}|i\> + \theta} = 2k.
		\end{equation}
		Hence, $\Pi \caR_k(\rho)\Pi$ has a negative eigenvalue iff
		\begin{equation}
			\sum_{i=0}^{k-1}\frac{1}{(k-1)\<i|\rho_{\rmA}|i\>} > 2k,
		\end{equation}
		which is equivalent to \eref{eq: extra_condition}. This observation completes the proof of Proposition \ref{prop:unfaithful}.
		
	\end{proof}
	
	\section{Proofs of \thsref{thm:main} and \ref{thm:maximal_order}}
	
	\subsection{An auxiliary lemma}
	
	Given $x \in \bbR$ and  $l\in \bbN_0$, define
	\begin{gather}
		e(x,l) \equiv
		\left(1,x,x^2,\ldots,x^{l}\right)^\top.
	\end{gather}
	Consider the moment sequence $S = (s_n)_{n\in\bbN_0}$ with $s_n = x^n$ and let $S_{k,2l+k} \equiv (s_k,s_{k+1},\ldots,s_{2l+k})$ be a truncated subsequence. The Hankel matrix $H(S_{k,2l+k})$ has entries
	\begin{equation}
		[H(S_{k,2l+k})]_{m,n} = s_{m+n+k}=x^k e(x,l)_m e(x,l)_n, \quad m,n = 0,1,\ldots,l.
	\end{equation}
	So  the Hankel matrix itself can be expressed as follows:
	\begin{equation}\label{eq:Hankel_decomposition}
		H(S_{k,2l+k}) = x^k e(x,l)e(x,l)^\top.
	\end{equation} 
	Next, as a generalization, consider the moment sequence $S' = (s'_n)_{n\in\bbN_0}$ with $s'_n = \sum_{i=0}^{\chi-1}m_i x_i^n$, where $x_i \in \bbR$ and $m_i \in \bbN$. By linearity, we have
	\begin{equation}\label{eq:Hankeldecomposition}
		H(S'_{k,2l+k}) = \sum_{i=0}^{\chi-1}m_i x_i^k e(x_i,l)e(x_i,l)^\top.
	\end{equation}
	The following lemma is instructive for understanding the properties of the Hankel matrix $H(S'_{k,2l+k})$. 
	
	\begin{lemma}\label{lem:vandermonde}
		Suppose $\{x_i\}_{i=0}^{\chi-1}$ is a set of distinct real numbers and $l\in \bbN_0$. Then 
		\begin{equation}\label{eq:vandermonde}
			\rank\left(\left[e(x_0,l), e(x_1,l),\ldots, e(x_{\chi-1},l)\right]\right) = \min\left\{l+1,\chi\right\}.
		\end{equation}
		If $l+1 \ge \chi$, then the set $\{e(x_i,l)\}_{i=0}^{\chi-1}$ is  linearly independent.
	\end{lemma}
	\begin{proof}
		Let $V$ be the Vandermonde matrix  associated with the set $\{x_i\}_{i=0}^{\chi-1}$, that is,
		\begin{equation}
			V_{ij} \equiv x_i^j\quad i,j = 0,1,\ldots,\min\left\{l, \chi-1\right\}.
		\end{equation}
		Thanks to the Vandermonde determinant formula, we have
		\begin{equation}
			\det(V) = \prod_{0 \le i < j \le \dim(V)-1}(x_i - x_j) \neq 0.
		\end{equation}
		Therefore,
		\begin{equation}
			\min\{l+1, \chi\}=   \mathrm{rank}(V)\leq  \rank\left(\left[e(x_0,l), e(x_1,l),\ldots, e(x_{\chi-1},l)\right]\right)\leq \min\left\{l+1,\chi\right\},
		\end{equation}
		which implies \eref{eq:vandermonde}. 
		
		If $l +1 \ge \chi$, then $\rank\left(\left[e(x_0,l), e(x_1,l),\ldots, e(x_{\chi-1},l)\right]\right) = \chi$ by \eref{eq:vandermonde}, so the set $\{e(x_i,l)\}_{i=0}^{\chi-1}$ is  linearly independent, which completes the proof of \lref{lem:vandermonde}. 
	\end{proof}
	
	\subsection{\label{app:proofmain}Proofs of \thsref{thm:main} and \ref{thm:maximal_order}}

	In preparation  for proving \thsref{thm:main} and \ref{thm:maximal_order}, suppose $\rho\in \caS(\caH_{\rmA\rmB})$ and  $\caR_k(\rho)$ has $\chi$ distinct nonzero eigenvalues, namely,  $\{x_i\}_{i=0}^{\chi-1}$, where $x_i$ has multiplicity $m_i$. According to the definition in \eref{eq:BN} and the formula in \eref{eq:Hankeldecomposition}, $B_N[\rho,k]$ can be expressed as follows:
	\begin{align}
		B_N[\rho,k] &= \sum_{i=0}^{\chi-1}m_i p_{N,k}(x_i)b_N(x_i)b_N(x_i)^\top,\label{eq:bn_decomposition}
	\end{align}
	where
	\begin{align}
		b_N(x) &\equiv e(x,\lfloor (N-1)/2\rfloor) = \left(1,x,x^2,\ldots,x^{\lfloor\frac{N-1}{2}\rfloor}\right)^\top ,\label{eq:bn_vector}\\
		p_{N,k}(x) &\equiv \begin{cases}
			x, & N \textrm{ odd};\\
			x(k-x), & N \textrm{ even}.
		\end{cases}\label{eq:pNkx}
	\end{align}

	\begin{proof}[Proof of \thref{thm:main}]
		Suppose $\caR_k(\rho)$ has $\chi$ different nonzero eigenvalues $\{x_i\}_{i=0}^{\chi-1}$, where $x_i$ has multiplicity $m_i \ge 1$. Then $B_N$ has the decomposition in \eref{eq:bn_decomposition} with $x_i\in \bbR$ and $m_i\in \bbN$. Since $\sn(\rho) \le k$ by assumption, we have $0 \le x_i \le k$ for all $i$ thanks to \thref{thm:pure_state} and \coref{coro:spectrum}. Therefore, $m_ip_{N,k}(x_i) \ge 0$ for all $i$, which means $B_N \ge 0$ and completes the proof of \thref{thm:main}.
	\end{proof}
	
	\begin{proof}[Proof of \thref{thm:maximal_order}]
		Let $\{x_i\}_{i=0}^{\chi-1}$  be the set of distinct nonzero eigenvalues of $\caR_k(\rho)$  and denote by $m_i$ the  multiplicity of $x_i$. Then $B_N$ has the decomposition in  \eref{eq:bn_decomposition}  with  $m_i\in \bbN$ and $x_i\in [-1+k/d, k]$ by \coref{coro:spectrum}. In addition,  $\{b_N(x_i)\}_{i=0}^{\chi-1}$ are linearly independent according to \lref{lem:vandermonde}, given that $N \ge 2\chi -1$ by assumption. Hence, $B_N\ge 0$ iff $m_i p_{N,k}(x_i) \ge 0$ for all $i$. Meanwhile, $m_i p_{N,k}(x_i) \ge 0$ iff $x_i\ge 0$. Therefore, $B_N \ge 0$ iff $x_i \ge 0$ 
		for all $i$, that is, $\caR_k(\rho)\geq 0$, which completes the proof of \thref{thm:maximal_order}.
	\end{proof}

	\section{\label{app:k_reduction_detect}Proofs of  \thref{thm:pure_state_detect} and \pref{prop:isotropic}           }
	
	\subsection{\label{app:proof_of_theorem5}Proof of \thref{thm:pure_state_detect}}
	
	\begin{proof}[Proof of \thref{thm:pure_state_detect}]
		Let $k=1,2,\ldots, r-1$; then $\caR_{k}(\psi) \not\geq 0$ and  $\caR_k(\psi)$ has at most $2\tilde{r}$ distinct nonzero eigenvalues according to \thref{thm:pure_state}. Suppose $N \ge 4\tilde{r}-1$; then $\caR_{k}(\psi) \not\geq 0$ iff $B_N[\psi,k] \not\geq 0$ thanks to \thref{thm:maximal_order}. Therefore, $B_N[\psi,k] \not\geq 0$, which  confirms \eref{eq: theorem6statement1}.
		
		Let $\{\ell_j\}_{j=0}^{\tilde{r}-1}$	be the set of distinct nonzero Schmidt coefficients of $|\psi\>$, where $\ell_j$ has multiplicity $m_j$. 
		According to \lref{lem:Tomega}, the spectrum of $\caR_k(\psi)$ can be expressed as follows:
		\begin{equation}
			\sigma(\caR_{k}(\psi)) = \bigcup_{j=0}^{\tilde{r}-1}\{k\ell_j\}^{\times (d_\rmB m_j-1)}\cup \{\omega_j\}_{j=0}^{\tilde{r}-1} \cup\{0\}^{\times d_\rmB(d- r)},
		\end{equation}
		where $\omega_j$ for $j=0,1,\ldots, \tilde{r}-1$ are the eigenvalues of the matrix $\Tomega(\blambda)$ defined in \lref{lem:Tomega}. Note that $\sum_j m_j = r$ and
		\begin{align}
			0<k\ell_j\leq k, \quad -1< \omega_j< k \quad \forall j=0, 1,\ldots, \tilde{r}-1; \quad 0< \omega_j<k\quad \forall j=0, 1,\ldots, \tilde{r}-2. 
		\end{align}
		In conjunction with Eqs. (\ref{eq:BN}) and (\ref{eq:bn_decomposition}),  the Hankel matrix $B_N[\psi,k]$ can be expressed as follows:
		\begin{equation}
			\begin{gathered}
				B_N=	B_N[\psi,k] =A_1+A_2,\\
				A_1\equiv \sum_{j=0}^{\tilde{r}-1}p_{N,k}(\omega_j) b_N(\omega_j)b_N(\omega_j)^\top,\quad
				A_2\equiv\sum_{j=0}^{\tilde{r}-1} (d_\rmB m_j - 1) p_{N,k}(k\ell_j) b_N(k\ell_j)b_N(k\ell_j)^\top,
			\end{gathered}
		\end{equation}
		where $b_N$ and $p_{N,k}(x)$ are defined in  \eqsref{eq:bn_vector}{eq:pNkx}, respectively. Note 
		that $p_{N,k}(k\ell_j)\geq 0$, while  $p_{N,k}(\omega_j)$ has the same sign as $\omega_j$ for $j=0,1,\ldots, \tilde{r}-1$.

		Now suppose  $N \le 2\tilde{r}$ and $k=r-1$. Let $\ell_{\min} \equiv \min_j \{\ell_j\}$ and $\omega_{\min} \equiv \min_j\{\omega_j\} = -\caN_{r-1}(\psi)$; then $\omega_{\min}\in  (-1,0)$. In addition, $A_1$, $A_2$, and $B_N$ are square matrices with at most $\tilde{r}$ rows (columns), and $1 \le \|b_N(\omega_{\min})\|_2^2 \le \tilde{r}$. In conjunction with \lref{lem:vandermonde}, we can deduce that $A_2$ has full rank, which means $\Im(A_1) \subset \Im(A_2)$. 
		
		By definition in \eref{eq:A3} we have
		\begin{equation}
			A_0 = \sum_{j=0}^{\tilde{r}-1} p_{N,r-1}[(r-1)\ell_j] b_N[(r-1)\ell_j]b_N[(r-1)\ell_j]^\top,
		\end{equation}
		which means $\sigma_{\min}(A_2) \ge (d_\rmB-1)\sigma_{\min}(A_0)$ given that $m_j \ge 1$ and $p_{N,r-1}[(r-1)\ell_j] \ge 0$ for all $j$.
		If in addition $N$ is odd, then 
		\begin{align}
			\sigma_{\min}(A_1) & \ge p_{N,r-1}(\omega_{\min})\|b_N(\omega_{\min})\|_2^2 \geq   \tilde{r} \omega_{\min},\\
			\sigma_{\min}\left(B_N\right) &\ge \sigma_{\min}(A_1) + \sigma_{\min}(A_2) \ge \tilde{r}\omega_{\min} + (d_\rmB -1)\sigma_{\min}(A_0).
		\end{align}
		If instead $N$ is even, then 
		\begin{align}
			\sigma_{\min}(A_1) & \ge p_{N,r-1}(\omega_{\min})\|b_N(\omega_{\min})\|_2^2 \geq   r\tilde{r} \omega_{\min},\\
			\sigma_{\min}(B_N) &\ge r\tilde{r}\omega_{\min} + (d_\rmB-1)\sigma_{\min}(A_0).  \label{eq:BNeigMinEven}
		\end{align}
		Note that $\omega_{\min} = -\caN_{r-1}(\psi)$ and that \eref{eq:BNeigMinEven} holds in both cases. If the condition in \eref{eq:dBLBcondition} is satisfied, that is,
		$d_\rmB- 1 \ge r\tilde{r}\caN_{r-1}(\psi)/\sigma_{\min}(A_0)$, then $\sigma_{\min}(B_N) > 0$, which  confirms \eref{eq: theorem6statement2} and completes the proof of \thref{thm:pure_state_detect}.
	\end{proof}

	\subsection{Proof of \pref{prop:isotropic}}
	
	\begin{proof}
		The $k$-reduced operator $\caR_k(\rho_F)$ has spectrum
		\begin{equation}
			\sigma(\caR_k(\rho_F)) = \left\{\lambda_1\right\}^{\times (d^2 - 1)}\sqcup\left\{\lambda_0\right\},\quad \lambda_1 = \frac{F - 1}{d^2 - 1} + \frac{k}{d} > 0,\quad \lambda_0 = \frac{k}{d} - F.
		\end{equation}
		Note that $\caR_k(\rho_F)$ has at most two distinct nonzero eigenvalues. 
		Given $N = 3$, then  $B_N[\rho_F,k] \ge 0$ iff $\caR_k(\rho_F) \ge 0$ according to \thref{thm:maximal_order}. If in addition $k \le \lceil dF\rceil - 1$, then $\lambda_0 < 0$ and $\caR_k(+_r) \not\ge 0$, which further implies that $B_N[+_r,k] \not\ge 0$.
	\end{proof}

	\section{\label{app:cmc_detect}Proof of \pref{prop:cmc_isotropic}}
	
	\begin{proof}
		By assumption the local Hilbert spaces $\caH_\rmA$ and $\caH_\rmB$ have the same dimension $d$ and thus can be identified as two copies of a $d$-dimensional Hilbert space $\caH$. 
		Let $\{P_j\}_{j=1}^{d^2-1}$ be an orthogonal basis
		on $\caL_0^\rmH(\caH)$
		that satisfies the normalization condition $\tr(P_jP_k)=d\delta_{j,k}$. Then $\{P_j^*\}_{j=1}^{d^2-1}$ is also an orthogonal basis on this space. In conjunction with the definition of $\rho_F$ in \eref{eq:isotropic_state}, the entries of its CM can be calculated as follows:
		\begin{align}
			T_{jk} &= \frac{d^2F-1}{d(d^2-1)}\langle +_d|P_j\otimes P^*_k |+_d\rangle = \frac{d^2F-1}{d^2(d^2-1)}\sum_{m,n=0}^{d-1}\langle m m|P_j\otimes  P_k^*|n n\rangle \nonumber\\
			&= \frac{d^2 F-1}{d^2(d^2-1)}\sum_{m,n=0}^{d-1}\langle m|P_j|n\rangle\langle n|P_k^\dag|m\rangle = \frac{d^2 F - 1}{d^2(d^2 - 1)}\Tr\bigl(P_j  P_k^\dagger\bigr) = \frac{d^2 F-1}{d(d^2-1)}\delta_{j,k},
		\end{align}
		where $j,k = 1,2,\ldots,d^2-1$. 
		Therefore, all singular values of $T$ are equal to $\frac{d^2 F - 1}{d(d^2 - 1)}$, which has multiplicity $d^2-1$. As a simple corollary, we have 
		
		\begin{align}
			\|T\|_1 &= (d^2-1)\times\frac{d^2 F - 1}{d(d^2 - 1)} = \frac{d^2F - 1}{d},\\
			\|T\|_2^2 &= (d^2-1)\times \left[\frac{d^2 F - 1}{d(d^2 - 1)}\right]^2 = \frac{(d^2F-1)^2}{d^2(d^2-1)},\\
			\|T\|_4^4 &= (d^2-1)\times \left[\frac{d^2 F - 1}{d(d^2 - 1)}\right]^4 = \frac{(d^2F-1)^4}{d^4(d^2-1)^3},
		\end{align}
		
		which completes the proof of \pref{prop:cmc_isotropic}. 
	\end{proof}

	\section{\label{app:haar_random}Proof of \thref{thm:sample_complexity}}
	
	In this appendix, we will show that, given access to a unitary 3-design, the sample complexity of estimating the functions $p_2 \equiv \Tr(\rho^2)$,  $p_3 \equiv \Tr(\rho^3)$, and $t_2 \equiv \Tr[\rho_{\rmA} \Tr_{\rmB}(\rho^2)]$ using the classical shadow method is $\caO(D)$, where $D$ is the dimension of the underlying Hilbert space $\caH_{\rmA\rmB}$. Note that the protocol for estimating $a_2,a_3$ is identical to the counterpart for estimating $p_2,p_3$, but the underlying Hilbert space $\caH_\rmA$ has dimension $\sqrt{D}$ instead. Hence, the sample complexity of estimating $a_2,a_3$ is  smaller than that of $p_2,p_3$, and the sample complexity of the third order moment-based $k$-reduction criterion is $\caO(D)$.

	\subsection{An auxiliary lemma}
	
	We first recall the necessary background on unitary designs. Denote by $\Phi^{(k)}$ the $k$-th twirling superoperator defined as follows:
		\begin{equation}
		\Phi^{(k)}(O) \equiv \int d\mu(U) \, U^{\otimes k} O (U^\dagger)^{\otimes k}, \quad O\in \caL\left(\caH_{\rmA\rmB}^{\otimes k}\right).
	\end{equation}
	Here $\mu(U)$ denotes the Haar measure on the  group $\rmU(\caH_{\rmA\rmB})$ of unitary operators on $\caH_{\rm A\rm B}$. A unitary ensemble $\caU$ on $\rmU(\caH_{\rmA\rmB})$ 
	is called a unitary $k$-design if
	\begin{equation}
		\mathbb{E}_{U \sim \caU} [U^{\otimes k} \cdot (U^\dagger)^{\otimes k}] = \Phi^{(k)}(\cdot).
	\end{equation}
	 When $O = (|\psi\>\<\psi|)^{\otimes k}$ is the $k$-th tensor power of a pure state, by virtue of Schur's lemma we can deduce that \begin{equation} \label{eq:Schur}
		\Phi^{(k)}\left[(|\psi\>\<\psi|)^{\otimes k}\right] = \int_{\mathrm{Haar}} dv(|v\>\<v|)^{\otimes k} = \frac{(D-1)!}{(D+k-1)!}\sum_{\pi\in \mathrm{Sym}_k}\bbW_{\pi},
	\end{equation}
	where  $D$ is the dimension of the Hilbert space $\caH_{\rm A\rm B}$, $\mathrm{Sym}_k$ is the symmetric group on $k$ elements, and $\bbW_{\pi}$ denotes the permutation operator on $\caH^{\otimes k}_{\rmA\rmB}$ corresponding to $\pi\in \mathrm{Sym}_k$. The above equation will be useful in the following derivations.
	
	In the classical shadow protocol, we first sample a random unitary $U$ from a given unitary ensemble $\caU$, then measure the rotated state $U \rho U^\dagger$ in the computational basis and record the measurement outcome $b$.  When $\caU$ forms a unitary 2-design, an unbiased estimator for the state $\rho$ can be constructed as follows:
	\begin{equation}\label{eq:ShadowEstimator}
		\hat{\rho} \equiv (D+1) U^\dagger |b\rangle\langle b| U - I.
	\end{equation}
	On this basis, we can construct unbiased estimators for $p_2$, $p_3$, and $t_2$.  To analyze the variances of these estimators, we need to introduce some  notation and an auxiliary lemma.  Define
	\begin{equation}
		R_\rho \equiv I \otimes I + I \otimes \rho + \rho \otimes I, \qquad \mathbb{W} \equiv \mathbb{W}_{(12)},
	\end{equation}
	where $(12)$ denotes the transposition swapping the first and second elements.
The following lemma is quite helpful for computing the variances of the estimators $p_2$, $p_3$, and $t_2$. A similar result can be found in Lemma 14 of \cite{grier2022sample}.
	\begin{lemma}\label{lem:classicalshadow}
		Suppose $\caU$ is a unitary 3-design, and $\hr$, $\hr'$ are two independent estimators for $\rho$. Then $\bbE[\hr] = \rho$ and
		\begin{align}
			\label{eq:F7} \bbE[\hr\otimes\hr] &= \frac{D+1}{D+2}R_\rho\bbW - \frac{1}{D+2}R_\rho,\\ 
			\label{eq:F8}    \bbE[\hr\otimes\hr]^2 &= \frac{(D+1)^2 + 1}{(D+2)^2}R_\rho^2 - \frac{2(D+1)}{(D+2)}R_\rho^2\bbW,\\ 
			\bbE[(\hr\otimes\hr')(\hr'\otimes\hr)] &= \frac{(D+1)^3}{(D+2)^2}R_\rho \bbW - \frac{2(D+1)}{(D+2)^2}R_\rho^2\bbW + \frac{1}{(D+2)^2}R_\rho^2  \nonumber\\
			\label{eq:F9} &\quad + \frac{D(D+1)^2}{(D+2)^2}(\rho \otimes \rho)\bbW + \frac{(D+1)^2}{(D+2)^2}(1 + p_2)\bbW.
		\end{align}
	\end{lemma}
	\begin{proof}
The equality $\bbE[\hr] = \rho$ follows from \eref{eq:ShadowEstimator} and the equation below:
		\begin{align}
			\bbE[U^\dag |b\>\<b|U] &= \bbE_{U\sim\caU}\left[\sum_b\langle b|U \rho U^\dag|b\rangle U^\dag |b\>\<b|U\right] = \sum_b\Tr_2\left\{\bbE_{U\sim\caU}\left[(U^\dag)^{\otimes 2}(|b\rangle\langle b|)^{\otimes 2} U^{\otimes 2}\right](I\otimes \rho)\right\} \nonumber\\
			&= \sum_b\Tr_2\left\{\Phi^{(2)}\left[(|b\rangle\langle b|)^{\otimes 2}\right](I\otimes \rho)\right\}\quad \text{ since } \caU \text{ is a unitary 2-design}\nonumber \\
			&= \sum_b\Tr_2\left[\frac{1}{D(D+1)}(I + \bbW)(I\otimes \rho)\right] = \frac{\rho + I}{D+1},
		\end{align}
	where the fourth equality follows from 	\eref{eq:Schur} with $k=2$ and in deriving the last equality we have taken into account  a factor of $D$ from the summation over $b$.

To prove \eref{eq:F7}, we first evaluate the following expectation value:		
		\begin{align}
			\bbE[(U^\dag |b\>\<b|U)^{\otimes 2}] &= \bbE_{U\sim\caU}\left[\sum_b\langle b|U \rho U^\dagger|b\rangle (U^\dag |b\>\<b|U)^{\otimes 2}\right]\nonumber\\
			&= \sum_b\Tr_3\left\{\bbE_{U\sim\caU}\left[(U^\dag)^{\otimes 3}(|b\rangle\langle b|)^{\otimes 3} U^{\otimes 3}\right](I\otimes I\otimes \rho)\right\}\nonumber \\
			&= \sum_b\Tr_3\left\{\Phi^{(3)}\left[(|b\rangle\langle b|)^{\otimes 3}\right](I\otimes I\otimes \rho)\right\}\quad \text{ since } \caU \text{ is a unitary 3-design}\nonumber \\
			&= \frac{1}{(D+1)(D+2)}(I\otimes I + \rho\otimes I + I\otimes\rho)(I + \bbW),
		\end{align}
	where the last equality follows from 	\eref{eq:Schur} with $k=3$. Now, \eref{eq:F7} can be proved as follows:
		\begin{align}
			\bbE[\hr\otimes\hr] &= (D+1)^2\bbE[(U^\dag |b\>\<b|U)^{\otimes 2}] - (D+1)\left\{\bbE[U^\dag |b\>\<b|U]\otimes I + I\otimes \bbE[U^\dag |b\>\<b|U]\right\} + I\otimes I\nonumber\\
			&= \frac{D+1}{D+2}(I\otimes I + I\otimes \rho + \rho\otimes I)\bbW - \frac{1}{D+2}(I\otimes I + \rho\otimes I + I\otimes\rho).
		\end{align}
		
	\Eref{eq:F8} can be proved using a similar reasoning employed above.
		
		To prove \eref{eq:F9}, we first rewrite the LHS as follows:
		\begin{align} \label{eq:F14}
			\bbE[(\hr\otimes\hr')(\hr'\otimes\hr)] &=\bbE\left\{\Tr_{34}\left[(\hat{\rho} \otimes \hat{\rho'}\otimes \hat{\rho'} \otimes \hat{\rho}) \cdot \bbW_{(13)(24)}\right]\right\}\nonumber\\
			&= \bbE\left\{\Tr_{34}\left[\bbW_{(234)}(\hat{\rho} \otimes \hat{\rho}\otimes \hat{\rho'} \otimes \hat{\rho'})\bbW^{-1}_{(234)} \cdot \bbW_{(13)(24)}\right]\right\}\nonumber\\
			&= \Tr_{34}\left\{\bbW_{(234)}\left(\bbE[\hr\otimes\hr]\otimes\bbE[\hr\otimes\hr]\right)\bbW_{(134)}\right\}\nonumber\\
			&= \frac{(D+1)^2}{(D+2)^2}\Tr_{34}\left[\bbW_{(234)}\left(R_\rho \bbW \otimes R_\rho \bbW\right)\bbW_{(134)}\right]\nonumber\\
			&\equad - \frac{(D+1)}{(D+2)^2}\Tr_{34}\left[\bbW_{(234)}\left(R_\rho \bbW \otimes R_\rho\right)\bbW_{(134)}\right]\nonumber\\
			&\equad - \frac{(D+1)}{(D+2)^2}\Tr_{34}\left[\bbW_{(234)}\left(R_\rho \otimes R_\rho \bbW\right)\bbW_{(134)}\right]\nonumber\\
			&\equad + \frac{1}{(D+2)^2}\Tr_{34}\left[\bbW_{(234)}\left(R_\rho \otimes R_\rho\right)\bbW_{(134)}\right].
		\end{align}
In addition, direct calculation yields
\begin{gather}
			\Tr_{34}\left[\bbW_{(234)}\left(R_\rho \bbW \otimes R_\rho \bbW\right)\bbW_{(134)}\right] = (D+1)R_\rho \bbW + (1+p_2)\bbW + D(\rho \otimes \rho)\bbW, \\
			\Tr_{34}\left[\bbW_{(234)}\left(R_\rho \bbW \otimes R_\rho\right)\bbW_{(134)}\right] = (1 + 2 \rho \otimes I)R_\rho\bbW,\\
			\Tr_{34}\left[\bbW_{(234)}\left(R_\rho \otimes R_\rho\bbW\right)\bbW_{(134)}\right] = (1 + 2I\otimes \rho)R_\rho\bbW,\\
			\Tr_{34}\left[\bbW_{(234)}\left(R_\rho \otimes R_\rho\right)\bbW_{(134)}\right] = R_\rho^2.
		\end{gather}
		In conjunction with \eref{eq:F14}, we can derive \eref{eq:F9} and complete the proof of \lref{lem:classicalshadow}.
	\end{proof}
	
%			Observe that for any operators $F_1,F_2,F_3,F_4$, we have
%	\begin{gather}
%		\Tr_{34}\left[\bbW_{(234)}(F_1 \otimes F_2 \otimes F_3 \otimes F_4)(\bbW \otimes \bbW) \bbW_{(134)}\right] = \Tr(F_2 F_4)(F_1 \otimes F_3)\bbW,\\
%		\Tr_{34}\left[\bbW_{(234)}(F_1 \otimes F_2 \otimes F_3 \otimes F_4)(\bbW \otimes I \otimes I) \bbW_{(134)}\right] = [F_1 \otimes (F_3 F_2 F_4)]\bbW,\\
%		\Tr_{34}\left[\bbW_{(234)}(F_1 \otimes F_2 \otimes F_3 \otimes F_4)(I \otimes I \otimes \bbW) \bbW_{(134)}\right] = [(F_1  F_4 F_2)\otimes F_3]\bbW,\\
%		\Tr_{34}\left[\bbW_{(234)}(F_1 \otimes F_2 \otimes F_3 \otimes F_4) \bbW_{(134)}\right] = (F_1 F_4) \otimes (F_3 F_2).
%	\end{gather}

	\subsection{Variance of \texorpdfstring{$\Tr(\rho^2)$}{empty} by classical shadow}
	
	The same computation was done in \cite{elben2020mixed}. Here we write down the computation process with our notation for completeness. Label the state estimator of each run as $\hr_1, \hr_2, \ldots, \hr_M$ separately, where $M$ is the total number of samples. Then the estimator for the purity is (recall that $p_n \equiv \Tr(\rho^n)$):
	\begin{equation}
		\hat{p}_2 = \binom{M}{2}^{-1}\sum_{j<k}\Tr(\hr_j \hr_k).
	\end{equation}
	Its variance $\text{Var}\left[\hat{p}_2\right]$ can be computed by
	\begin{align}
		\text{Var}\left[\hat{p}_2\right] &= \bbE\left[\hat{p}_2^2\right] - p_2^2 = \binom{M}{2}^{-2}\sum_{j_1<k_1,j_2<k_2}\bbE\left[\Tr(\hr_{j_1}\hr_{k_1})\Tr( \hr_{j_2}\hr_{k_2}) - p_2^2\right] \nonumber\\
		&= \binom{M}{2}^{-1}\left\{2(M-2)\left[\Tr\{(\rho\otimes\rho) \bbE[\hr\otimes\hr]\} - p_2^2\right] + \left[\Tr(\bbE[\hr\otimes\hr]^2) - p_2^2\right]\right\} \nonumber\\
		&\le \frac{4(M-2)}{M(M-1)}\Tr\{(\rho\otimes\rho) \bbE[\hr\otimes\hr]\} + \frac{2}{M(M-1)}\Tr(\bbE[\hr\otimes\hr]^2).
	\end{align}
	Note that
	\begin{align}
		\Tr(R_\rho\bbW R_\rho\bbW) = \Tr(R^2_\rho) = D^2 + 4D + 2D p_2 + 2,\quad 
		\Tr(R^2_\rho\bbW) = D + 4 + 4 p_2.
	\end{align}
	In conjunction with \eqsref{eq:F7}{eq:F8} in \lref{lem:classicalshadow}, we obtain
	\begin{align}
		\Tr\{(\rho\otimes\rho)\bbE[\hr\otimes\hr]\} &= \frac{D+1}{D+2}(p_2 + 2p_3) - \frac{1}{D+2}(1 + 2p_2)\le p_2 + 2p_3,\\
		\Tr(\bbE[\hr\otimes\hr]^2) &=  (D+1)^2 - \frac{D^2 + 2(D+1)(D+4)}{(D+2)^2} + 2\frac{D^2-2}{D+2}p_2\le (D+1)^2 + 2D p_2.
	\end{align}
	Hence,
	\begin{equation}
		\text{Var}\left[\hat{p}_2\right] \le \frac{4(M-2)}{M(M-1)}(p_2 + 2p_3) + \frac{2}{M(M-1)}[(D+1)^2 + 2D p_2].
	\end{equation}
	To guarantee that $\text{Var}\left[\hat{p}_2\right] = \caO(1)$, we need $M = \caO(D)$. 
	
	\subsection{Variance of \texorpdfstring{$\Tr(\rho^3)$}{empty} by classical shadow} The third moment $p_3$ has estimator
	\begin{equation}
		\hat{p}_3 = \frac{1}{2}\binom{M}{3}^{-1}\sum_{i<j<k}\Tr\left(\hr_i\hr_j \hr_k\right) + \frac{1}{2}\binom{M}{3}^{-1}\sum_{i<j<k}\Tr\left(\hr_j\hr_i \hr_k\right).
	\end{equation}
	Its variance is
	\begin{align}\label{eq:varp3}
		\text{Var}\left[\hat{p}_3\right] &= \frac{1}{4}\binom{M}{3}^{-2}\sum_{i_1<j_1<k_1,i_2<j_2<k_2}\bbE[\Tr(\hr_{i_1}\hr_{j_1}\hr_{k_1})\Tr(\hr_{i_2}\hr_{j_2}\hr_{k_2})- p_3^2] \nonumber\\
		&\equad + \frac{1}{4}\binom{M}{3}^{-2}\sum_{i_1<j_1<k_1,i_2<j_2<k_2}\bbE[\Tr(\hr_{j_1}\hr_{i_1}\hr_{k_1})\Tr(\hr_{j_2}\hr_{i_2}\hr_{k_2})- p_3^2] \nonumber\\
		&\equad  + \frac{1}{2}\binom{M}{3}^{-2}\sum_{i_1<j_1<k_1,i_2<j_2<k_2}\bbE[\Tr(\hr_{i_1}\hr_{j_1}\hr_{k_1})\Tr(\hr_{j_2}\hr_{i_2}\hr_{k_2})- p_3^2]. 
	\end{align}
	Now we analyze the expectation value of each summand with different indices. Based on how many pairs of $\hr$ sharing the same index, we can classify the summand into four categories.
	\begin{enumerate}
		\item All indices are distinct, then the summand equals
		\begin{equation}\bbE[\Tr(\hr_{1}\hr_{2}\hr_{3})\Tr(\hr_{4}\hr_{5}\hr_{6})] - p_3^2 = 0.
		\end{equation}

		\item There exists exactly a pair of identical indices, then in conjunction with \eref{eq:F7}, the summand equals
		\begin{align}
			&\equad \bbE[\Tr(\hr_{1}\hr_{2}\hr_{3})\Tr(\hr_{1}\hr_{4}\hr_{5})] = \Tr\{\bbE[\hr\otimes\hr](\rho^2\otimes\rho^2)\}\nonumber\\ 
			&= \frac{D+1}{D+2}\Tr[R_\rho \bbW (\rho^2\otimes\rho^2)] - \frac{1}{D+2}\Tr[R_\rho (\rho^2\otimes\rho^2)]\nonumber\\
			&= \frac{D+1}{D+2}\left(p_4 + 2p_5\right) - \frac{1}{D+2}\left(p_3^2 + 2p_2p_3\right) = \caO(1).
		\end{align}
		
		\item There exist exactly two pairs of identical indices, then the summand has the following three types.
		
		Using \eref{eq:F8}, we have
		\begin{align}
			&\equad \bbE[\Tr(\hr_{1}\hr_{2}\hr_{3})\Tr(\hr_{1}\hr_{2}\hr_{4})] = \Tr\{\bbE[\hr\otimes\hr]^2(\rho\otimes\rho)\} \nonumber \\
			&= \frac{(D+1)^2}{(D+2)^2}\Tr[R_\rho \bbW R_\rho \bbW (\rho\otimes\rho)] + \frac{1}{(D+2)^2}\Tr[R_\rho^2 (\rho\otimes\rho)]\nonumber\\
			&\equad - \frac{D+1}{(D+2)^2}\Tr[R_\rho \bbW R_\rho (\rho\otimes\rho)] - \frac{D+1}{(D+2)^2}\Tr[R_\rho^2 \bbW (\rho\otimes\rho)] \nonumber \\
			&= \frac{(D+1)^2}{(D+2)^2}(1 + 4p_2 + 2p_2^2 + 2p_3) + \frac{(1 + 2p_2)^2}{(D+2)^2}  - \frac{2D+2}{(D+2)^2}(p_2 + 4p_3 + 4p_4) = \caO(1).
		\end{align}
		
		Using \eref{eq:F7}, we have
		\begin{align}
			&\equad \bbE[\Tr(\hr_{1}\hr_{2}\hr_{3})\Tr(\hr_{1}\hr_{4}\hr_{2})] = \Tr\{\bbE[\hr\otimes\hr](I\otimes\rho)\bbE[\hr\otimes\hr](\rho\otimes I)\} \nonumber \\
			&= \frac{(D+1)^2}{(D+2)^2}\Tr[R_\rho \bbW (I\otimes \rho)R_\rho \bbW (\rho\otimes I)] + \frac{1}{(D+2)^2}\Tr[R_\rho (I\otimes \rho) R_\rho (\rho\otimes I)]\nonumber \\
			&\equad -\frac{D+1}{(D+2)^2}\Tr[R_\rho \bbW (I\otimes \rho)R_\rho (\rho\otimes I)] -\frac{D+1}{(D+2)^2}\Tr[R_\rho (I\otimes \rho)R_\rho \bbW (\rho\otimes I)]\nonumber \\
			&= \frac{(D+1)^2}{(D+2)^2}[D(p_2 + 2p_3 + p_4) + 2p_3 + 2p_2 + p_2^2] + \frac{(1 + 2p_2)^2}{(D+2)^2} - \frac{2D+2}{(D+2)^2}(p_2 + 4p_3 + 4p_4)= \caO(D).
		\end{align}
		
		Using \eref{eq:F9}, we have
		\begin{align}
			&\equad \bbE[\Tr(\hr_{1}\hr_{2}\hr_{3})\Tr(\hr_{2}\hr_{1}\hr_{4})] = \Tr\{\bbE[(\hr\otimes\hr')(\hr'\otimes\hr)] (\rho\otimes \rho)\}\nonumber\\
			&= \frac{(D+1)^3}{(D+2)^2}\Tr[R_\rho \bbW (\rho\otimes\rho)] + \frac{1}{(D+2)^2}\Tr[R_\rho \bbW R_\rho \bbW(\rho\otimes\rho)] - \frac{2(D+1)}{(D+2)^2}\Tr[R_\rho^2\bbW(\rho\otimes\rho)]\nonumber\\
			& \equad \frac{D(D+1)^2}{(D+2)^2}\Tr[(\rho \otimes \rho)\bbW(\rho \otimes \rho)] + \frac{(D+1)^2}{(D+2)^2}(1+p_2)\Tr[\bbW(\rho \otimes \rho)]\nonumber\\
			&= \frac{(D+1)^3}{(D+2)^2}(p_2 + 4p_3 + 4p_4) + \frac{D(D+1)^2}{(D+2)^2}p_4 + \frac{(D+1)^2}{(D+2)^2}(1+p_2)p_2 \nonumber\\
			&\equad + \frac{1}{(D+2)^2}\left(1 + 4p_2 + 2p_2^2 + 2p_3\right)  - \frac{2(D+1)}{(D+2)^2}(p_2 + 4p_3 + 4p_4) = \caO(D).
		\end{align}
		
		\item There exist three pairs of identical indices, then the summand has the following two types.
		
		Using \eref{eq:F7}, we have
		\begin{align}
			&\equad \bbE[\Tr(\hr_{1}\hr_{2}\hr_{3})\Tr(\hr_{1}\hr_{2}\hr_{3})] = \Tr(\bbE[\hr\otimes\hr]^3)\nonumber\\
			&= \frac{(D+1)^3}{(D+2)^3}\Tr(R_\rho \bbW R_\rho \bbW R_\rho \bbW) - \frac{3(D+1)^2}{(D+2)^3}\Tr(R_\rho^2 \bbW R_\rho \bbW) + \frac{3(D+1)}{(D+2)^3}\Tr(R_\rho^3 \bbW) - \frac{1}{(D+2)^3}\Tr(R_\rho^3)\nonumber\\
			&= \frac{(D+1)^3 + 3(D+1)}{(D+2)^3}(D + 6 + 12p_2 + 8p_3) \nonumber\\
			&\quad- \frac{1 + 3(D+1)^2}{(D+2)^3}\left[D^2 + 2D(3 + 3p_2 + p_3) + 6(1 + p_2)\right] = \caO(D).
		\end{align}
		
		Using \eqsref{eq:F7}{eq:F9}, we have
		\begin{align}
			&\equad \bbE[\Tr(\hr_{1}\hr_{2}\hr_{3})\Tr(\hr_{2}\hr_{1}\hr_{3})]  = \Tr(\bbE[(\hr\otimes\hr')(\hr'\otimes\hr)]\bbE[\hr\otimes \hr])\nonumber\\
			&= \frac{(D+1)^4}{(D+2)^3}\Tr(R_\rho^2) + \frac{D(D+1)^3}{(D+2)^2}\Tr[(\rho\otimes\rho)R_\rho] + \frac{(D+1)^3}{(D+2)^2}(1+p_2)\Tr(R_\rho) \nonumber \\
			&\equad + \frac{3(D+1)}{(D+2)^3}\Tr(R_\rho^3 \bbW) - \frac{2(D+1)^2 + 1}{(D+2)^3}\Tr(R_\rho^3) - \frac{(D+1)^3}{(D+2)^3}\Tr(R_\rho^2 \bbW) \nonumber\\
			&\equad - \frac{D(D+1)^2}{(D+2)^3}\Tr[R_\rho(\rho\otimes \rho)\bbW] - \frac{(D+1)^2}{(D+2)^3}(1 + p_2)\Tr(R_\rho \bbW)\nonumber\\
			&= \frac{(D+1)^4}{(D+2)^3}\left(D^2 + 4D + 2Dp_2 + 2\right) + \frac{D(D+1)^3}{(D+2)^3}(1 + 2p_2) \nonumber\\
			&\equad + \frac{(D+1)^3}{(D+2)^2}(1 + p_2)(D^2 + 2Dp_2) + \frac{3(D+1)}{(D+2)^3}\left(D + 6 + 12p_2 + 8p_3\right)\nonumber\\
			&\equad  - \frac{2(D+1)^2 + 1}{(D+2)^3}\left[D^2 + 2D(3 + 3 p_2 + p_3) + 6(1 + p_2)\right]  - \frac{(D+1)^2}{(D+2)^3}(D + 4 + 4p_2) \nonumber \\
			&\equad - \frac{D(D+1)^2}{(D+2)^3}(p_2 + 2p_3) - \frac{(D+1)^2}{(D+2)^3}(1 + p_2)(D + 2) = \caO(D^3).
		\end{align}
		
	\end{enumerate}
	One can verify that in \eref{eq:varp3}, there are $\caO(M^6),\caO(M^5),\caO(M^4),\caO(M^3)$ summands of types 1,2,3,4 respectively. Hence,
	\begin{equation}
		\mathrm{Var}[\hat{p}_3] = \caO(M^{-1}) + \caO(M^{-2}D) + \caO(M^{-3}D^3).
	\end{equation}
	To guarantee $\mathrm{Var}[\hat{p}_3] = \caO(1)$, we need $M = \caO(D)$.
	
	\subsection{Variance of \texorpdfstring{$\Tr[\rho_{\rmA}\Tr_{\rmB}(\rho^2)]$}{empty} by classical shadow} 
	
	We estimate $t_2 \equiv \Tr[\rho_{\rmA}\Tr_{\rmB}(\rho^2)]$ using the classical shadow method in the following steps. Use $M$ samples to generate $M$ state estimator for $\rho$ labeled as $\hr_1, \hr_2, \ldots, \hr_M$. Then use $L$ samples to generate $L$ state estimators for $\rho_{\rmA}$ labeled as $\ha_1, \ha_2, \ldots, \ha_L$. The total sample cost is $L+M$. 
	
	Let $\caH$ denote a local Hilbert space with $\dim(\caH) = \sqrt{D}$ such that $\caH_\rmA, \caH_\rmB \cong \caH$. Introduce a permutation $\pi^* \equiv (135)(24)$, and let $\bbW_{\pi^*}$ be the corresponding permutation operator on $\caH^{\otimes 5}$. Then we have
	\begin{equation}
		\Tr[(\rho\otimes\rho\otimes\rho_{\rmA})\bbW_{\pi^*}] = \Tr[\rho_{\rmA}\Tr_{\rmB}(\rho^2)] = t_2.
	\end{equation}
	Therefore, an estimator of $t_2$ can be constructed as follows:
	\begin{equation}\label{eq:vart2}
		\hatt_2 = \frac{1}{LM(M-1)}\sum_{\ell=0}^{L-1}\sum_{i=0,j\neq i}^{M-1}\Tr\left[(\hr_i \otimes \hr_j \otimes \ha_\ell)\bbW_{\pi^*} \right] = \frac{1}{LM(M-1)}\sum_{\ell=0}^{L-1}\sum_{i=0,j\neq i}^{M-1}\Tr\left[(\hat{a}_\ell\otimes I_{\rmB}) \hr_i \hr_j\right].
	\end{equation}

	The variance of $\hatt_2$ is
	\begin{align}
		\text{Var}\left[\hatt_2\right] = \frac{1}{L^2M^2(M-1)^2}\sum_{\ell_1,\ell_2}^{L-1}\sum_{i_1 \neq j_1, i_2 \neq j_2}^{M-1}\bbE\left\{\Tr[(\hat{a}_{\ell_1}\otimes I) \hr_{i_1} \hr_{j_1}]\Tr[(\hat{a}_{\ell_2}\otimes I_{\rmB}) \hr_{i_2} \hr_{j_2}]- t_2^2\right\}.
	\end{align}
	Note that
	\begin{align}
		\sum_{\ell_1, \ell_2} \bbE\left[\hat{a}_{\ell_1}\otimes I_{\rmB}\otimes \hat{a}_{\ell_2}\otimes I_{\rmB}\right] &= L(L-1)\rho_{\rmA}\otimes I_{\rmB}\otimes\rho_{\rmA}\otimes I_{\rmB} + L\bbW_{(23)}\left(\bbE[\hat{a}\otimes\hat{a}]\otimes I_{\rmB}\otimes I_{\rmB}\right)\bbW_{(23)},\\
		\sum_{i_1\neq j_1,i_2\neq j_2}\bbE[\hr_{i_1}\hr_{j_1}\otimes \hr_{i_2}\hr_{j_2}] &= \caO(M^4)\rho^2\otimes\rho^2 + \caO(M^2)\left\{\bbE[\hr\otimes\hr]^2 + \bbE[(\hr\otimes\hr')(\hr'\otimes\hr)]\right\}\nonumber\\
		&\equad + \caO(M^3)\{\bbE[\hr\otimes\hr](\rho\otimes\rho) + (\rho\otimes\rho)\bbE[\hr\otimes\hr]\}\nonumber\\
		&\equad + \caO(M^3)\left\{(I\otimes\rho)\bbE[\hr\otimes\hr](\rho\otimes I) + (\rho\otimes I)\bbE[\hr\otimes\hr](I\otimes \rho)\right\}.
	\end{align}
	In the expansion of \eref{eq:vart2}, we have the following terms (recall that $a_n \equiv \Tr(\rho_{\rmA}^n)$).
	
	Using \eref{eq:F8}, we have
	\begin{align}
		&\equad \Tr\{(\rho_{\rmA}\otimes I_{\rmB}\otimes\rho_{\rmA}\otimes I_{\rmB})\bbE[\hr\otimes\hr]^2\} \nonumber \\
		&= \frac{(D+1)^2 + 1}{(D+2)^2}\Tr[(\rho_{\rmA}\otimes I_{\rmB}\otimes \rho_{\rmA}\otimes I_{\rmB})R_\rho^2] - \frac{2(D+1)}{(D+2)^2}\Tr[(\rho_{\rmA}\otimes I_{\rmB}\otimes \rho_{\rmA}\otimes I_{\rmB})R_\rho^2\bbW_{(13)(24)}] \nonumber \\
		&= \frac{(D+1)^2 + 1}{(D+2)^2}\left[\left(\sqrt{D} + 2a_2 + t_2\right)^2 + t_2^2 - 2(a_2 + t_2)^2\right] \nonumber \\
		&\equad - \frac{2(D+1)}{(D+2)^2}\left\{\sqrt{D}a_2 + 4a_3 + 2\Tr\left[\rho_{\rmA}^2\Tr_{\rmB}(\rho^2)\right] + 2\Tr\left[(\rho_{\rmA}\otimes I_{\rmB})\rho(\rho_{\rmA}\otimes I_{\rmB})\rho\right]\right\} = \caO(D).
	\end{align}
	
	Using \eref{eq:F7}, we have
	\begin{align}
		&\equad \Tr\{(\rho_{\rmA}\otimes I_{\rmB}\otimes\rho_{\rmA}\otimes I_{\rmB})(\rho\otimes\rho)\bbE[\hr\otimes\hr]\} \nonumber \\
		&= \frac{D+1}{D+2}\Tr[(\rho_{\rmA}\otimes I_{\rmB}\otimes\rho_{\rmA}\otimes I_{\rmB})(\rho\otimes\rho)R_\rho \bbW_{(13)(24)}]  - \frac{1}{D+2}\Tr[(\rho_{\rmA}\otimes I_{\rmB}\otimes\rho_{\rmA}\otimes I_{\rmB})(\rho\otimes\rho)R_\rho]\nonumber\\
		&= \frac{D+1}{D+2}\left\{\Tr[(\rho_{\rmA}\otimes I_{\rmB})\rho(\rho_{\rmA}\otimes I_{\rmB})\rho] + 2\Tr[(\rho_{\rmA}\otimes I_{\rmB})\rho(\rho_{\rmA}\otimes I_{\rmB})\rho^2]\right\} -\frac{1}{D+2}(a_2^2 + 2a_2t_2) = \caO(1).
	\end{align}
	
	Using \eref{eq:F9}, we have
	\begin{align}
		&\equad \Tr\{(\rho_{\rmA}\otimes I_{\rmB}\otimes\rho_{\rmA}\otimes I_{\rmB}) \bbE[(\hr\otimes\hr')(\hr'\otimes\hr)]\}\nonumber\\
		&= \frac{(D+1)^3}{(D+2)^2}\Tr\left[(\rho_{\rmA}\otimes I_{\rmB}\otimes\rho_{\rmA}\otimes I_{\rmB})R_\rho \bbW_{(13)(24)}\right] + \frac{D(D+1)^2}{(D+2)^2}\Tr[(\rho_\rmA \otimes I_\rmB \otimes \rho_\rmA \otimes I_\rmB)(\rho \otimes \rho)\bbW_{(13)(24)}] \nonumber \\
		&\equad + \frac{(D+1)^2}{(D+2)^2}(1+p_2)\Tr[(\rho_\rmA \otimes I_\rmB \otimes \rho_\rmA \otimes I_\rmB)\bbW_{(13)(24)}] \nonumber \\
		&\equad + \frac{1}{(D+2)^2}\Tr[(\rho_{\rmA}\otimes I_{\rmB}\otimes\rho_{\rmA}\otimes I_{\rmB})R_\rho^2]  - \frac{2(D+1)}{(D+2)^2}\Tr[(\rho_{\rmA}\otimes I_{\rmB}\otimes\rho_{\rmA}\otimes I_{\rmB})R_\rho^2\bbW_{(13)(24)}]\nonumber \\
		&= \frac{(D+1)^3}{(D+2)^2}\left(\sqrt{D}a_2 + 2a_3\right) + \frac{D(D+1)^2}{(D+2)^2}\Tr\{[(\rho_\rmA \otimes I_\rmB)\rho]^2\} + \frac{(D+1)^2}{(D+2)^2}(1+p_2)\sqrt{D}a_2\nonumber \\
		&\equad +\frac{1}{(D+2)^2}\left[\left(\sqrt{D} + 2a_2 + t_2\right)^2 + t_2^2 - 2(a_2 + t_2)^2\right]\nonumber\\
		&\equad - \frac{2(D+1)}{(D+2)^2}\left\{\sqrt{D}a_2 + 4a_3 + 2\Tr[\rho_{\rmA}^2\Tr_{\rmB}(\rho^2)] + 2\Tr[(\rho_{\rmA}\otimes I_{\rmB})\rho(\rho_{\rmA}\otimes I_{\rmB})\rho]\right\} = \caO(D^{\frac{3}{2}});
	\end{align}
	
	For notational simplicity, we introduce the operator
	$ R_\rmA^{(12)} \equiv R_{\rho_{\rmA}} \otimes I_{\rmB} \otimes I_{\rmB}. $
	Note that $\mathbb{E}[\hat{a} \otimes \hat{a}]$ takes a form very similar to that of $\mathbb{E}[\hat{\rho} \otimes \hat{\rho}]$, with the only difference being the Hilbert-space dimensions involved. 
	
	Using \eref{eq:F7}, we obtain
	\begin{align}
		&\equad \Tr\{\bbW_{(23)}(\bbE[\hat{a}\otimes\hat{a}]\otimes I_{\rmB}\otimes I_{\rmB})\bbW_{(23)}(\rho^2\otimes\rho^2)\} \nonumber\\
		&= \frac{\sqrt{D}+1}{\sqrt{D}+2}\Tr\left[\bbW_{(23)}(R_{\rho_{\rmA}}\bbW_{(12)}\otimes I\otimes I)\bbW_{(23)}(\rho^2\otimes \rho^2)\right]  - \frac{1}{\sqrt{D}+2}\Tr\left[\bbW_{(23)}R_{\rho_{\rmA}}^{(12)}\bbW_{(23)}(\rho^2\otimes \rho^2)\right]\nonumber\\
		&= \frac{\sqrt{D}+1}{\sqrt{D}+2}\left\{\Tr_{\rmA}[\Tr_{\rmB}(\rho^2)\Tr_{\rmB}(\rho^2)] + 2\Tr_{\rmA}\left[\rho_{\rmA}\Tr_{\rmB}(\rho^2)\Tr_{\rmB}(\rho^2)\right]\right\} -\frac{1}{\sqrt{D}+2}(p_2^2 + 2p_2t_2) = \caO(1).
	\end{align}
	
	Using \eqsref{eq:F7}{eq:F8}, we have
	\begin{align}
		&\equad \Tr\left\{\bbW_{(23)}(\bbE[\hat{a}\otimes\hat{a}]\otimes I_{\rmB}\otimes I_{\rmB})\bbW_{(23)}\bbE[\hr\otimes\hr]^2\right\}\nonumber\\
		&= \frac{(D+1)^2 + 1}{(D+2)^2}\Tr[\bbW_{(23)}(\bbE[\hat{a}\otimes\hat{a}]\otimes I_{\rmB}\otimes I_{\rmB})\bbW_{(23)} R_\rho^2] \nonumber\\
		&\equad - \frac{2(D+1)}{(D+2)^2}\Tr[\bbW_{(23)}(\bbE[\hat{a}\otimes\hat{a}]\otimes I_{\rmB}\otimes I_{\rmB})\bbW_{(23)} R_\rho^2\bbW_{(13)(24)}] \nonumber\\
		&= \frac{(D+1)^2+1}{(D+2)^2}\frac{\sqrt{D}+1}{\sqrt{D}+2}\Tr[\bbW_{(23)}R_{\rho_{\rmA}}^{(12)}\bbW_{(12)}\bbW_{(23)}R_{\rho}^2]\nonumber\\
		&\equad  + \frac{2(D+1)}{(D+2)^2}\frac{1}{\sqrt{D}+2}\Tr[\bbW_{(23)}R_{\rho_{\rmA}}^{(12)}\bbW_{(23)}R_\rho^2\bbW_{(13)(24)}]\nonumber\\
		&\equad - \frac{(D+1)^2 + 1}{(D+2)^2}\frac{1}{\sqrt{D}+1}\Tr[\bbW_{(23)}R_{\rho_{\rmA}}^{(12)}\bbW_{(23)}R_\rho^2]\nonumber\\
		&\equad - \frac{2(D+1)}{(D+2)^2}\frac{\sqrt{D}+1}{\sqrt{D}+2}\Tr[\bbW_{(23)}R_{\rho_{\rmA}}^{(12)}\bbW_{(12)}\bbW_{(23)}R^2_{\rho}\bbW_{(13)(24)}]\nonumber\\
		&=\frac{(D+1)^2+1}{(D+2)^2}\frac{\sqrt{D}+1}{\sqrt{D}+2}\left(D^{\frac{3}{2}} + 2D + 2\sqrt{D} + 4\sqrt{D}a_2\right)\nonumber\\
		&\equad + \frac{2(D+1)}{(D+2)^2}\frac{1}{\sqrt{D}+2}\left(D + 2\sqrt{D} + 4 + 4p_2 + 8a_2 + 8t_2\right)\nonumber\\
		&\equad - \frac{(D+1)^2 + 1}{(D+2)^2}\frac{1}{\sqrt{D}+1}\left(D^2 + 4D + 4\sqrt{D}a_2 + 2D t_2 + 4a_2 + 2\sqrt{D}p_2\right)\nonumber\\
		&\equad - \frac{2(D+1)}{(D+2)^2}\frac{\sqrt{D}+1}{\sqrt{D}+2}\left[D^{\frac{3}{2}} + 2D(2 + a_2 + t_2) + 2\sqrt{D}(2 + p_2) + 4a_2 + 2p_2\right] = \caO\left(D^{\frac{3}{2}}\right).
	\end{align}
	
	Using \eref{eq:F7}, we have
	\begin{align}
		&\equad \Tr\{\bbW_{(23)}(\bbE[\hat{a}\otimes\hat{a}]\otimes I_{\rmB}\otimes I_{\rmB})\bbW_{(23)}(\rho\otimes\rho)\bbE[\hr\otimes\hr]\} \nonumber\\
		&= \frac{D+1}{D+2}\frac{\sqrt{D}+1}{\sqrt{D}+2}\Tr[\bbW_{(23)}R_{\rho_{\rmA}}^{(12)}\bbW_{(12)}\bbW_{(23)}(\rho\otimes\rho)R_{\rho}\bbW_{(13)(24)}]\nonumber\\
		&\equad  + \frac{1}{D+2}\frac{1}{\sqrt{D}+2}\Tr[\bbW_{(23)}R_{\rho_{\rmA}}^{(12)}\bbW_{(23)}(\rho\otimes\rho)R_\rho]\nonumber\\
		&\equad - \frac{D+1}{D+2}\frac{1}{\sqrt{D}+2}\Tr[\bbW_{(23)}R_{\rho_{\rmA}}^{(12)}\bbW_{(23)}(\rho\otimes\rho)R_{\rho}\bbW_{(13)(24)}]\nonumber\\
		&\equad - \frac{1}{D+2}\frac{\sqrt{D}+1}{\sqrt{D}+2}\Tr[\bbW_{(23)}R_{\rho_{\rmA}}^{(12)}\bbW_{(12)}\bbW_{(23)}(\rho\otimes\rho)R_{\rho}]\nonumber\\
		&= \frac{D+1}{D+2}\frac{\sqrt{D+1}}{\sqrt{D}+2}\left\{a_2 + 2\Tr[(\rho_{\rmA}\otimes\rho_{\rmB})\rho] + 2t_2 + 2\Tr[(\rho_{\rmA}\otimes \Tr_{\rmB}(\rho^2))\rho] + 2\Tr[(\rho_{\rmA}\otimes\rho_{\rmB})\rho^2]\right\}\nonumber\\
		&\equad - \frac{D+1}{D+2}\frac{1}{\sqrt{D}+2}\left\{a_2 + 2\Tr[(\rho_{\rmA}\otimes \rho_{\rmB})\rho] + 2t_2 + 2\Tr[(\rho_{\rmA}\otimes \Tr_{\rmA}(\rho^2) + \rho_{\rmA}\otimes\rho_{\rmB})\rho]\right\}\nonumber\\
		&\equad + \frac{1}{D+2}\frac{1}{\sqrt{D}+2}\left(1 + 4p_2 + 2t_2 + 2p_2a_2\right) - \frac{1}{D+2}\frac{\sqrt{D}+1}{\sqrt{D}+2}\left\{a_2 + 2a_3 + 2t_2 + 4\Tr[\rho_{\rmA}^2\Tr_{\rmB}(\rho^2)]\right\} = \caO(1).
	\end{align}
	
	Using \eqsref{eq:F7}{eq:F9}, we have
	\begin{align}
		&\equad\: \Tr\{\bbW_{(23)}(\bbE[\hat{a}\otimes\hat{a}]\otimes I_{\rmB}\otimes I_{\rmB})\bbW_{(23)}\bbE[(\hr\otimes\hr')(\hr'\otimes\hr)]\}\nonumber\\
		&= \frac{(D+1)^3}{(D+2)^2}\frac{\sqrt{D}+1}{\sqrt{D}+2}\Tr\left[\bbW_{(23)}R_{\rho_{\rmA}}^{(12)}\bbW_{(12)}\bbW_{(23)}R_\rho \bbW_{(13)(24)}\right]\nonumber\\
		&\equad + \frac{D(D+1)^2}{(D+2)^2}\frac{\sqrt{D}+1}{\sqrt{D}+2}\Tr\left[\bbW_{(23)}R_{\rho_{\rmA}}^{(12)}\bbW_{(12)}\bbW_{(23)}(\rho \otimes \rho)\bbW_{(13)(24)}\right]\nonumber\\
		&\equad + \frac{(D+1)^2}{(D+2)^2}(1+p_2)\frac{\sqrt{D}+1}{\sqrt{D}+2}\Tr\left[\bbW_{(23)}R_{\rho_{\rmA}}^{(12)}\bbW_{(12)}\bbW_{(23)}\bbW_{(13)(24)}\right]\nonumber\\
		&\equad + \frac{1}{(D+2)^2}\frac{\sqrt{D}+1}{\sqrt{D}+2}\Tr\left[\bbW_{(23)}R_{\rho_{\rmA}}^{(12)}\bbW_{(12)}\bbW_{(23)}R_\rho^2\right]\nonumber\\
		&\equad + \frac{2(D+1)}{(D+2)^2}\frac{1}{\sqrt{D}+2}\Tr\left[\bbW_{(23)}R_{\rho_{\rmA}}^{(12)}\bbW_{(23)}R_\rho^2\bbW_{(13)(24)}\right]\nonumber\\
		&\equad - \frac{2(D+1)}{(D+2)^2}\frac{\sqrt{D}+1}{\sqrt{D}+2}\Tr\left[\bbW_{(23)}R_{\rho_{\rmA}}^{(12)}\bbW_{(12)}\bbW_{(23)}R_\rho^2\bbW_{(13)(24)}\right]\nonumber\\
		&\equad - \frac{(D+1)^3}{(D+2)^2}\frac{1}{\sqrt{D}+2}\Tr\left[\bbW_{(23)}R_{\rho_{\rmA}}^{(12)}\bbW_{(23)}R_\rho \bbW_{(13)(24)}\right]\nonumber\\
		&\equad - \frac{D(D+1)^2}{(D+2)^2}\frac{1}{\sqrt{D}+2}\Tr\left[\bbW_{(23)}R_{\rho_{\rmA}}^{(12)}\bbW_{(23)}(\rho \otimes \rho) \bbW_{(13)(24)}\right]\nonumber\\
		&\equad - \frac{(D+1)^2}{(D+2)^2}(1+p_2)\frac{1}{\sqrt{D}+2}\Tr\left[\bbW_{(23)}R_{\rho_{\rmA}}^{(12)}\bbW_{(23)} \bbW_{(13)(24)}\right]\nonumber\\
		&\equad - \frac{1}{(D+2)^2}\frac{1}{\sqrt{D}+2}\Tr\left[\bbW_{(23)}R_{\rho_{\rmA}}^{(12)}\bbW_{(23)}R_\rho^2\right].
	\end{align}
	After computation, we prove that it equals
	\begin{align}
		&\equad\: \frac{(D+1)^3}{(D+2)^2}\frac{\sqrt{D}+1}{\sqrt{D}+2}\left(D^{\frac{3}{2}} + 2D + 2\sqrt{D} + 2\sqrt{D}a_2 + 2\right) \nonumber \\
		&\equad + \frac{D(D+1)^2}{(D+2)^2}\frac{\sqrt{D}+1}{\sqrt{D}+2}\left\{a_2 + \Tr[(\rho_\rmA \otimes \rho_\rmB)\rho] + \Tr(\rho_\rmA \rho_\rmB^2)\right\}\nonumber\\
		&\equad + \frac{(D+1)^2}{(D+2)^2}(1+p_2)\frac{\sqrt{D}+1}{\sqrt{D}+2}\left(D^{\frac{3}{2}} + 2D\right)  + \frac{1}{(D+2)^2}\frac{\sqrt{D}+1}{\sqrt{D}+2}\left(D^{\frac{3}{2}} + 2D + 2\sqrt{D} + 4\sqrt{D}a_2\right)\nonumber\\
		&\equad + \frac{2(D+1)}{(D+2)^2}\frac{1}{\sqrt{D}+2}\left(D + 2\sqrt{D} + 4 + 4p_2 + 8a_2 + 8t_2\right)\nonumber\\
		&\equad -\frac{2(D+1)}{(D+2)^2}\frac{\sqrt{D}+1}{\sqrt{D}+2}\left[D^{\frac{3}{2}} + 2D(2 + a_2 + t_2) + 2\sqrt{D}(2 + p_2) + 4a_2 + 2p_2\right]\nonumber\\
		&\equad -\frac{(D+1)^3}{(D+2)^2}\frac{1}{\sqrt{D}+2}\left[D^2 + 2D^{\frac{3}{2}} + D(4 + 2p_2 + 4a_2 + 2t_2)\right]\nonumber\\
		&\equad - \frac{(D+1)^3}{(D+2)^2}\frac{1}{\sqrt{D}+2}\{\Tr[(\rho \otimes \rho)\bbW_{(1243)}] + \Tr[(\rho_\rmA \otimes I_\rmB \otimes I)(\rho\otimes \rho)(\bbW_{(1243)} + \bbW_{(1342)})]\}\nonumber \\
		&\equad - \frac{(D+1)^2}{(D+2)^2}(1+p_2)\frac{1}{\sqrt{D}+2}[D + 2\sqrt{D}] -\frac{(D+1)^3}{(D+2)^2}\frac{1}{\sqrt{D}+2}\left[\sqrt{D}(4 + 2p_2) + 2 + 4a_2\right]\nonumber\\
		&\equad - \frac{1}{(D+2)^2}\frac{1}{\sqrt{D}+2}\left(D^2 + 4D + 4\sqrt{D}a_2 + 2D t_2 + 4a_2 + 2\sqrt{D}p_2\right) = \caO\left(D^{\frac{5}{2}}\right).
	\end{align}
	Therefore,
	\begin{align}
		\text{Var}\left[\hatt_2\right] &= \frac{1}{L^2M^4}\left[\caO\left(L^2 M^2D\right) + \caO\left(L^2 M^3\right) + \caO\left(L^2 M^2 D^{\frac{3}{2}}\right) + \caO\left(LM^4\right) + \caO\left(LM^2 D^{\frac{5}{2}}\right)\right]\nonumber\\
		&= \caO\left(M^{-2}D\right) + \caO\left(M^{-1}\right) + \caO\left(M^{-2}D^{\frac{3}{2}}\right) + \caO\left(L^{-1}\right) + \caO\left(L^{-1}M^{-2}D^{\frac{5}{2}}\right).
	\end{align}
	To guarantee that $\text{Var}\left[\hatt_2\right] = \caO(1)$, we need
	\begin{equation}
		M = \caO\left(D^{\frac{3}{4}}\right),\quad LM^2 = \caO\left(D^{\frac{5}{2}}\right). 
	\end{equation}
	If we choose $L = \caO(D)$, then the total sample complexity is $L + M = \caO(D)$.

\end{document}